\documentclass{article}

\usepackage[authoryear]{natbib}
\usepackage{url}
\usepackage{amstext,latexsym,amsbsy,amsfonts,amsmath,amssymb,amsthm}
\usepackage{algorithm}
\usepackage{algorithmic}
\usepackage{subfig}
\usepackage{booktabs}
\usepackage{xcolor}
\usepackage{bbm}
\usepackage{diagbox}
\usepackage{mathtools}
\usepackage{placeins}

\usepackage[pagewise]{lineno}
\usepackage{microtype}
\usepackage{booktabs}
\usepackage{amssymb,amsfonts, amsthm}
\usepackage[algo2e, ruled]{algorithm2e}

\usepackage{natbib}
\usepackage[
  colorlinks,
  citecolor=blue,
  linkcolor=red,
  anchorcolor=red,
  urlcolor=blue,
  backref=page
]{hyperref}

\usepackage{cleveref}
\crefname{equation}{eq.}{eqs.}
\Crefname{equation}{Eq.}{Eqs.}
\usepackage{longtable}

\makeatletter
\newcommand*{\addFileDependency}[1]{%
  \typeout{(#1)}
  \@addtofilelist{#1}
  \IfFileExists{#1}{}{\typeout{No file #1.}}
}
\makeatother

\usepackage{cleveref}

\DeclareMathOperator*{\argmin}{arg\,min}

\newcommand{\norm}[1]{\left\lVert#1\right\rVert}

\newcommand{\pde}[1]{\frac{\partial}{\partial #1}}
\usepackage{stackengine}

\makeatletter
\newcommand{\distas}[1]{\mathbin{\overset{#1}{\kern\z@\sim}}}%

\newtheorem{theorem}{Theorem}
\newtheorem{definition}{Definition}
\newtheorem{lemma}{Lemma}
\newtheorem{remark}{Remark}

\newtheorem{proposition}{Proposition}

\newcommand{\bas}[1]{\begin{align*}#1\end{align*}}
\newcommand{\ba}[1]{\begin{align}#1\end{align}}

\newcommand{\bE}{\mathbb{E}}
\newcommand{\bR}{\mathbb{R}}

\newcommand{\cA}{\mathcal{A}}
\newcommand{\cL}{\mathcal{L}}
\newcommand{\cO}{\mathcal{O}}

\newcommand{\cN}{\mathcal{N}}

\newcommand{\kl}{\text{KL}}

\newcommand{\beqs}{\vspace{0mm}\begin{eqnarray}}
\newcommand{\eeqs}{\vspace{0mm}\end{eqnarray}}
\newcommand{\barr}{\begin{array}}
\newcommand{\earr}{\end{array}}
\newcommand{\Amat}[0]{{{\bf A}}}

\newcommand{\Imat}{{\bf I}}

\newcommand{\Xmat}[0]{{{\bf X}}}

\newcommand{\gv}[0]{{\boldsymbol{g}} }

\newcommand{\uv}{\boldsymbol{u}}

\newcommand{\xv}{\boldsymbol{x}}
\newcommand{\yv}{\boldsymbol{y}}
\newcommand{\zv}{\boldsymbol{z}}
\newcommand{\cdotv}{\boldsymbol{\cdot}}
\newcommand{\expfunc}{\ensuremath{f}}

\newcommand{\Sigmamat}[0]{{\boldsymbol{\Sigma}}}

\newcommand{\alphav}{\boldsymbol{\alpha}}
\newcommand{\betav}[0]{{\boldsymbol{\beta}}}

\newcommand{\epsilonv}{\boldsymbol{\epsilon}}

\newcommand{\thetav}{\boldsymbol{\theta}}

\newcommand{\piv}{\boldsymbol{\pi}}

\newcommand{\phiv}{\boldsymbol{\phi}}

\newcommand{\E}{\mathbb{E}}

\newcommand{\cZ}{\mathcal{Z}}

\newcommand{\brackets}[1]{\left[ #1 \right]}
\newcommand{\parenth}[1]{\left( #1 \right)}

\newcommand{\given}{\,|\,}
\newcommand{\deltaf}{\Delta_{\zv,k} f}

\DeclareSymbolFont{symbolsC}{U}{txsyc}{m}{n}
\DeclareMathSymbol{\notniFromTxfonts}{\mathrel}{symbolsC}{61}

\definecolor{alizarin}{rgb}{0.82, 0.1, 0.26}

\definecolor{applegreen}{rgb}{0.55, 0.71, 0.0}

\newcommand{\var}{\text{var}}
\newcommand{\half}{\frac{1}{2}}

\newcommand{\hpless}{\mathbf{1}_{[u<\sigma(\phi)]}}
\newcommand{\hpgreater}{\mathbf{1}_{[u>1-\sigma(\phi)]}}
\newcommand{\hplessv}{\mathbf{1}_{[\uv<\sigma(\phiv)]}}
\newcommand{\hpgreaterv}{\mathbf{1}_{[\uv>1-\sigma(\phiv)]}}

\usepackage[utf8]{inputenc}
\DeclareUnicodeCharacter{1D4C1}{L}

\definecolor{ymz}{RGB}{191, 87, 0}
\definecolor{bowei}{RGB}{255,128,0}

\def\arrvline{\hfil\kern\arraycolsep\vline\kern-\arraycolsep\hfilneg}

\allowdisplaybreaks

\newcommand*\samethanks[1][\value{footnote}]{\footnotemark[#1]}

\def\spacingset#1{\renewcommand{\baselinestretch}%
{#1}\small\normalsize} \spacingset{1}

\renewenvironment{abstract}%
{%
  \vskip 0.075in%
  \centerline%
  {\large\bf Abstract}%
  \vspace{0.5ex}%
  \begin{quote}%
}
{
  \par%
  \end{quote}%
  \vskip 1ex%
}

\usepackage{float}

\usepackage{geometry}  %
\renewcommand{\baselinestretch}{1.0}

\title{Probabilistic Best Subset Selection via \\ Gradient-Based Optimization}

\author{
  Mingzhang Yin  \thanks{Columbia University}\\
   \texttt{mingzhang.yin@columbia.edu} \\
  \and
  Nhat Ho \thanks{The University of Texas at Austin}\\
   \texttt{minhnhat@utexas.edu} \\
  \and
 Bowei Yan\thanks{Independent Researcher}\\
  \texttt{yanbowei@gmail.com } \\
  \and
Xiaoning Qian \thanks{Texas A\&M University}\\
  \texttt{xqian@ece.tamu.edu } \\
  \and
  Mingyuan Zhou \samethanks[2]\\
  \texttt{mingyuan.zhou@mccombs.utexas.edu} 
}

\begin{document}

\maketitle

\begin{abstract}
In high-dimensional statistics, variable selection recovers the latent sparse patterns from all possible covariate combinations. This paper proposes a novel optimization method to solve the exact $L_0$-regularized regression problem, which is also known as the best subset selection. 
We reformulate the optimization problem from a discrete space to a continuous one via probabilistic reparameterization. The new objective function is differentiable but its gradient often cannot be computed in a closed form. Then we propose a family of unbiased gradient estimators to optimize the best subset selection objectives by the stochastic gradient descent. Within this family, we identify the estimator with uniformly minimum variance. Theoretically, we study the general conditions under which the method is guaranteed to converge to the ground truth in expectation. The proposed method can find the true regression model from thousands of covariates in seconds. In a wide variety of synthetic and semi-synthetic data, the proposed method outperforms existing variable selection tools based on the relaxed penalties, coordinate descent, and mixed integer optimization in both sparse pattern recovery and out-of-sample prediction. 
\end{abstract}

\section{Introduction}\label{sec:introduction}

Variable selection by regularized regression is widely applied to uncover sparse structures in high dimensional data. 
It is a natural approach to solving a regression problem with the constraints on the $L_0$-norm of the coefficients, as it directly regularizes the number of variables included in the regression model.  This is known as the best subset selection problem~\citep{friedman2001elements, fan2010selective}. 
In this paper, we study $L_0$-regularized regression with the following optimization objective function
\ba{
\min_{\betav \in \bR^p} \Big\{ \frac{1}{n}\norm{\yv - \Xmat\betav}^2_2+ \lambda \|\betav\|_0 \Big\}, 
\label{eq:bss_regular}
}
where $\yv=(y_1,\ldots,y_n)^{\top} \in \bR^{n}$ is a vector of observed response variables for n units, $\Xmat=(\xv_1,\ldots,\xv_n)^{\top} \in \bR^{n\times p}$ is the design matrix, and $\betav \in \bR^p$ are the regression coefficients. The $L_0$ regularization in \Cref{eq:bss_regular} is defined as $\norm{\betav}_0: =\sum_{j=1}^p \mathbf{1}_{[ \beta_j \neq 0 ]}$ where $\mathbf{1}_{[\cdotv]}$ is an indicator function that equals to one if the condition is true and zero otherwise. This regularization counts the number of nonzero elements.  We consider the high-dimensional regime, where the number of covariates $p$ exceeds the sample size $n$, and can potentially 
grow with~$n$. 

The best subset selection objective in \Cref{eq:bss_regular} is nonconvex and is discontinuous with respect to the coefficient $\betav$,  which is an NP-hard problem  to solve \citep{natarajan1995sparse}.  A widely used approach to improving the computational efficiency is to approximate the  $L_0$-norm with continuous regularizations. As a pioneer framework, the Bridge regression \citep{frank1993statistical, fu1998penalized} considers the $L_{q}$ penalties ($q>0$), defined as $\sum_{j=1}^p \beta_j^q$ for  the regression coefficient $ \beta_j$. When $q\geq 1$, the $L_{q}$ penalty is convex, while when $q\leq 1$, the regularization encourages sparse estimation and hence performs variable selection \citep{fan2010selective}. In the intersection of these two domains lies the widely used least absolute shrinkage and selection operator (\textsc{Lasso}). Asymptotically, \textsc{Lasso} is accurate for both sparse pattern recovery and coefficient estimation \citep{zhao2006model, candes2009near, wainwright2009sharp}. However, in the finite sample setting, it introduces downward bias due to the shrinkage effect of the $L_1$ norm and might select an excessively large subset based on the cross validation in practice \citep{bertsimas2016best,hastie2020best}. 

In this paper, we solve the best subset selection problem directly. Compared to continuous approximations, the solution of the exact best subset selection enjoys superior statistical properties, such as the unbiasedness of regression coefficients, known as the oracle property \citep{greenshtein2006best, zhang2012general, belloni2013least}, and the low in-sample risk \citep{foster1994risk}. For example, \citet{johnson2015risk} shows that the predictive risk of the $L_1$-regularized linear regression cannot outperform $L_0$-regularized regression by more than a small constant factor and in some cases is infinitely worse under assumptions of the design matrix. 

Thanks to the benefits of $L_0$ penalty and the fast improvements in modern computational tools, recently there is renewed interest in solving the exact best subset selection problem.  \citet{bertsimas2016best} 
solve the constrained best subset selection problem 
\ba{
\min_{\betav \in \bR^p} \Big\{ \frac{1}{n}\norm{\yv - \Xmat\betav}^2_2\Big\} \quad \text{subject to} \quad \norm{\betav}_0 \leq \hat{S}
\label{eq:bss_constrained} 
}
with a two-stage algorithm using mixed integer optimization (MIO). 
MIO methods scale up the  number of covariates $p$ for the best subset selection from 30s to 1000s \citep{furnival1974regressions,bertsimas2016best}. However, the MIO steps rely upon a nonconvex optimization tool such as Gurobi \citep{gurobi}, which is not straightforward to generalize
beyond linear
regression problems.  Moreover, solving \Cref{eq:bss_constrained} exactly with an MIO algorithm might be computationally expensive \citep{gomez2021mixed}.  The computational efficiency of the MIO method is  improved by searching a hierarchy of local minima of the objective function in \Cref{eq:bss_regular} with coordinate descent and local combinatorial search \citep{hazimeh2018fast}.   ~\looseness=-1

This paper proceeds in the direction of solving the best subset selection by providing a gradient-based solution. We first propose a probabilistic objective for the exact $L_0$-regularized regression by casting the discrete optimization problem to an equivalent one in the continuous space. The new objective function is differentiable, but the exact gradient is infeasible to compute in practice when the covariates are high dimensional. Then, we propose a family of unbiased estimators to approximate the infeasible gradient. Within this unbiased estimator family,  we identify the estimator with minimal variance and a non-vanishing signal-to-noise ratio. By construction, the gradient estimators are computationally efficient in high dimensions. The result is an end-to-end solver for the best subset selection with the modern stochastic gradient descent (SGD) as the main workhorse. We provide theoretical insights on the conditions that guarantee the convergence of the gradient descent updates to the ground truth in expectation. ~\looseness=-1 %

The proposed gradient method has several strengths. First, it scales to practical problems with high dimensional covariates. The exact gradient can be estimated accurately with a few random samples, and the algorithm searches for the best subset in the steepest descent direction of the objective function.  
Second, the gradient method has high flexibility. Aside from
 the regularized regression objective in \Cref{eq:bss_regular}, we further apply the gradient method to solve a new variational objective for the Bayesian linear regression with the spike-and-slab prior \citep{mitchell1988bayesian}. Empirically, compared to existing variable selection methods with continuous relaxation and $L_0$ penalties, we find the gradient-based method improves the sparsity pattern recovery, the coefficient estimation, and out-of-sample prediction across a wide range of problem settings. 
 
\paragraph{Related work.}   A major paradigm in variable selection is based on continuous penalties. In addition to the Bridge regression we have discussed, there are a variety of nonconvex penalties for sparsity learning. Smoothly clipped absolute deviation (SCAD) \citep{fan2001variable} and minimax concave penalty (MCP) \citep{zhang2010nearly}, for example, approximate the hard-thresholding property of $L_0$ regularizer by the piece-wise nonconvex functions. Some work construct penalties to directly resemble the functional form of the $L_0$ penalty, known as the pseudo-$L_0$ penalties \citep{liu2007variable,shen2012likelihood,dicker2013variable}.  Since there is a rich literature on variable selection, we refer  readers to several representative publications \citep{friedman2001elements, fan2010selective, bertsimas2016best, hastie2020best} and the references therein for comprehensive reviews.

The $L_0$-regularized subset selection is closely related to the standard information theoretic methods for model selection. In particular, 
when the data is Gaussian distributed,  Eq.~\eqref{eq:bss_regular} is equivalent to the Akaike information criterion (AIC) \citep{akaike1974new, akaike1998information} for $\lambda = 1/n$, and is equivalent to  the Bayesian information criterion (BIC) \citep{schwarz1978estimating} for $\lambda = \log(n)/(2n)$. The objective function in \Cref{eq:bss_regular}  incorporates the information criterions directly into the optimization, hence performing model comparison among  a wide range of candidate models.  ~\looseness=-1

Last, the proposed gradient method is related to the growing field of discrete optimization in machine learning \citep{mohamed2020monte}. These works often consider a latent variable model with Bernoulli or Categorical latent variables, such as the actions in policy learning and the activation layers in belief networks \citep{gu2015muprop,CARSM2020}. A variety of variance reduction techniques  lie at the core of these discrete optimization methods, such as Gumbel-softmax relaxation \citep{maddison2016concrete,jang2016categorical},  control variates \citep{tucker2017rebar}, antithetic sampling \citep{yin2018arm}, and Gaussian-based continuous relaxation \citep{yamada2020feature}.  In spite of empirical success in optimizing neural networks in the areas such as deep generative models and reinforcement learning, the statistical properties of these methods are largely unknown. By studying the best subset selection problem with gradient methods, we systematically analyze the bias-variance trade-off and the convergence properties, and further propose a new discrete optimization tool. ~\looseness=-1

\paragraph{Organization.} The remainder of this paper is organized as follows. In Section~\ref{Sec:l0_reformulation}, we provide a probabilistic reformulation of the $L_0$-regularized linear regression problem. In
Section~\ref{sec:estimators}, we propose a family of unbiased gradient estimators  to solve the reformulated optimization problem. In Section~\ref{Sec:convergence}, we analyze the conditions that guarantee the convergence to the ground truth in expectation. In Section~\ref{sec:vb}, we generalize the proposed gradient methods to optimize a novel variational lower bound for the Bayesian best subset selection. Experiments and results are described in Section~\ref{sec:experiemnt}, and discussions follow in Section~\ref{sec:discussion}.

\paragraph{Notation.} We use $n$ as the sample size, $p$ as the number of covariates, and $S$ as the number of non-zero true coefficients. We use $\xv_i$ to denote the $i^{th}$ row of matrix $\Xmat$ and $X_j$ as its $j^{th}$ column. We use $\Xmat_{\mathcal{A}}$ to denote a submatrix of $\Xmat$ as $\Xmat_{\mathcal{A}} = \{X_j \}_{j\in \mathcal{A}}$, $\mathcal{A}\subseteq [p]$, where $[p]\coloneqq \{1,\ldots,p\}$. We assume a constant bias term is contained in the covariate matrix. 

\section{Reformulation of \texorpdfstring{$L_0-$Penalized}~~Regression}
\label{Sec:l0_reformulation}

The underlying assumption of best subset selection is that the response variables only depend on a subset of covariates $\Xmat_{\mathcal{A}}$, where $\mathcal{A}\subseteq\{1,\ldots,p\}$ is called the active set. The true active set $S$ is assumed to be much smaller than $p$. 
We decompose the regression coefficients as $\betav=\alphav\odot \zv$, using a spike-and-slab construction 
\citep{george1993variable},  where~$\odot$ denotes an element-wise product. The binary vector $\zv \in \{0,1\}^p$ represents the inclusion of covariates in the active set, and $\alphav \in \bR^p$ encodes the scale. With the augmented latent variables, the optimization problem \eqref{eq:bss_regular} can be equivalently expressed as
\ba{
\min_{\alphav, \zv} ~ \frac{1}{n} \norm{\yv - \Xmat (\alphav\odot \zv)}^2 + \lambda \norm{\zv}_0.
\label{eq:bss2} 
}
Similar to optimizing $\betav$ itself, optimizing $\zv$ is a combinatorial problem and remains NP-hard.

A major bottleneck in solving \Cref{eq:bss2} is  the discrete nature of variable $\zv$ that precludes computing the gradient with respect to it, which could have provided the direction of the steepest descent. This motivates us to reformulate the discrete optimization problem \Cref{eq:bss2} to an optimization problem in the continuous space. Such continuity greenlights the gradient-based methods, the workhorse behind modern machine learning. Specifically,  instead of directly optimizing $\zv$, we consider $\zv$ as a random variable with distribution $p(\zv; \piv) = \prod_{j=1}^p\text{Bern}(z_j; \pi_j)$, $\pi_j \in [0,1]$, where $\text{Bern}(z_j; \pi_j)$ is the Bernoulli distribution with parameter $\pi_j$.  Then the problem  in \Cref{eq:bss2} can be transformed to a probabilistic objective in the form of expectation. The new objective function  allows us to construct a stochastic gradient with Monte Carlo estimation. %
We have the following theorem for the reformulation. 

\begin{theorem}[Probabilistic Reformulation]
The $L_0$-regularized best subset selection problem \eqref{eq:bss_regular} is equivalent to the following problem
\ba{
\min_{\piv} ~\E_{\zv \sim p(\zv ; \piv)} \left[\min_{\alphav}\frac{1}{n}~ \norm{\yv - \Xmat (\alphav\odot \zv)}_2^2 + \lambda \norm{\zv}_0\right],
\label{eq:bss3} 
}
where $\piv = (\pi_1, \pi_2, \ldots, \pi_p) \in [0,1]^p$, $p(z_j = 1) = \pi_j,~ j \in \{1,\ldots,p\}$.
\label{thm:continuous}
\end{theorem}

The equivalence can be proved by the fact that the optimal solution of Eq.~\eqref{eq:bss2} is a feasible solution of Eq.~\eqref{eq:bss3} that achieves the same objective value, and vice versa. Though Eq.~\eqref{eq:bss3} optimizes over a space larger than that of Eq.~\eqref{eq:bss2}, the optima of $\pi_j$ is always at its extremes, $i.e.$ either zero or one. By Eq.~\eqref{eq:bss2}, we relax the optimization from a discrete space to a continuous space. The proof is  in Appendix~\ref{subsec:proof:thm:continuous}.

The objective in Eq.~\eqref{eq:bss3} is a bi-level optimization problem, where the inner optimization for a given $\zv$ is an ordinary least square (OLS) problem on the design matrix $\Xmat_{\cZ}$, which has a closed-form solution. Denote the active set inferred by $\zv$ as $\cZ: = \{j\}_{j:z_j \neq 0}$. Since the computation of OLS has a cubic-scaling cost with the number of regressors, the computational speed of the inner optimization increases when the size of  $\mathcal{Z}$ decreases, $i.e.$, the sparser the faster. Since we assume the true active set size $S$ is small, an algorithm that accurately selects variables can solve the inner optimization efficiently. 

For computational convenience, we reparameterize $\piv=(\pi_1,\ldots,\pi_p)$ with the sigmoid function as $\pi_j = \sigma(\phi_j) = 1/(1+\exp(-\phi_j)), ~ j \in [p]$ and further relax the optimization space to an unconstrained space. Although under the sigmoid reparameterization, probability $\pi_j$ can reach~0~or~1 only when the logit $\phi_j$ goes to the infinity, it can be accurately approximated in practice when the absolute values of the logits $\phi_j$ are sufficiently large.

A naive approach guaranteed to select the best subset is to exhaust all possible subsets. The exhaustion method is often infeasible because it requires evaluating the function under expectation $2^p$ times with $p$ covariates. 
A key to improving computational efficiency is minimizing the number of function evaluations. Instead of a random search, our idea is to guide the function evaluation by the first-order information given by the stochastic gradient. In the following section, we build a family of unbiased gradient estimators in the univariate case, then generalize the estimators to high dimensions.

\section{A Family of Unbiased Gradient Estimators}
\label{sec:estimators}

Consider a general probabilistic objective
\ba{
\min_{\phiv}~ \mathcal{E}(\phiv) = \E_{\zv \sim p_{\phiv}(\zv)}[f(\zv)]
\label{eq:objective},
}
where $p_{\phiv}(\zv) = \prod_{j=1}^p \text{Bern}(\sigma(\phi_j))$. The $L_0$-regularized objective in \Cref{eq:bss3} is a specific instance of \Cref{eq:objective} with function 
\ba{
f(\zv) = \min_{\alphav} \norm{\yv - \Xmat (\alphav\odot \zv)}_2^2/n + \lambda \norm{\zv}_0.
\label{eq:fun-f}
}

Taking the gradient of  $\mathcal{E}(\phiv)$ in \Cref{eq:objective} with respect to $\phiv$, and exchanging the derivative and integral, we have 
\ba{
\nabla_{\phiv} \mathcal{E}(\phiv) = \nabla_{\phiv} \E_{\zv \sim p_{\phiv}(\zv)}[f(\zv)] %
=\E_{\zv \sim p_{\phiv}(\zv)}[f(\zv) \nabla_{\phiv} \log p_{\phiv}(\zv)].
\label{eq:reinforce}
}
The last equality of \Cref{eq:reinforce} is called the \emph{score method} in statistics \citep{serfling2009approximation} and \emph{REINFORCE} in reinforcement learning (RL) \citep{williams1992simple}. 
The expectation in Eq.~\eqref{eq:reinforce} cannot be computed analytically for high dimensional $\zv$, but it can be approximated by an unbiased Monte Carlo estimation $\frac 1 K \sum_{k=1}^K f(\zv_k) \nabla_{\phiv} \log p_{\phiv}(\zv_{k}) $, with $\zv_1,\ldots,\zv_K \stackrel{iid}\sim p_{\phiv}(\zv)$. This approximation is widely used as the policy gradient in RL and robotics \citep{peters2006policy}. ~\looseness=-1

The continuous reformulation in Eq.~\eqref{eq:bss3} and the score method in Eq.~\eqref{eq:reinforce} make it possible to solve the best subset selection with gradient descent. One advantage of approximating the gradient in Eq.~\eqref{eq:reinforce} with Monte Carlo estimation is that the number of function evaluations does not grow with the dimension of variable $\zv$. This property is essential for scaling to high dimensions. Another advantage of the score method is that computing such a gradient estimator only requires the value of function $f(\zv)$ instead of its derivatives. It is applicable to the situations where $f(\zv)$ is discontinuous or even has no explicit expression ($e.g.$, in RL, $f(\zv)$ could be the unknown reward function for action $\zv$). 

However, the score function gradient is known for high variance \citep{greensmith2004variance,jang2016categorical,tucker2017rebar}. Though the Monte Carlo estimation with $K$ samples reduces the variance in the order $\cO(1/K)$, it needs many function evaluations for an accurate gradient estimate at each step. %

To solve this problem, we develop a general framework for unbiased gradient estimation with variance reduction techniques.  We construct a family of unbiased gradient estimators and identify the estimator with minimal variance in this family.  %

\subsection{Insights from univariate gradient setting}
\label{sec:uni}
We first develop gradient estimators for a univariate latent variable and then generalize them to high dimensional variables.  When the variable $z$ in Eq.~\eqref{eq:objective} is a scalar, the analytic gradient is
\ba{
\nabla_{\phi} \bE_{z\sim \text{Bern}(\sigma(\phi))}[f(z)] = \sigma(\phi)(1-\sigma(\phi))[f(1)-f(0)].
\label{eq:truegrad}
}
The closed form of expectation above provides an explicit condition that an unbiased gradient estimator must satisfy. To cope with the best subset selection objective in Eq.~\eqref{eq:bss3}, the gradient estimator should further not rely on the continuity of function $f(\zv)$, and the number of required function evaluations does not increase with the dimension of the covariates. According to these desiderata, we define a family of gradient estimators.

\begin{definition}[Unbiased gradient family]
For an objective $\bE_{z\sim \emph{\text{Bern}}(\sigma(\phi))}[f(z)]$, suppose an estimator of the gradient with respect to $\phi$ is in the form of $g(u; \sigma(\phi))$ where $u \sim \emph{\text{Unif}}(0,1)$ is a uniform random variable. For function $f:\{0,1\} \to \bR$, the unbiased gradient family $\mathcal{G}$ consists of function $g(u; \sigma(\phi))$ that satisfies
\begin{itemize}
 \item Unbiasedness:~~ $\E_{u\sim\emph{\text{Unif}}(0,1)}[g(u; \sigma(\phi))] = \sigma(\phi)(1-\sigma(\phi))[f(1)-f(0)],$
 \item Functional form: 
 \ba{g(u; \sigma(\phi)) = a(u;\sigma(\phi))f(\hpless) + b(u;\sigma(\phi))f(\hpgreater)
 \label{eq:linear_form_estimator}
 } 
 where $a(u;\sigma(\phi))$, $b(u;\sigma(\phi))$ are independent of function $f(\cdotv)$.
\end{itemize}
\label{def:univariate}
\end{definition}

The gradient family $\mathcal{G}$ in \Cref{def:univariate} is a set of unbiased gradient estimators with a specific functional form in Eq.~\eqref{eq:linear_form_estimator}.
This functional form ensures efficiency when generalizing to the high dimensional settings, which will be discussed in detail in Section~\ref{sec:multi}. The estimator family $\mathcal{G}$ with this functional form incorporates many popular unbiased gradient estimators. Here we list several representative ones. 

The REINFORCE estimator belongs to the family $\mathcal{G}$ with
\ba{
g_{R}(u; \sigma(\phi)) = f(\hpless)(\hpless - \sigma(\phi)).
\label{eq:reinforce1}
}
A recently proposed ARM gradient \citep{yin2018arm}  is an element of $\mathcal{G}$ with
\ba{
g_{\text{ARM}}(u; \sigma(\phi)) = [f(\hpgreater) - f(\hpless)](u - \textstyle\half).
\label{eq:arm1}
}
We can add an indicator mask to the ARM gradient without changing its univariate distribution, which we call it ARM$_0$ estimator, with the expression as
\ba{
\resizebox{.92\hsize}{!}{$\textstyle g_{\text{ARM}_0}(u; \sigma(\phi)) = [f(\mathbf{1}_{[u>\sigma(-\phi)]}) - f(\hpless)](u - \half)\big|\mathbf{1}_{[u>\sigma(-\phi)]} - \hpless\big|$.}
\label{eq:arm-mask1}
}
In the univariate case, the ARM$_0$ estimator is identical to the ARM estimator, but when it comes to the multivariate case, %
ARM$_0$ can produce sparse gradients where many elements may become exact zeros.  %
This straightforward sparsification of the ARM estimator has been adopted by several recent works 
\citep{boluki2020aistats,CARSM2020,ARMhash2020}. 

In this paper, we propose a new gradient estimator that can further reduce the  variance, which is given by
\ba{
g_{\text{U2G}}(u; \sigma(\phi)) = \frac{\sigma(|\phi|)}{2}[f(\hpgreater) - f(\hpless)](\mathbf{1}_{[u>\sigma(-\phi)]} - \hpless).
\label{eq:u2g1}
}
We call it unbiased uniform gradient (U2G) estimator because it takes a constant value at the non-zero region. U2G estimator is derived by finding the minimum-variance unbiased estimator (MVUE) from the family in \Cref{def:univariate}; \Cref{sec:optimality} will discuss the derivation in detail.  U2G estimator is concurrently discovered by \cite{disarm2020} via Rao-Blackwellization over the ARM estimator and exhibits promising performance in optimizing neural network parameters of deep  generative models. We summarize the estimators above in Table~\ref{tab:estimators} in a functional form compatible with Definition~\ref{def:univariate}.
\begin{table}[t]
\centering
\resizebox{\textwidth}{!}{\begin{tabular}{c|c|c|c|c}
\toprule
&REINFORCE&ARM & ARM$_0$& U2G \\
\hline
$a(u;\sigma(\phi))$ & $\hpless - \sigma(\phi)$& $\half-u$& $( \half - u)|\mathbf{1}_{[u>\sigma(-\phi)]} - \hpless|$ & $\sigma(|\phi|)( \hpless-\mathbf{1}_{[u>\sigma(-\phi)]})/2$\\
$b(u;\sigma(\phi))$ &0&$u - \half$& $(u - \half)|\mathbf{1}_{[u>\sigma(-\phi)]} - \hpless|$& $\sigma(|\phi|)(\mathbf{1}_{[u>\sigma(-\phi)]}-\hpless)/2$\\ 
\bottomrule
\end{tabular}}
\caption{ \small Parameterization of unbiased gradient estimators according to \Cref{def:univariate}.}
\label{tab:estimators}
\end{table}

We are interested in the bias-variance behvior of the gradient estimators. The estimators of the family $\mathcal{G}$ are unbiased thus  having the same expectation value. But the  variance differs between them, depending on how the estimator is expressed as a function $g(u;\sigma(\phi))$ of a uniform variable  $u$. As an illustrative example,  Figure~\ref{fig:cf} shows the functions 
$g(u)$ and $g^2(u)$ for the  estimators in \Cref{eq:reinforce1,eq:arm1,eq:u2g1} with $f(1)=5$, $f(0)=4$, and $\pi=\sigma(\phi)=2/3$. Since the estimators are unbiased, the net signed areas under the curve of the first row in Figure~\ref{fig:cf} are the same. The variance of each gradient estimator, up to the same additive constant, is represented by the area under the curve in each subplot of the second row.

\begin{figure}[t]
\subfloat{\includegraphics[width=0.32\linewidth]{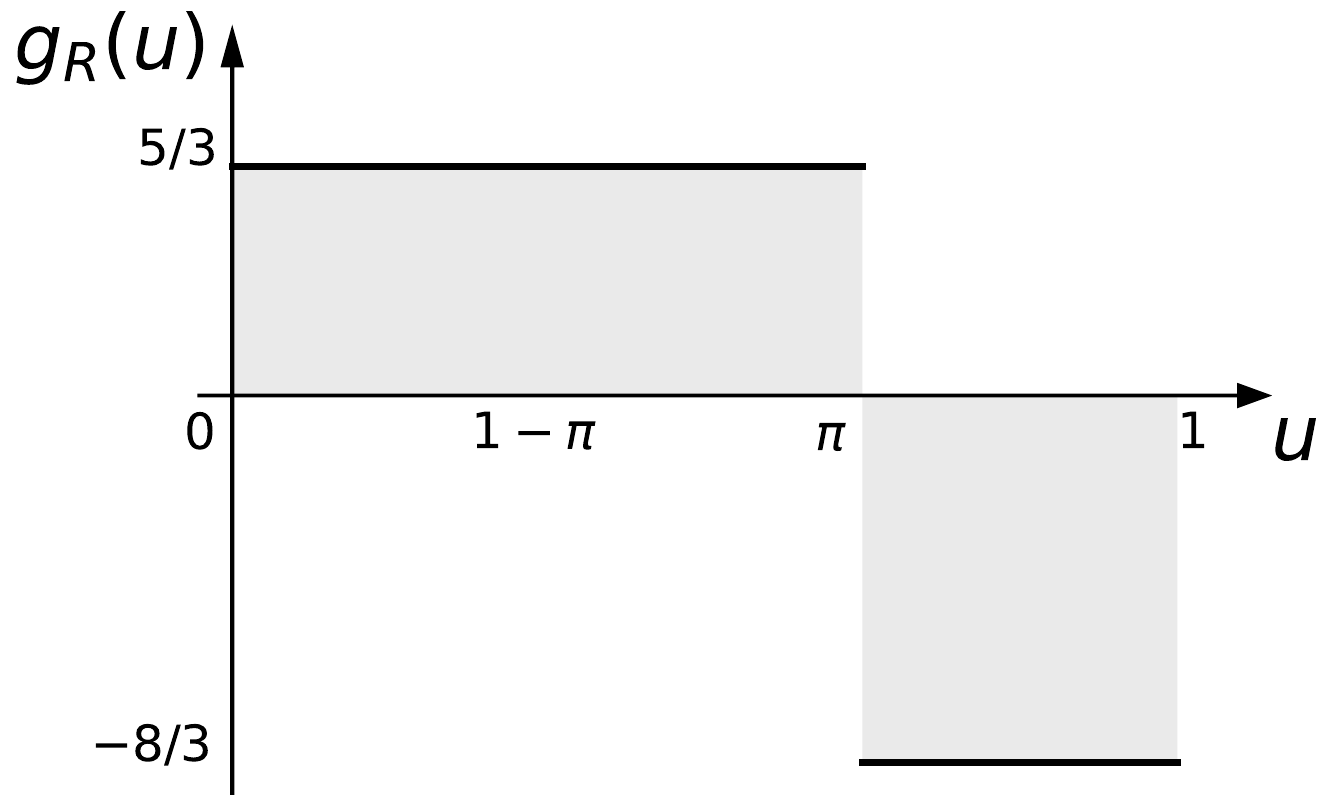}}\hfill
\subfloat{\includegraphics[width=0.32\linewidth]{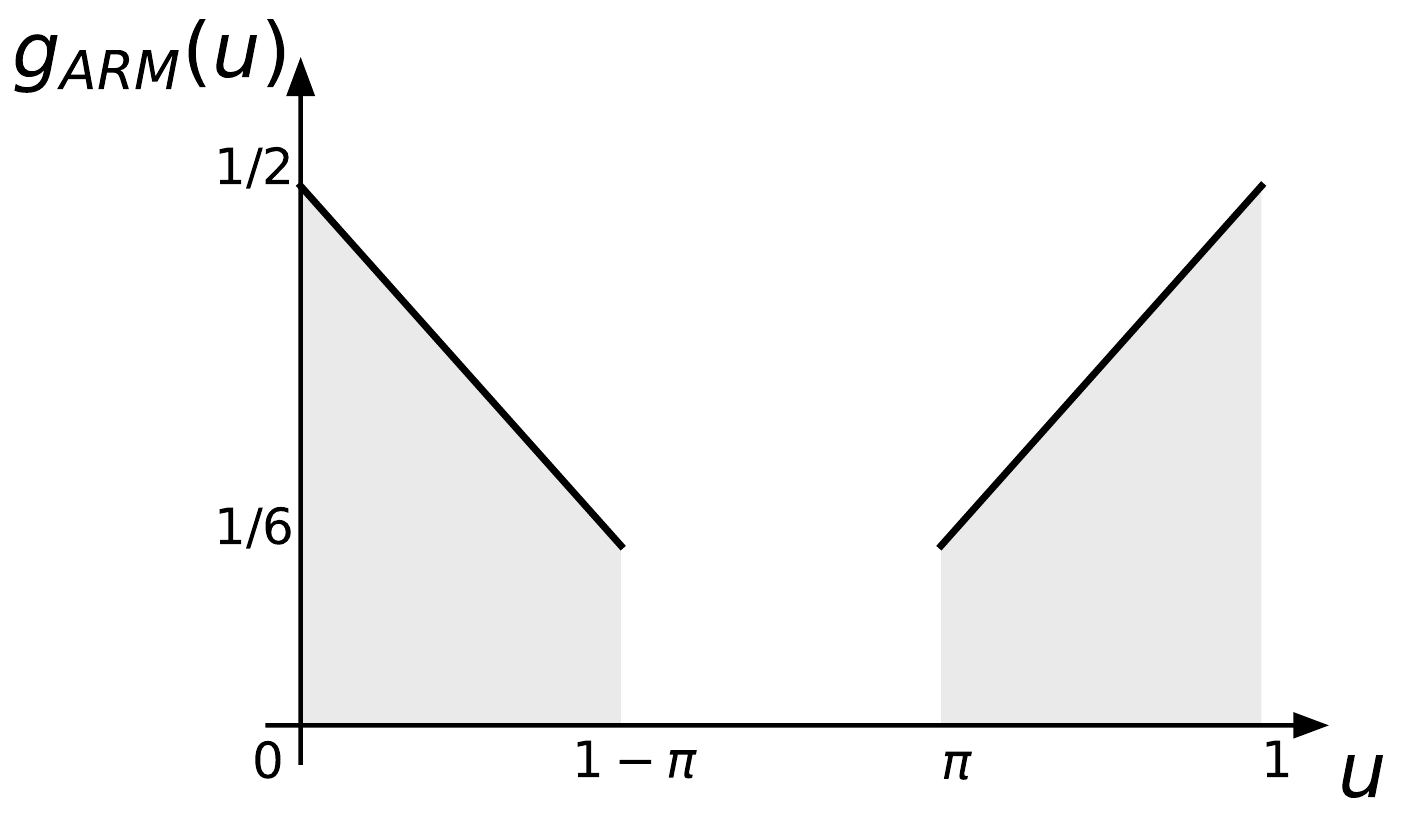}}\hfill
\subfloat{\includegraphics[width=0.32\linewidth]{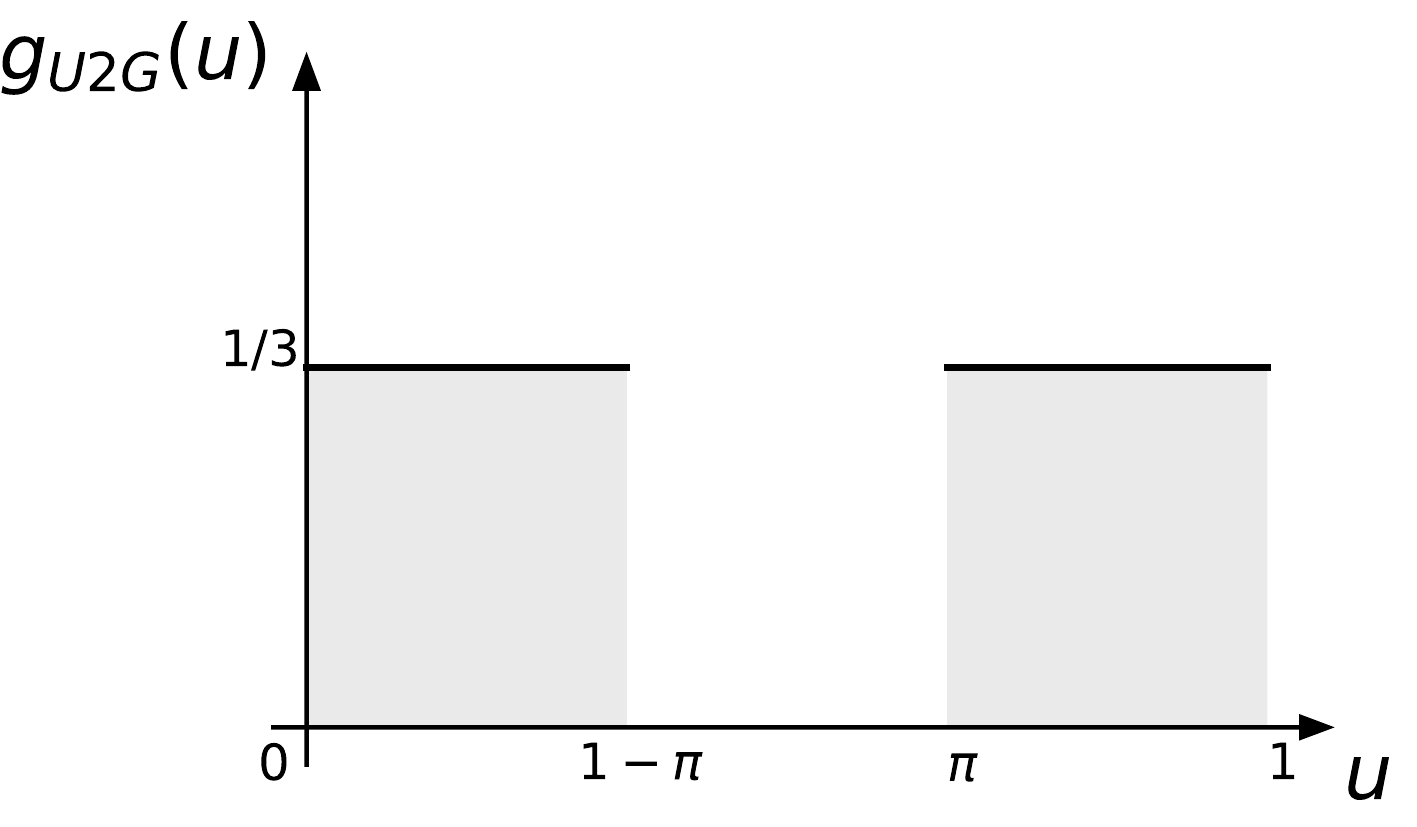}} \\
\subfloat[REINFORCE]{\includegraphics[width=0.32\linewidth]{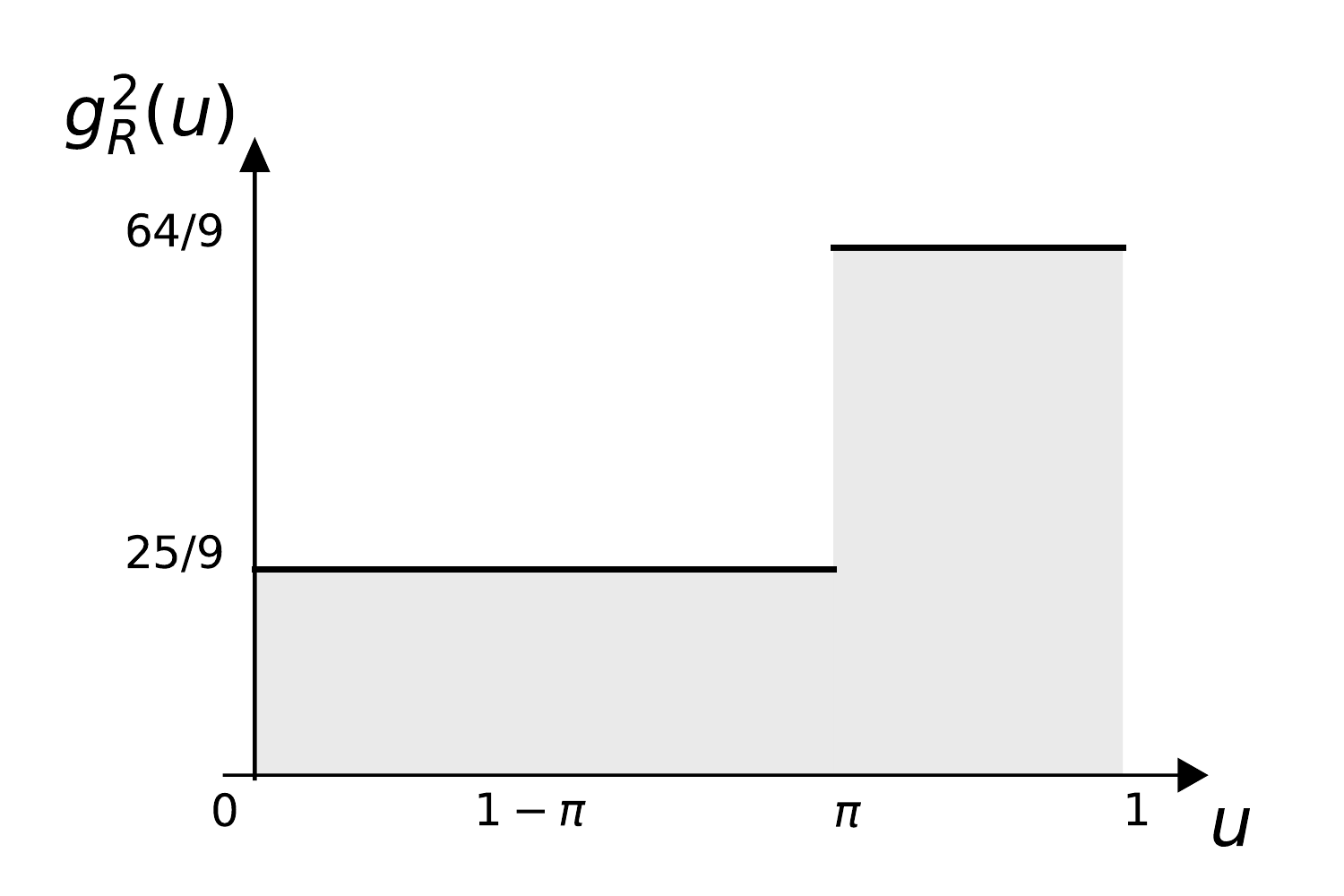}}\hfill
\subfloat[ARM]{\includegraphics[width=0.32\linewidth]{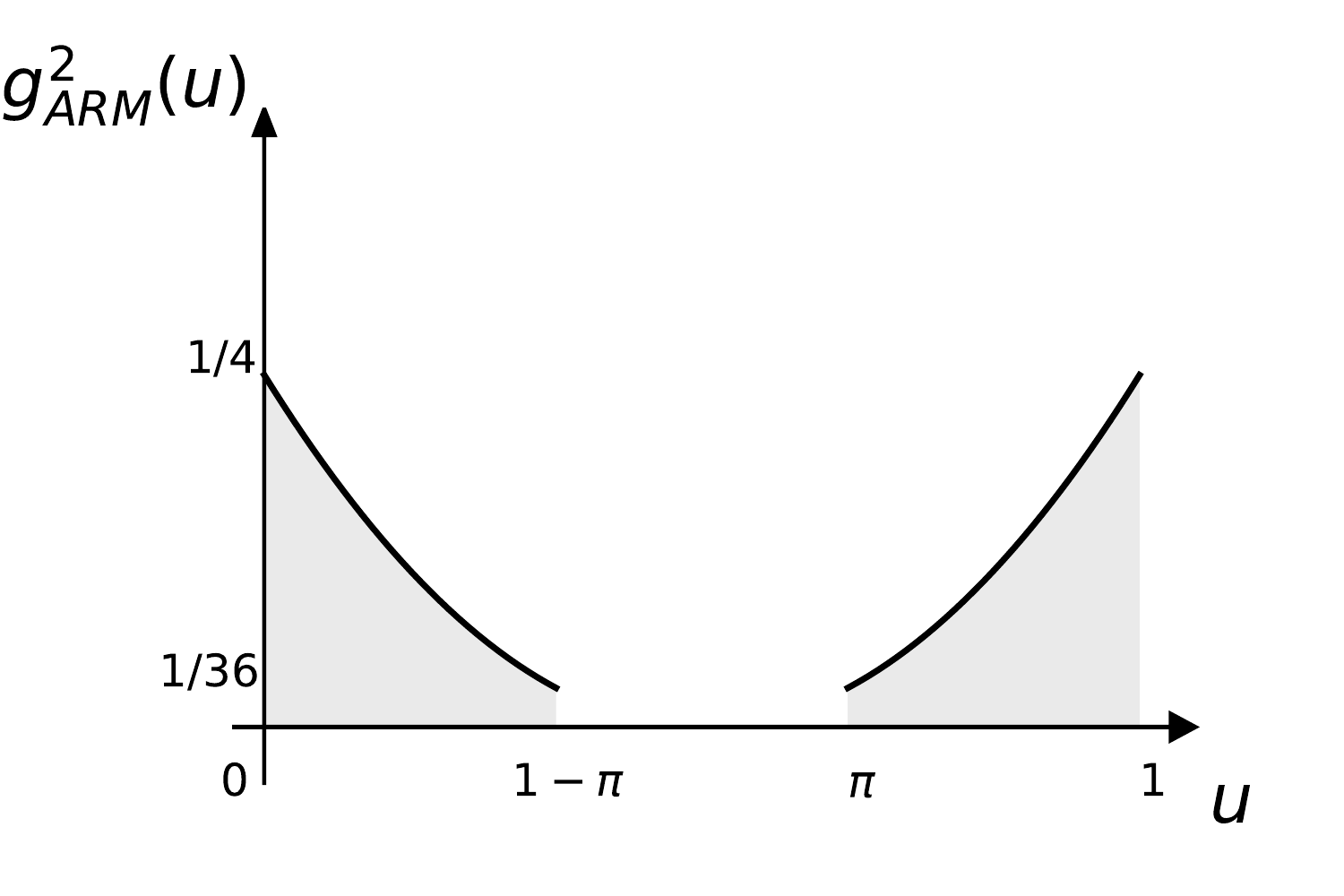}}\hfill
\subfloat[U2G]{\includegraphics[width=0.32\linewidth]{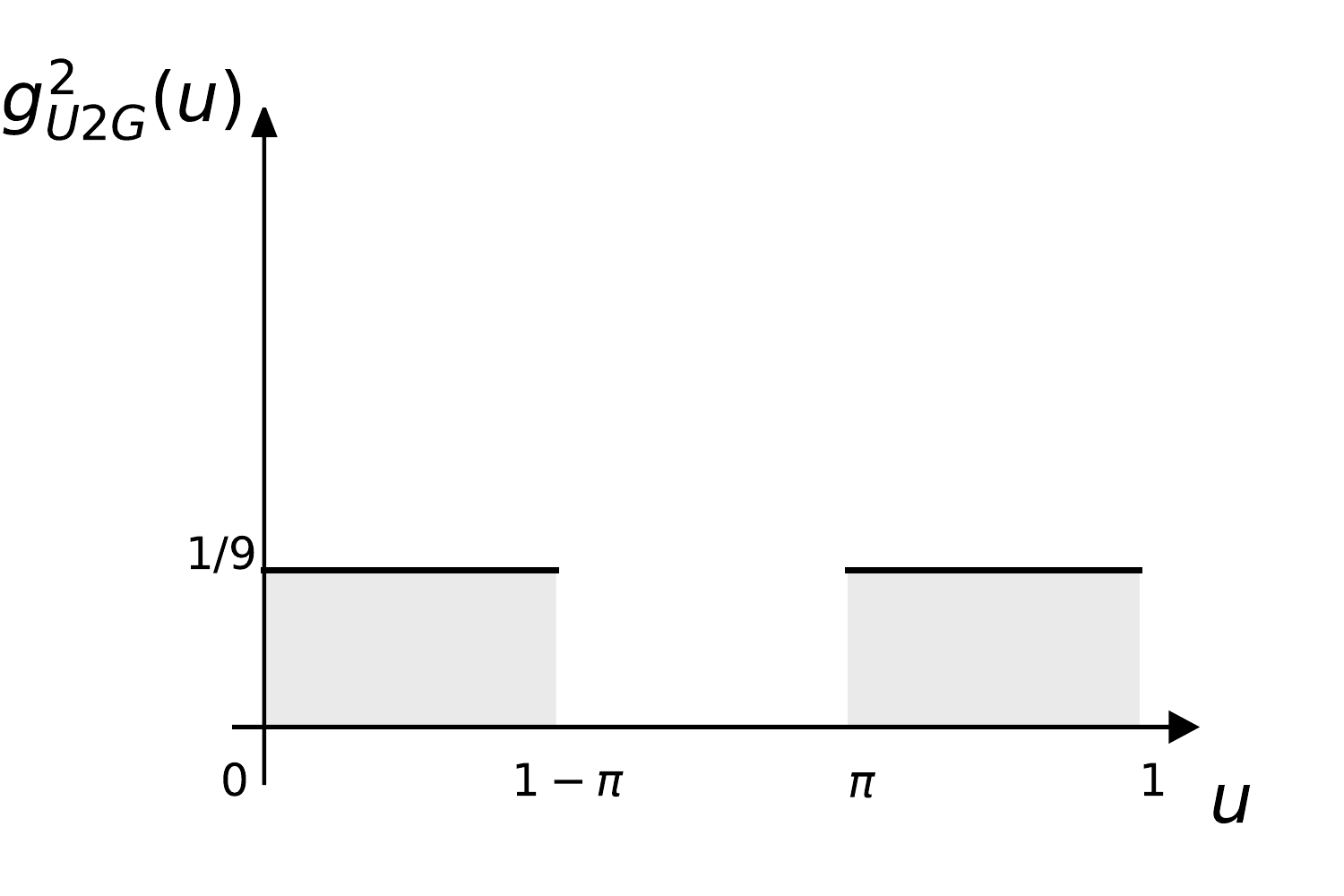}}
\caption{ \small The characteristic curves of gradient estimators. In this illustrative example, $f(1)=5$, $f(0)=4$, and $\pi=\sigma(\phi)=2/3$. The top row is the function $g(u)$ with respect to $u$; the second row is the function $g^2(u)$. %
The unbiased estimators have the same integration in the first row, but different gradient variance as shown in the second row (up to a constant). 
}
\label{fig:cf}
\end{figure}

 Compare the first column of Figure~\ref{fig:cf} to the other two columns, and compare Eq.~\eqref{eq:reinforce1} to Eqs.~\eqref{eq:arm1} and \eqref{eq:u2g1}. Intuitively, if the function $f(\cdot)$ appears in the estimator as a relative difference $f(z) - f(z')$, the scale of the gradient does not increase with the scale of $f(\cdot)$, and  the magnitude of the second moment is controlled. The relative difference $f(z) - f(z')$ echos the model comparison nature of variable selection. %
 
 Comparing the second and the third columns of Figure~\ref{fig:cf}, intuitively, we find the variance (or the second moment) is reduced if the direction and magnitude of the gradient estimator do not change with variable $u$. Formally, we have the following result.
\begin{proposition}
\label{prop:R-ARM-U2G}
For positive (or negative) function $f(z)$, $i.e.~f(z)\geq 0$, we have
\bas{
\var[g_{\emph{\text{U2G}}}] \leq \var[g_{\emph{\text{ARM}}}] \leq \var[g_{R}],
}
where the second inequality requires $|f(1) - f(0)| \leq \min\{|f(1)|, |f(0)|\}$.
\end{proposition}
The proof of \Cref{prop:R-ARM-U2G} is presented in
Appendix~\ref{subsec:proof:prop:R-ARM-U2G}.

\subsubsection{Optimality of U2G estimator}
\label{sec:optimality}
We derive U2G estimator by identifying the estimator with minimal variance  in the family $\mathcal{G}$ in Definition~\ref{def:univariate}. To simplify the notation, let $f_1 = f(1)$, $f_0 = f(0)$, $\pi = \sigma(\phi)$, and $\Delta = |f_0 - f_1|$. Without loss of generality, we  assume $\pi \geq 1/2$. An ideal gradient for variable selection should be able to distinguish the potential models, even when their difference is small. Accordingly, a gradient estimator should have a non-diminishing SNR, defined as $\text{SNR} \coloneqq \E[g(u)]/\sqrt{\var[g(u)]}$. ~\looseness=-1 

However, as shown in Eq.~\eqref{eq:truegrad}, the scale of the true gradient 
diminishes as the difference $\Delta$ between the potential models shrinks. 
For a non-diminishing SNR, the variance of the estimator $g(u)$ has to decrease to zero as $\Delta \to 0$,
that is 
\ba{
\lim_{\Delta \to 0}~\var[g(u;\pi)] = 0, \quad \text{for all} \ \pi.
\label{eq:var_condition}
}
This condition ensures the estimated gradient can distinguish the optimal model from the others, even when the objective values are close.  

Under the condition in \Cref{eq:var_condition}, the following proposition shows that U2G is the uniformly minimum-variance unbiased estimator (UMVUE) within the proposed estimator family $\mathcal{G}$. U2G estimator hence has   the optimal statistical efficiency. The proof is in  \Cref{sec:umvue}. ~\looseness=-1 
\begin{proposition}
Among the unbiased gradient estimators defined in Definition~\ref{def:univariate} and assume $\forall \pi,~ \lim_{|f(1)-f(0)| \to 0} \allowbreak ~\var[g(u;\pi)] = 0$, U2G has the uniformly minimum variance for all $\pi$.
\label{prop:umvue}
\end{proposition}

Specifically, for univariate latent variable, the variance of U2G estimator is
\bas{
\var[g_{\text{U2G}}(u;\pi)] &= ~\pi |\pi - \half| (1-\pi)\max\{\pi,1-\pi\}[f(1)-f(0)]^2 \\
&\leq C [f(1)-f(0)]^2
}
with $C \approx 0.0388$. The $\text{SNR}$ for U2G estimator is 
\bas{
\text{SNR}(\pi) = \sqrt{\pi(1-\pi)/(|\pi - \half| \max\{\pi,1-\pi\})},
}
which is the same for arbitrary function $f(\cdot)$ in the objective, and only vanishes when the algorithm converges, $i.e.$ $\pi \to 0~\text{or}~1$. Similar properties hold for the ARM estimator. The variance and $\text{SNR}$ of the ARM and U2G estimators are shown in Figure~\ref{fig:var_uni}. U2G estimator has lower variance and higher SNR than ARM estimator, especially when the Bernoulli probability  the uncertainty is high ($\pi$ is close to $0.5$). 

\begin{figure}[!t]
\centering
\includegraphics[width=0.83\textwidth]{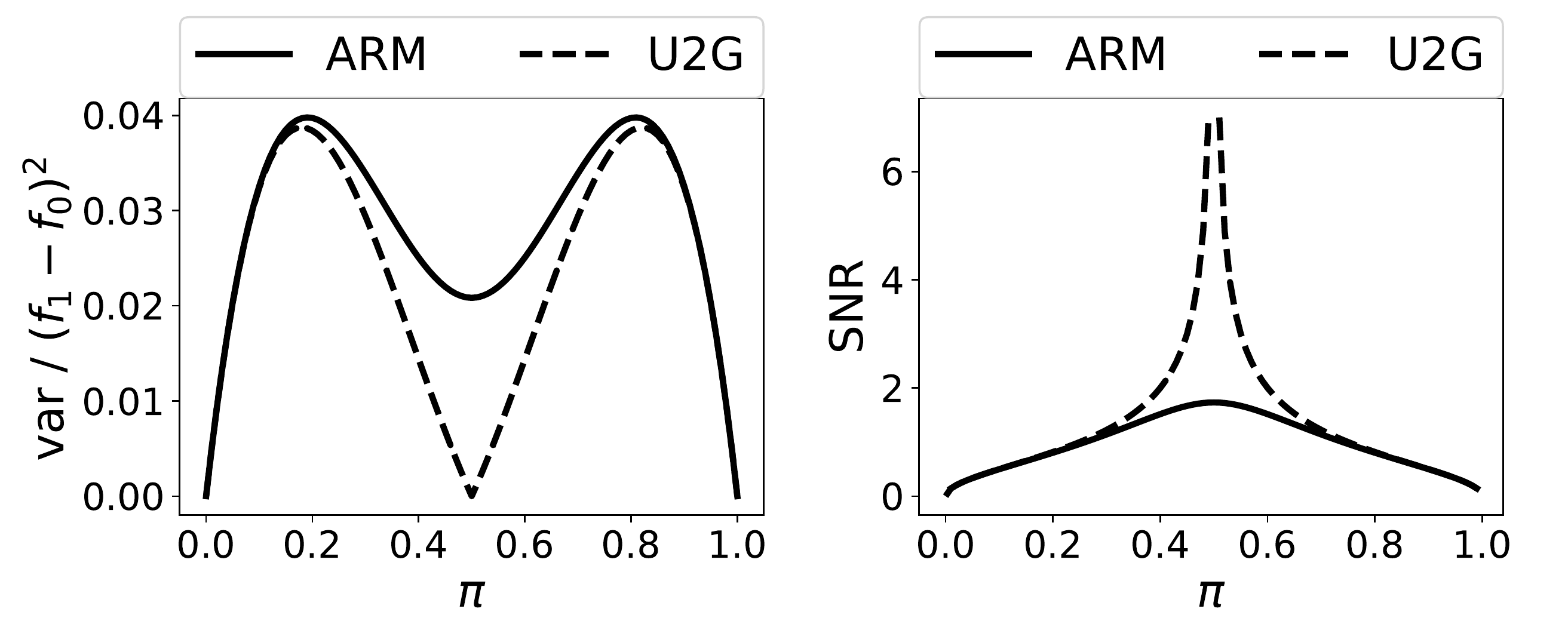}
\caption{\small Variance and $\text{SNR}$ of univariate ARM and U2G estimators.}
\label{fig:var_uni}
\end{figure}

\subsection{Multivariate Generalization}
\label{sec:multi}
We generalize the univariate gradient estimators in \Cref{sec:uni} to the setting with high dimensional variable $\zv$. When $z$ is univariate, it is unnecessary to estimate the gradient given the true gradient in \Cref{eq:truegrad}. 

However, the true gradient is not accessible in the multivariate case. An element of the true gradient vector is

\ba{
\pde{\phi_v} \mathcal{E}(\phiv) =& \pde{\phi_v} \E_{\zv \sim \prod_{j=1}^p p(z_j; \sigma(\phi_j))}[f(\zv)] \notag \\
=& \E_{\zv_{-v} }[ \sigma(\phi_v)(1-\sigma(\phi_v))(f(\zv_{-v},z_v=1) - f(\zv_{-v},z_v=0)) ].
\label{eq:true-multi}
}

Estimating the gradient vector element-wisely by \Cref{eq:true-multi} requires high computational cost. To estimate the $v$-th element with the Monte Carlo method by Eq.~\eqref{eq:true-multi}, we first sample a binary vector $\zv$. Then the function $f(\cdot)$ needs to be computed on two binary vectors with the $v$-th element as 0 and 1 respectively and other elements equal to $\zv_{-v}$. %
 This has to be done for each element of $\zv$ separately. Therefore, it requires at least $2p$ evaluations of function $f(\cdot)$  to get an unbiased gradient estimate at each step. 
 For the best subset selection, the function $f(\cdot)$ in \Cref{eq:fun-f} involves a least square regression on the support set of $\zv$. Hence, evaluating $f(\cdot)$ $2p$ times for each gradient step is  computationally intractable for the large-$p$ setting. ~\looseness=-1

 The functional form of the estimators in Definition~\ref{def:univariate} circumvent this computational problem. %
 Applying the univariate gradient estimators in Definition~\ref{def:univariate}, %
 for multivariate $\zv$, an element of the gradient vector can be computed as
\ba{
&\pde{\phi_v} \bE_{\zv \sim \prod_{j=1}^p p(z_j; \sigma(\phi_j))} [f(\zv)] \notag \\ =&~ \bE_{\zv_{-v} } \pde{\phi_v} \bE_{z_v \sim p(z_v; \sigma(\phi_v))} [f(z_v, \zv_{-v} )] \notag \\
=&~ \bE_{\zv_{-v} } \E_{u_v\sim\text{Unif}(0,1)}[a(u_v;\sigma(\phi_v))f(\mathbf{1}_{[u_v<\sigma(\phi_v)]},\zv_{-v}) + b(u_v;\sigma(\phi_v))f(\mathbf{1}_{[u_v>1-\sigma(\phi_v)]},\zv_{-v}) ] \notag \\
=&~ \E_{\uv \sim \prod_{j=1}^p\text{Unif}(0,1)}[a(u_v;\sigma(\phi_v))f(\hplessv) + b(u_v;\sigma(\phi_v))f(\hpgreaterv)].
\label{eq:multivariate}
}
The first equality is by the factorization of $p(\zv)$, the second equality is by the unbiasedness of the univariate gradient estimator form in \Cref{def:univariate}, and the last equality is by the law of the unconscious statistician (LOTUS) \citep{ross2014introduction}. \Cref{eq:multivariate} ensures that 
\ba{
\resizebox{0.92\hsize}{!}{$\gv(\uv; \sigma(\phiv)) = a(u_v;\sigma(\phi_v))f(\hplessv) + b(u_v;\sigma(\phi_v))f(\hpgreaterv), ~\uv \sim \prod_{j=1}^p\text{Unif}(0,1),$}
\label{eq:multi-estimator}
}
is an unbiased estimator for the gradient of the objective \Cref{eq:objective}. 

The estimator in \Cref{eq:multi-estimator} has a key computational advantage. We can evaluate $f(\hplessv)$ and $f(\hpgreaterv)$ as few as a single time with $\uv \sim \prod_{j=1}^p\text{Unif}(0,1)$, and share it across all the elements of the gradient vector, greatly reducing computational time.

Written in a vector form, the estimators in Section~\ref{sec:uni} have their multivariate form as
\begin{gather}
\resizebox{0.5\hsize}{!}{$\gv_{R}(\uv;\sigma(\phiv)) = f(\hplessv)(\hplessv - \sigma(\phiv))$}, \nonumber \\
\resizebox{0.92\hsize}{!}{$\gv_{\text{ARM}_0}(\uv;\sigma(\phiv))= [f(\hpgreaterv) - f(\hplessv)](\uv - \half)\odot|\hpgreaterv - \hplessv|$}, \nonumber \\
\resizebox{0.92\hsize}{!}{$ \gv_{\text{U2G}}(\uv;\sigma(\phiv)) = \frac{1}{2}[f(\hpgreaterv) - f(\hplessv)]\sigma(|\phiv|)\odot(\hpgreaterv - \hplessv)$,} \label{eq:grads-multi}
\end{gather}
where $\uv \sim \prod_{j=1}^p\text{Unif}(0,1)$, and all the operations are element-wise. Due to the indicator mask, the gradient vectors of ARM$_0$ and U2G are sparse when the probability close to the extremes such as when the algorithm is close to
the convergence. We observe in practice that the sparsity in gradient estimation, while not required to ensure unbiasedness, can improve the stability of the  convergence process. 

In practice, the gradient can be computed as the Monte-Carlo estimate of $\gv(\uv;\sigma(\phiv))$. With $\uv_k \sim \prod_{j=1}^p\text{Unif}(0,1)$, $k =1, \cdots, K$, the Monte-Carlo estimate is
\ba{
\hat{\gv}(\phiv) = \sum_{k=1}^K \gv(\uv_k;\sigma(\phiv)),
\label{eq:grads-multi-mc}
}
where $\gv(\cdot)$ is in \Cref{eq:multi-estimator} or as one of the specific estimators in Eq.~\eqref{eq:grads-multi}.

Consider the variance of the gradient estimator $\gv(\uv; \sigma(\phiv))$, $i.e.$ the diagonal of the covariance matrix. By the law of total variance, the variance of element $v$ of the gradient vector can be decomposed as
\ba{
\var_{\uv}[\gv_{v}(\uv; \sigma(\phiv))] 
= \var\{\E[\gv_{v}(\uv;\sigma(\phiv))| \uv_{-v}]\} + \E\{\var[\gv_{v}(\uv;\sigma(\phiv))| \uv_{-v}]\}.
\label{eq:multi_var} 
}
The first term on the right-hand side (RHS) of Eq.~\eqref{eq:multi_var} is the irreducible variance, shared by all unbiased gradient estimators. It can be further computed as
\bas{
\var_{\uv_{-v}}\{\E_{u_v}[\gv_{v}(\uv;\sigma(\phiv))| \uv_{-v}]\} = (\pi_v)^2(1-\pi_v)^2\var_{\uv}[\Delta_{\zv,v}f],
}
where $\zv = \hplessv$,  $\Delta_{\zv,v}f : = \expfunc (\tilde{\zv}) - \expfunc(\zv)$, and  $\tilde{\zv} \in \{0,1\}^{ p}$   differing from $\zv$ only at the v-th dimension. %
The second term of Eq.~\eqref{eq:multi_var} measures the average variance of the estimator vector in a single dimension. Given a fixed $\uv_{-v}$, as shown in the univariate case, U2G estimator has the minimal variance for all estimators in Definition~\ref{def:univariate} with non-vanishing SNR. Therefore, by averaging over all $\uv_{-v}$, the second term of U2G estimator is small. This means the total variance of U2G estimator is well controlled in the multivariate case.

\begin{algorithm}[t] \small{
\SetKwData{Left}{left}\SetKwData{This}{this}\SetKwData{Up}{up}
\SetKwFunction{Union}{Union}\SetKwFunction{FindCompress}{FindCompress}
\SetKwInOut{Input}{input}\SetKwInOut{Output}{output}
\Input{ \small
Bernoulli distribution $\{q_{\phi_j}(z_j)\}_{j\in[p]}$ with probability $\{\sigma(\phi_j)\}_{j\in[p]}$, target $\mathcal{E}(\phiv) = \E_{\zv \sim p_{\phiv}(\zv)}[f{}(\zv)]$, $\zv = (z_1,\cdots,z_p)$, $\phiv = (\phi_1,\cdots,\phi_p)$, $p_{\phiv}(\zv)=\prod_{j=1}^p p_{\phi_j}(z_j)$
}
\Output{  Maximum likelihood estimator of $p_{\phiv}(\zv)$ as $\hat{\zv} = \mathbf{1}_{[\sigma(\phiv)>1/2]}$ }
\BlankLine
 Initialize $\phiv$ randomly

\While{not converged}{
Sample $\uv_k \distas{i.i.d.} \prod_{j=1}^p\text{Unif}(0,1)$ for $k=1,\cdots,K$

Evaluate $f(\mathbf{1}_{[\uv_k>1-\sigma(\phiv)]})$ and $f(\mathbf{1}_{[\uv_k<\sigma(\phiv)]})$

Compute $\gv_k = g(\uv_k;\sigma(\phiv),f)$ by an estimator in Eq.~\eqref{eq:grads-multi}

Update $\phiv = \phiv - \frac{1}{K}\rho \sum_{k=1}^K \gv_k $ with stepsizes $\rho $
}
\caption{\small Best subset selection with probabilistic reformulation }\label{alg:M1}}
\end{algorithm}

\section{Convergence in Expectation}
\label{Sec:convergence}

In this section, we provide theoretical insights to the convergence properties of the gradient method under the expectation of data generation and gradient estimation. %
We assume that the observations $(\Xmat, \yv)$ are generated from the following model with the active set $\cA \subset \{1,2,\cdots,p\}$
\ba{
&\yv = \Xmat \betav^*  + \epsilonv, \quad \epsilonv \sim \cN(0, \sigma^2\Imat),
\label{eq:freq_model}
}
where $\beta^*_j = 0$ for $j \notin \mathcal{A}$. Let $\zv^*  \in \{0,1\}^p$ indicate the true active set where $z^*_j$ equals~$1$ if $j \in \mathcal{A}$ and $0$ otherwise. We assume a random design matrix $\Xmat = (\xv_1, \cdots, \xv_n)^{\top}$ in which $\xv_i \sim \cN(0, \Imat_{p})$ for $i \in [n]$. In order to ease the presentation, we denote
\ba{
f_{\Xmat, \yv}(\zv) =&  \min_{\alphav}  \frac{1}{n}~ \norm{\yv - \Xmat (\alphav\odot \zv)}_2^2 + \lambda \norm{\zv}_0.
\label{eq:f_freq} 
}

Here we use subscripts to make the dependency of $f$ on $(\Xmat,\yv)$ explicit. Denote $\Xmat_{\zv} \in \bR^{n\times \|\zv\|_0}$ as the  matrix consisting of $\{X_j: z_j \ne 0\}$, and $\Xmat_{-\zv}$ the complement in the design matrix. Furthermore, for any $\zv, \tilde{\zv} \in \{0,1\}^{ p}$ such that $z_k = 0$, $\tilde{z}_{k} = 1$, $z_j=\tilde{z}_j$ for all $j \ne k$, we denote $\deltaf : = \expfunc_{\Xmat, \yv} (\tilde{\zv}) - \expfunc_{\Xmat, \yv}(\zv)$. All the proof details in this section are given in the appendix.

First, we have the following lemma for the expectation of the gradient over the randomness of~$\uv$, given any fixed training data $(\Xmat, \yv)$.
\begin{lemma}
\label{lemma:expect_grad}
Consider $g_{\emph{\text{U2G}}}(\uv; \sigma(\phiv))$ as in Eq.~\eqref{eq:grads-multi}. For each $(\Xmat, \yv)$, we have
\begin{align*}
       \E_{\uv \sim\prod_{j=1}^p {\emph{\text{Unif}}}(u_j;0,1)} [\gv_{\emph{\text{U2G}}}(\uv; \sigma(\phiv))] = \piv (1-\piv)\odot\E_{\uv}[\Delta_{\zv}f ],
\end{align*}
where $\piv = (\sigma(\phi_1), \cdots, \sigma(\phi_p))$, $\zv = \hpgreaterv$, and $\Delta_{\zv}f = (\Delta_{\zv,1}f, \cdots, \Delta_{\zv,p}f)$. %
\end{lemma}

Lemma~\ref{lemma:expect_grad} shows that the gradient is closely related to $\Delta_{\zv}f$, whose randomness comes from latent variable $\uv$ and data $(\Xmat, \yv)$. 
By analyzing the expectation of $\Delta_{\zv}f$, the following result establishes the expectation of stochastic gradients %
\begin{lemma}
\label{lemma:expectation_u2g}
Given $g_{\emph{\text{U2G}}}(\uv; \sigma(\phiv))$ in Eq.~\eqref{eq:grads-multi}, the expected gradient is
\begin{align}
\E_{\Xmat, \yv, \uv}[g_{\emph{\text{U2G}}}(\uv; \sigma(\phiv))_k] = \biggr[ \lambda - \biggr( \frac{(n-\E[\norm{\zv}_0 \Big| \|\zv\|_{0}  < n] -1)(\beta_{k}^{*})^2}{n}  & \nonumber \\
& \hspace{-23 em} + \frac{\sigma^2 + \E_{\uv}[\norm{\betav_{-\zv}^{*}}_{2}^2 \Big| \|\zv\|_{0}  < n]}{n}  \biggr)  p(\|\zv\|_0 < n)\biggr]\times \pi_k(1-\pi_k), \label{eq:expectation_u2g}
\end{align}
for any $k \in \{1,2,\cdots,p\}$ where $\piv = (\sigma(\phi_1), \cdots, \sigma(\phi_p))$ and $\zv = \hpgreaterv$.
\end{lemma}

In the following proposition, based on the results of Lemma~\ref{lemma:expectation_u2g}, we show that if the sample size and  true coefficient magnitude are not too small, then with proper hyper-parameter $\lambda$ controlling the penalty strength, each element of the expected gradient points to the  direction that can recover the true active set.

\begin{proposition}
\label{prop:expected_grad}
Assume $\sum_{j=1}^p \sigma(\phiv_{j}) / (n - 1) \leq 1-\eta$, for certain $\eta \in (\frac 1n,1)$. If $n$ is sufficiently large such that $\sqrt{p\log(n)/2(n-1)^2} \leq \eta$ and $(\norm{\betav^{*}}_2^2+ \sigma^2) /(n-1)\min_{k \in \mathcal{A}}{(\beta_k^*)^2} \leq \eta $, then there exists $\lambda>0$ such that
\bas{
\E_{\Xmat, \yv, \uv}[g_{\emph{\text{U2G}}}(\uv; \sigma(\phiv))_j] < 0, ~\forall j\in \mathcal{A}; \quad \E_{\Xmat, \yv,\uv}[g_{\text{U2G}}(\uv; \sigma(\phiv))_j] > 0, ~\forall j\notin \mathcal{A}.
}
\end{proposition}

\begin{remark}
If the gradient points to the right direction element-wisely, then for each gradient step, in expectation, $\pi_j$ increases if and only if $j \in \mathcal{A}$. Therefore, if
\begin{align}
\varpi = \left(1-\min\left\{\sqrt{\frac{p\log(n)}{2 (n - 1)^2}},\frac{\norm{\betav^*}_2^2+ \sigma^2 }{(n-1)\min_{k \in \mathcal{A}}{(\beta_k^*)^2}}
\right\}\right)n > S, \label{eq:condition_initialization}
\end{align}
with initialization
\bas{
\sum_{j=1}^p \pi_j^{(0)} \leq \varpi - S
}
in expectation,  $\piv$ in Algorithm~\ref{alg:M1} converges to the indicator of the true active set.
\end{remark}
The proof of Proposition \ref{prop:expected_grad} provides a guidance in choosing hyperparameter $\lambda$ as
\ba{
\lambda \in \left(\frac{\norm{\betav^*}_2^2 + \sigma^2}{n}, \frac{n-1}{n} (\eta - \frac{1}{n}) \min_{j \in \mathcal{A}}{(\beta_j^*)^2} \right).
\label{eq:lambda_region}
}
Though in practice the true coefficient $\betav^*$ is unknown a priori, choosing $\lambda = \log(n)/(2n)$ as BIC falls in the region \eqref{eq:lambda_region} asymptotically, and serves as a good initial point for the cross validation in the finite sample case. %
Now, we study the convergence rate of the updates of Algorithm~\ref{alg:M1} in expectation, namely, with precise gradient each step. We show that these updates converge to the ground truth after $\cO(1/\epsilon)$ steps where $\epsilon > 0$ is the desired accuracy.
\begin{theorem}
\label{theorem:rate_population}
Let the update be $\phiv^{(t+1)}=\phiv^{(t)} - \rho \E_{\Xmat, \yv, \uv}[g_{\emph{\text{U2G}}}(\uv; \sigma(\phiv^{(t)}))]$ where $\rho > 0$ is the given step size. %
We assume that $\lambda \in \mathcal{I}$ where $\mathcal{I}$ is defined as in Eq.~\eqref{eq:lambda_region} %
Furthermore, the initialization $\phiv^{(0)}$ satisfies 
that
$\sum_{j=1}^p \sigma(\phi_j^{(0)}) \leq \varpi - S$ where $\varpi$ is defined in Eq.~\eqref{eq:condition_initialization}. Then, the following holds:
\begin{itemize}
\item[(a)] For any $j \in \cA$, as long as $t \geq t^1_j$ %
\begin{align*}
   \parenth{ 1 - \sigma(\phi_{j}^{(t)})}\brackets{ 1 - c_{1} \sigma^2(\phi_{j}^{(t)}) \parenth{1 - \sigma(\phi_{j}^{(t)})}} \leq 1 - \sigma(\phi_{j}^{(t+1)}) \leq \parenth{ 1 - \sigma(\phi_{j}^{(t)})}  & \\
   & \hspace{-14 em} \times \brackets{ 1 - C_{1}(\sigma(\phi_{j}^{(t)}))^2 \parenth{1 - \sigma(\phi_{j}^{(t)})}}.
\end{align*}
\item[(b)] For any $j \notin \cA$ and $t \geq t^2_j$ %
\begin{align*}
   \sigma(\phi_{j}^{(t)}) \brackets{ 1 - c_{2} \parenth{1 - \sigma(\phi_{j}^{(t)})}^2 \sigma(\phi_{j}^{(t)})} \leq \sigma(\phi_{j}^{(t+1)}) \leq   \sigma(\phi_{j}^{(t)})  & \\
   & \hspace{-8 em} \times \brackets{ 1 - C_{2} \parenth{1 - \sigma(\phi_{j}^{(t)})}^2 \sigma(\phi_{j}^{(t)})}.
\end{align*}
\end{itemize}
Here, with model parameters $\tau = \{ n, p, \sigma, \betav^{*}\}$, $c_{1}, c_{2}, C_{1}, C_{2}$ are some positive constants depending only on $\tau$ and $\rho$. $t_j^1$ and $t_j^2$ are constants depending on $\tau$, $\rho$ and initial $\phi_j^{(0)}$.
\end{theorem}
\begin{remark}
(i) The upper bounds of Theorem~\ref{theorem:rate_population} demonstrate that when $j \in \cA$, $1 - \sigma(\phi_{j}^{(t)}) \leq \epsilon$ after $t = \mathcal{O}(\epsilon^{-1})$ steps, which is \emph{sub-linear}. Similarly, when $j \notin \cA$, it takes $t = \mathcal{O}(\epsilon^{-1})$ number of iterations for $\sigma(\phi_{j}^{(t)})$ to be within $\epsilon$ radius from 0. The lower bounds in Theorem~\ref{theorem:rate_population} indicate that these sub-linear complexities are tight. As a consequence, in expectation, the updates of Algorithm~\ref{alg:M1} converge to the global optima at the sub-linear rate $\mathcal{O}(\epsilon^{-1})$.

(ii) The results of Theorem~\ref{theorem:rate_population} also yield an insight into the choice of step size $\rho$. %
Based on the specific forms of $c_{1}$ and $c_{2}$ in the proof, we need the step size $\rho$ to satisfy
\begin{align}
    \rho < \min \left\{\frac{2}{\lambda}, \frac{2}{ \max_j\{(\beta_{j}^{*})^2\} - \lambda + (\sigma^2 + \|\betav^{*}\|_{2}^2)/n}\right\}. \label{eq:range_stepsize}
\end{align}
\end{remark}
The convergence properties we present in this section are under the expectation. The empirical performance of a low variance gradient estimator such as U2G can be close to the theoretical results, as shown in Section~\ref{sec:experiemnt}. Before that, we extend the proposed gradient methods from solving the frequentist objective \eqref{eq:bss3} to solving the $L_0$-regularized regression in the Bayesian paradigm.

\section{Bayesian \texorpdfstring{$L_0$}~-Regularized Regression}
\label{sec:vb}
The objective function in \Cref{eq:objective} and the gradient estimators in \Cref{sec:estimators} are compatible with a general objective function $f(\zv)$ with high dimensional binary vector $\zv$. 
In this section, 
we consider the best subset selection as a posterior inference problem, and use the  gradient estimators in \Cref{sec:estimators} to solve the new objective function. 

Consider the Bayesian linear regression with the spike-and-slab prior \citep{mitchell1988bayesian,george1997approaches}, a probabilistic model  with  the likelihood and prior as
\begin{gather}
y_i \sim \cN(\xv_i^{\top} (\alphav \odot \zv), \sigma^2), ~i \in [n] \notag \\ 
 \alphav \sim \cN(\alphav; \mathbf{0}, \Sigmamat_{\alphav}), \quad z_j \sim \text{Bern}(\sigma(-\lambda_0)), ~j \in [p].
\label{eq:bayes}
\end{gather}
Above, $\alphav \in \bR^p$ and $\zv \in \{0,1\}^p$ are latent variables with Gaussian and Bernoulli prior, respectively, and $\odot$ is the element-wise product. The hyper-parameter $\lambda_0$ controls the level of sparsity. Setting the regression variable $\betav = \alphav \odot \zv$, the prior for $\betav$ is a spike-and-slab prior which has a slab Gaussian component and a spike component at~$0$,
\bas{
p(\betav) = \prod_{j=1}^p \big[ \sigma(\lambda_0) \delta_0 + (1-\sigma(\lambda_0))\cN(0, \sigma_\alpha^2) \big].
}

We consider the \emph{maximum a posterior} (MAP) estimator for the best subset selection. The posterior distribution for the latent variable model in \Cref{eq:bayes} is
\bas{
\resizebox{0.92\hsize}{!}{
$p(\alphav, \zv \given \Xmat, \yv; \lambda_0, \Sigmamat_{\alphav}) \propto \exp \Big(-(2\sigma^2)^{-1}\norm{\yv - \Xmat (\alphav\odot \zv)}_2^2 -  \alphav^{\top} \Sigmamat_{\alphav}^{-1} \alphav/2 - \lambda_0 \norm{\zv}_0 \Big).$
}
}
Suppose the hyper-parameter $\Sigmamat_{\alphav}$ is $\sigma_{\alpha}^2\Imat$. The MAP estimator can be obtained  by  minimizing the negative log-posterior $- \log p(\alphav, \zv \given \Xmat, \yv; \lambda_0, \Sigmamat_{\alphav})$, $i.e.,$
\ba{
\min_{\alphav, \zv}   \frac{1}{2}\norm{\yv - \Xmat (\alphav\odot \zv)}_2^2 + \frac{\sigma^2}{2 \sigma_{\alpha}^2} \norm{\alphav}^2_2 + \sigma^2\lambda_0 \norm{\zv}_0.
\label{eq:bss_bayes}
}

By \Cref{eq:bss_bayes}, the MAP estimator of the spike-and-slab regression  is equivalent to the solution of the frequentist linear regression with additive $L_2$ and $L_0$ penalties \citep{polson2017bayesian}.  
When the variance $\sigma_\alpha^2$ of the slab component in the $\betav$ prior is large,  the ratio $\sigma^2/\sigma_\alpha^2$ is small and the MAP estimator of the Bayesian linear regression in \Cref{eq:bss_bayes} is close to the solution of the best subset selection in \Cref{eq:bss2}.

The MAP estimator is not directly computable because solving Eq.~\eqref{eq:bss_bayes} is a combinatorial problem. To overcome the computational challenge, we resort to variational inference (VI) to approximate the posterior distribution and the MAP estimator. To be consistent with VI nomenclature, here we deviate from the notation in Eq.~\eqref{eq:objective}, and use $p(\zv)$ as the prior and $q_{\phiv}(\alphav, \zv)$ as the variational distribution with parameter~$\phiv$.

The VI methods find an approximated posterior by minimizing the Kullback--Leibler (KL) divergence from $p(\alphav, \zv \given \Xmat, \yv))$ to $q_{\phiv}(\alphav, \zv)$, denoted as $D_{\kl}(q_{\phiv}(\alphav, \zv) || p(\alphav, \zv \given \Xmat, \yv))$. Since the true posterior is often unknown, equivalently we can maximize the evidence lower bound (ELBO) \citep{blei2017} as a tractable objective, defined as ~\looseness=-1
\ba{
\notag
\cL(\phiv) =&~\log p(\yv | \Xmat) - D_{\kl}(q_{\phiv}(\alphav, \zv) || p(\alphav, \zv | \Xmat, \yv)) \\
=&~\E_{q_{\phiv}(\alphav, \zv)} \log \big[p(\yv | \Xmat, \zv, \alphav)p(\alphav, \zv ; \lambda_0, \sigma_{\alpha}^2) / q_{\phiv}(\alphav, \zv) \big]. 
\label{eq:elbo-raw}
}

Due to the limited expressiveness of the variational distribution and the zero-forcing property of the KL divergence, variational method often underestimates the posterior uncertainty. Recent analysis, however, provides theoretical guarantees to the accuracy of point estimation. The consistency and asymptotic normality of the VI point estimation have been established for specific models \citep{bickel2013asymptotic, pati2017statistical, zhang2017theoretical, yin2019theoretical}. A general Bernstein-von Mises theorem has been proved that the variational posterior converges to the KL minimizer of a normal distribution, centered at the truth \citep{wang2018frequentist}. Hence, VI provides an accurate point estimate to the MAP solution, as validated in simulations in \Cref{sec:experiemnt}.

To further improve the inference accuracy, we propose a tightened ELBO that is closer to the evidence $\cL(\phiv)$. By \Cref{eq:elbo-raw}, the gap between ELBO and the evidence equals  the KL divergence  from  the posterior to the variational distribution.
We can then reduce the gap by controlling this KL divergence. 
With the chain rule of the KL divergence, 
\ba{
\resizebox{.9\hsize}{!}{$D_{\kl}(q_{\phiv}(\alphav, \zv) || p(\alphav, \zv |\Xmat, \yv)) = D_{\kl}(q_{\phiv}(\zv) || p( \zv | \Xmat, \yv)) + \bE_{q(\zv)} D_{\kl}(q(\alphav | \zv) || p(\alphav | \Xmat, \zv, \yv))$.}
\label{eq:chain-rule}
}
If we choose $q(\alphav | \zv) = p(\alphav | \Xmat, \zv,\yv) $,  the second term on the RHS of Eq.~\eqref{eq:chain-rule} becomes 0. Marginalizing out the latent variable $\alphav$, we get a tightened ELBO as
\ba{
\cL(\phiv) = \E_{q_{\phiv}(\zv)} \log \Big[p(\yv | \Xmat, \zv; \sigma_{\alpha}^2) p(\zv ; \lambda_0)/q_{\phiv}(\zv)\Big].
\label{eq:elbo}
}
To maximize \Cref{eq:elbo}, we choose a mean-field  distribution  $q_{\phiv}(\zv)=\prod_{j=1}^{p}\text{Bern}(z_j;\sigma(\phi_j))$. %

The ELBO in \Cref{eq:elbo} is a special case of the general optimization objective in  \Cref{eq:objective} with $f(\zv) = \log [p(\yv | \Xmat, \zv; \sigma_{\alpha}^2) p(\zv ; \lambda_{0})/q_{\phiv}(\zv)]$.
Accordingly, the unbiased gradient estimators in \Cref{sec:estimators} can be directly applied to maximizing the ELBO in \Cref{eq:elbo}. The variational objective, compared to the frequentist objective in Eq.~\eqref{eq:bss3}, does not require computing an OLS solution when evaluating $f(\zv)$, thus improving efficiency, especially when $n$ is large.

\section{Experimental Results}
\label{sec:experiemnt}

In this section, we study the performance of the gradient-based methods on a variety of synthetic and semi-synthetic data sets. Codes for the simulations in this paper are available at \url{https://github.com/mingzhang-yin/Probabilistic-Best-Subset}. \\

\noindent \textbf{Measurement Metrics.~~~} Denote  $\widehat{\betav}$ as the estimated coefficients, $\betav^*$ as the true coefficients, 
and $(\xv,y)$ as a random sample from the population.  We use the population SNR to measure the level of information in data, defined as
\bas{
\text{SNR} \coloneqq \var(\xv^{\top} \betav^*)/\var(\epsilon) = {\betav^*}^{\top} \Sigmamat \betav^* / \sigma^2.
}
The population SNR describes the degree of signal in the data generation. In addition to the SNR,  the degree of challenge of variable selection is influenced by the number of data $n$, the number of the covariates $p$ and the size of active set $S$ \citep{hazimeh2018fast}. 

The evaluation metrics throughout can be categorized into two groups: one group of metrics measures the out-of-sample predictive performance and the other group measures the recovery quality of the sparsity pattern \citep{bertsimas2016best, hastie2020best}. The metrics for the predictive performance that we use are
\begin{itemize}
\item \emph{ Relative risk (RR)} that measures how model prediction deviates from the oracle prediction, the perfect score being 0: 
\bas{
\text{RR}(\widehat{\betav}) = \frac{\E(\xv^{\top} \widehat{\betav} - \xv^{\top} \betav^*)^2}{\E(\xv^{\top} \betav^*)^2} = \frac{(\widehat{\betav} - \betav^*)^{\top} \Sigma (\widehat{\betav} - \betav^*)}{{\betav^*}^{\top} \Sigma \betav^*}.
}
\item \emph{ Relative test error (RTE)} that measures the relative test MSE compared with the oracle predictor, the perfect score being 1:
\bas{
\text{RTE}(\widehat{\betav}) = \frac{\E(y- \xv^{\top} \widehat{\betav})^2}{\E(y- \xv^{\top} \betav^*)^2} = \frac{(\widehat{\betav} - \betav^*)^{\top} \Sigma (\widehat{\betav} - \betav^*) + \sigma^2}{\sigma^2}.
}
\item \emph{Proportion of variance explained (PVE)} that measures the proportion of variance in the response variable explained by the model, the perfect score being $\text{SNR}$/(1 + $\text{SNR}$):
\bas{
\text{PVE}(\widehat{\betav}) = 1 - \frac{\E(y- \xv^T \widehat{\betav})^2}{\var(y)} = 1 - \frac{(\widehat{\betav} - \betav^*)^T \Sigma (\widehat{\betav} - \betav^*)+\sigma^2}{{\betav^*}^T\Sigma \betav^*+\sigma^2}.
}
\end{itemize}

To evaluate the sparse pattern recovery, we consider the \emph{size of estimated active set} \citep{linero2018bayesian} as well as the \emph{precision}, \emph{recall}, and \emph{F1} scores,  given by prec = TP/(TP + FP), rec = TP/(TP +FN), and F1 = 2 $\cdot$ prec $\cdot$ rec/(prec + rec), where TP denotes the number of predictors correctly flagged as influential, FP denotes the number of predictors incorrectly flagged as influential, and FN denotes the number of predictors incorrectly flagged as noninfluential. The F1 score is an overall summary that balances the precision and recall.\\

\noindent \textbf{Implementation Details.~~~} We compare the proposed gradient-based methods with several representative sparse variable selection methods. In particular, we consider \textsc{Lasso} \citep{tibshirani1996regression}, a convex penalty regularized method, SCAD \citep{fan2001variable}, a nonconvex penalty regularized method, and MIO \citep{bertsimas2016best}, Fast-BSS \citep{hazimeh2018fast}, two best subset selection methods. 

For the methods in comparison, \textsc{Lasso} is implemented by R package \texttt{glmnet} \citep{friedman2010regularization}. SCAD is implemented by R package \texttt{ncvreg} \citep{ncvreg}. We use the R package \texttt{bestsubset} \citep{bestsubset2018r} for the best subset selection with MIO \citep{bertsimas2016best} and use the R package \texttt{L0Learn} \citep{hazimeh2022l0learn} for the Fast-BSS. %
 The regression functions in \texttt{glmnet}, \texttt{ncvreg}, \texttt{L0Learn} packages fit the regularization hyperparameters over a path of  100 values in default. 
 If not specified, we use the default configurations of the existing R packages. %

For the proposed gradient-based methods, we set the number of Monte Carlo samples for estimating the gradient $\bE[\gv(\uv;\sigma(\phiv))]$ in Eq.~\eqref{eq:grads-multi-mc} as $K=20$ throughout the experiments. The gradient-based algorithms take less than 20 seconds to converge when the number of covariates is in thousands, running on a MacBook Pro laptop with a 2.4GHz GHz CPU. %
To determine the convergence, we compute the entropy for the $j$-th covariate as $H_j = -\pi_j \log(\pi_j)$ where $\pi_j = \sigma(\phi_j)$ is the probability of the Bernoulli distribution defined in Eq.~\eqref{eq:bss2}. We stop the training when the average of the $5\%$ largest entropies is below $0.1$. Under this stopping criterion,  all the probabilities $\pi_j$ are close to either zero or one. By the analysis in Eq.~\eqref{Sec:convergence}, we choose the hyperparameter $\lambda$ for the gradient-based methods on a grid of values  starting from $\log(n)/(2n)$, and use a constant step-size in SGD smaller than $2/\lambda$. 

We choose the hyperparameter for all methods by cross validation. The data is randomly split into training, validation and test sets. We report the results corresponding to the hyperparameter that leads to the lowest prediction error on the validation sets.

\subsection*{Experiment 1: Synthetic Data with Correlated Covariates}

We consider the example in \citet{fan2001variable} with increased dimension. The true coefficient  $$\betav^* = (3, 1.5, 0, 0, 2, \underbrace{0, \cdots, 0}_\text{195}) \in \bR^{200}.$$ The design matrix $\Xmat$ are $n$ i.i.d. samples generated from $\cN(\mathbf{0}, \Sigma)$ where $\Sigma_{ij} = \rho^{|i-j|}$ with the correlation parameter $\rho \in (0,1)$, and $\yv \sim \cN(\Xmat\betav^*, \sigma^2\Imat)$. $n=60$ in this example.

We first compare all the considered methods under a high and a low $\text{SNR}$ regime by setting the standard deviation of the noise as $\sigma = 1$ and $\sigma = 3$. As shown in \Cref{tab:exp1}, the non-$L_0$-based methods tend to select larger active sets than the $L_0$-regularized regression, reflected as a high recall and a low precision. Furthermore, 
 the best subset methods have lower test error than the non-$L_0$-based methods for both SNRs, possibly because the $L_0$ penalty has no shrinkage effect on the magnitude of the coefficients. For the compared best subset methods,  fast-BSS has much higher computational efficiency than MIO due to the cyclic coordinate descent and has better accuracy, which is similarly observed in \citet{hastie2020best}. For the gradient-based methods, U2G and U2G-VI perform on par with or better than ARM$_0$ and ARM$_0$-VI, while REINFORCE estimator has a high error in prediction and estimation because of high gradient variance. Based on this observation, we further compare \textsc{LASSO}, SCAD, Fast-BSS, U2G and U2G-VI systematically on a set of extensive experiments. ~\looseness=-1

\begin{table}[ht]
\centering
\caption{\small Results of Experiment 1 with $n=60, p=200, S=3$ and $\rho=0.5$. Reported results are the mean of 100 independent trials.} 
\begin{tabular}{cccccccc}
\toprule
&Precision&Recall& F1&Nonzero& RR& RTE&PVE \\
 \midrule
\multicolumn{3}{l}{ \small{$n = 60, p = 200, \sigma =1, \text{SNR} = 21.3$}}&&&& \\
 \textsc{Lasso} &0.780& 1.000& 0.852&4.65& 0.039& 1.830 &0.918\\

 SCAD &0.983& 1.000& 0.990& 3.07&0.013& 1.271 &0.943 \\ 

 MIO & 1.000 & 1.000 & 1.000 & 3.00 & 0.003 & 1.056 & 0.952 \\
 
  Fast-BSS & 1.000 & 1.000 & 1.000 & 3.00 &  0.003 & 1.055 & 0.952 \\
  
 REINFORCE & 0.089 & 0.657 & 0.153 & 32.4 & 1.601 & 35.01 & - \\

 ARM$_0$ & 0.992 & 1.000 & 0.996 & 3.03 & 0.003 & 1.067 & 0.952 \\

 U2G & 0.990 & 1.000 & 0.994 & 3.04 & 0.003 & 1.069 & 0.952 \\

 ARM$_0$(VI) & 0.950 & 1.000 & 0.971 & 3.21 & 0.005 & 1.107 & 0.950 \\

 U2G(VI) & 0.950 & 1.000 & 0.971 & 3.21 & 0.005 & 1.107 & 0.95 \\ 
\midrule
\multicolumn{3}{l}{ \small{$n = 60, p = 200, \sigma =3, \text{SNR} = 2.4$}}&&&& \\
 \textsc{Lasso} &0.747&0.850&0.745&4.32&0.284&1.671&0.503\\

 SCAD &0.722&0.777&0.721&3.51&0.214&1.506&0.552\\
 
 MIO &0.780&0.780&0.780&3.00&0.125&1.294&0.615\\
 
 Fast-BSS & 0.953 &   0.680 &  0.787 &  2.16 &  0.121  &  1.292  &  0.619 \\
 
 REINFORCE & 0.092 & 0.503 & 0.151 & 20.2 & 1.092 & 3.579 & - \\

 ARM$_0$ & 0.856 & 0.863 & 0.844 & 3.15 & 0.107 & 1.251 & 0.628 \\

 U2G& 0.971 & 0.850      & 0.896 & 2.65      & 0.063 &
       1.174& 0.690 \\

 ARM$_0$(VI) & 0.921 & 0.890 & 0.889 & 2.96 & 0.081 & 1.190 & 0.646 \\

 U2G(VI) & 0.913 & 0.883 & 0.885 & 3.00   &      0.068 & 1.190 & 0.686 \\
\bottomrule
\end{tabular}
\label{tab:exp1}
\end{table}

In \Cref{fig:EXP1}, we study how the F1 score and relative risk (RR) change with the SNR and the covariate correlation. These metrics reflect the accuracy in estimating the active set and in predicting the outcome of the test data, respectively. We fix the correlation parameter $\rho=0.5$ and sweep SNR between 1 and 10. Then, we fix $\text{SNR}=3$ and sweep $\rho$ between 0 and 0.8. ~\looseness=-1

For the active set recovery, we find the $L_0$-based methods generally have higher F1 scores than non-$L_0$-based methods across all the settings. Among the $L_0$-based methods, U2G-VI has the highest F1 score for most SNR and $\rho$. The F1 score of Fast-BSS drops fast when the covariate correlation $\rho$ increases. In contrast, the gradient-based methods are more robust to  high covariate correlation. For the prediction at the test time, we find that when the SNR is below 1.5, non-$L_0$-based methods have lower RR than U2G and U2G-VI. We hypothesize that when the SNR is very low, selecting a large set of predictors compensates for the error in the coefficient estimation \citep{hastie2020best}. A similar phenomenon has been observed that $L_0$ regularization tends to overfit when the SNR is very low \citep{mazumder2017subset}. For the SNR larger than 1.5 and across different $\rho$, U2G and U2G-VI have the lowest RR. 

\begin{figure}[ht]
	\centering{
		{{\includegraphics[width=0.47\textwidth]{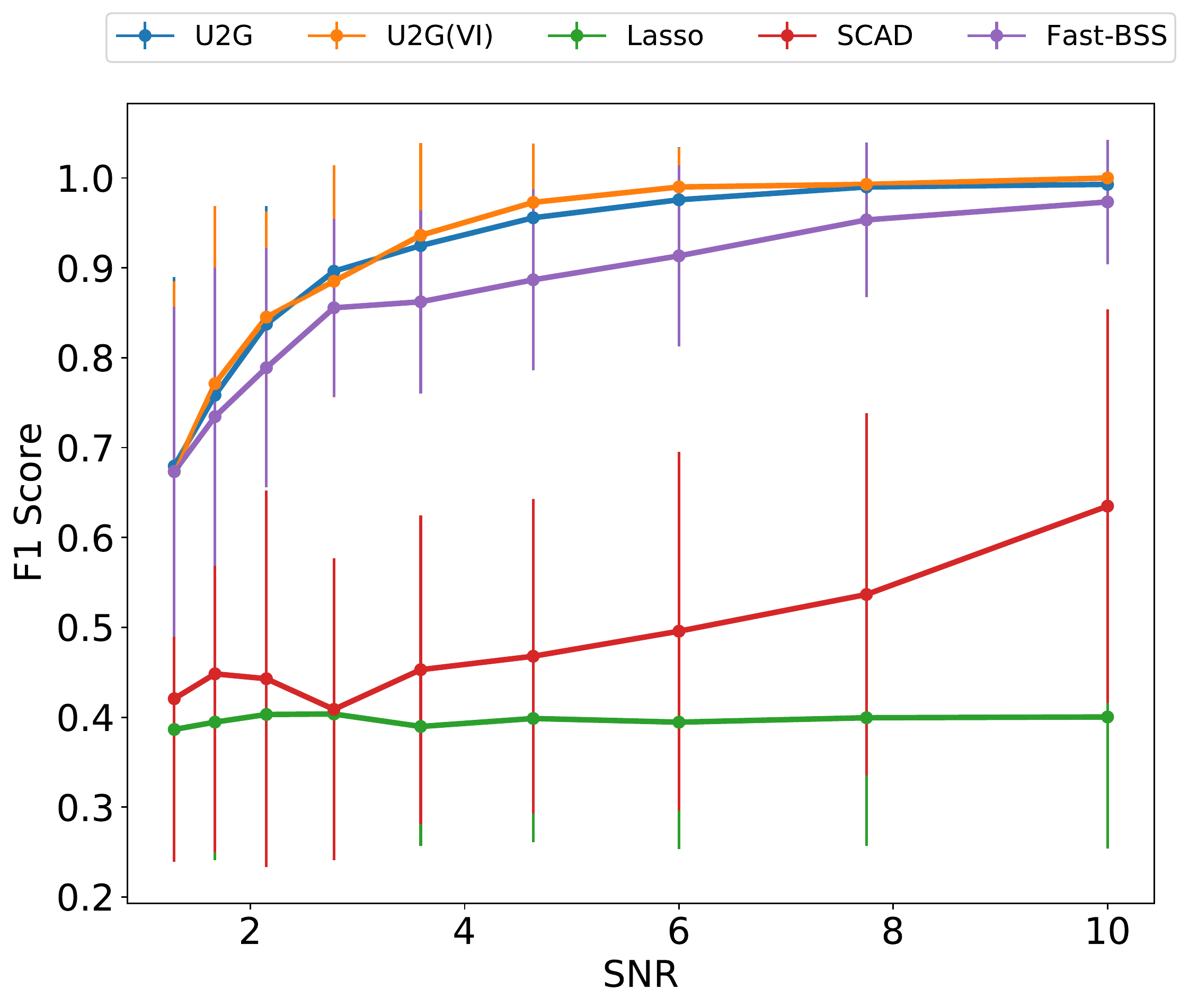} }} ~
		{{\includegraphics[width=0.48\textwidth]{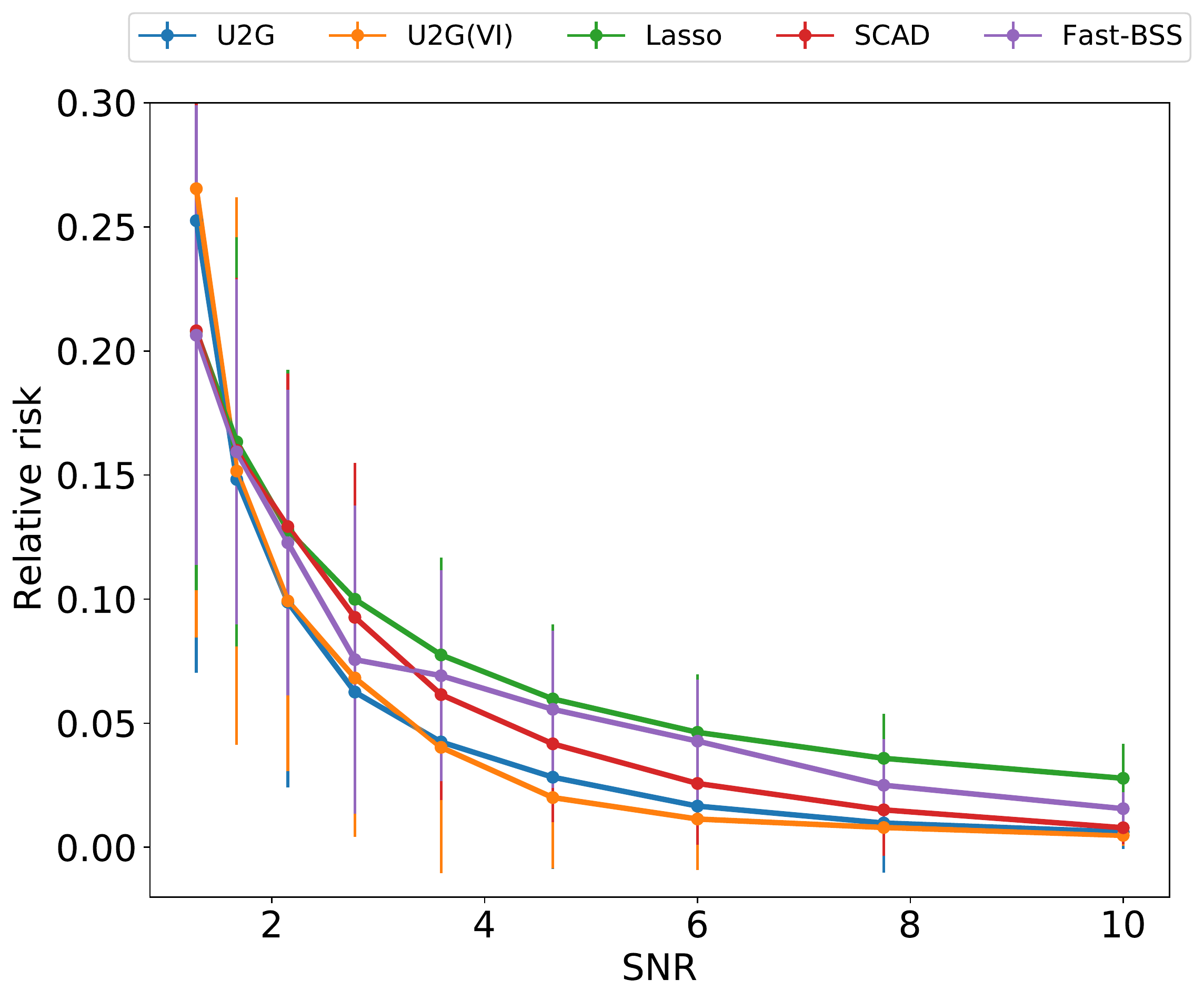} }}  \\
		
		{{\includegraphics[width=0.47\textwidth]{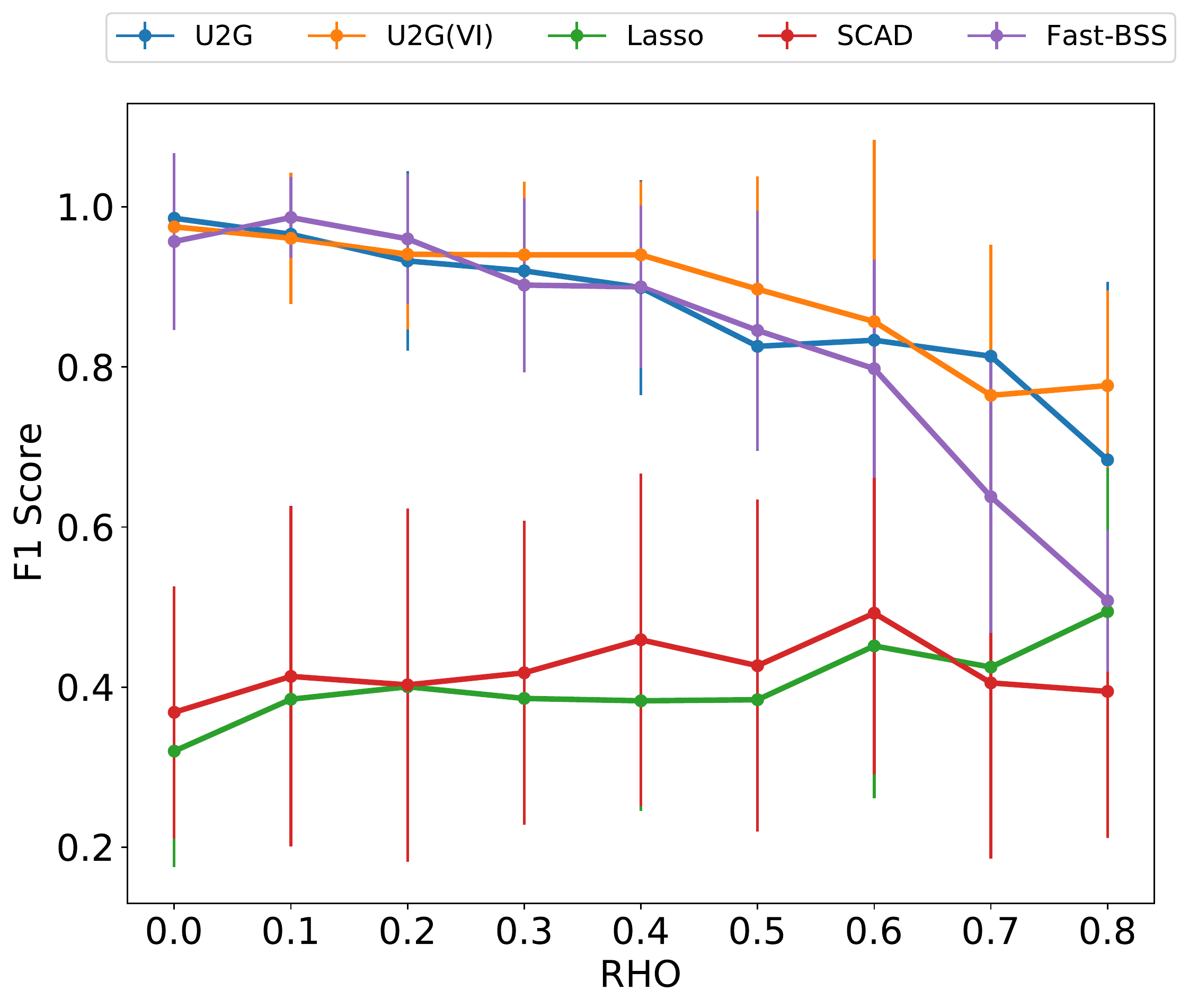} }}  ~~
		{{\includegraphics[width=0.48 \textwidth]{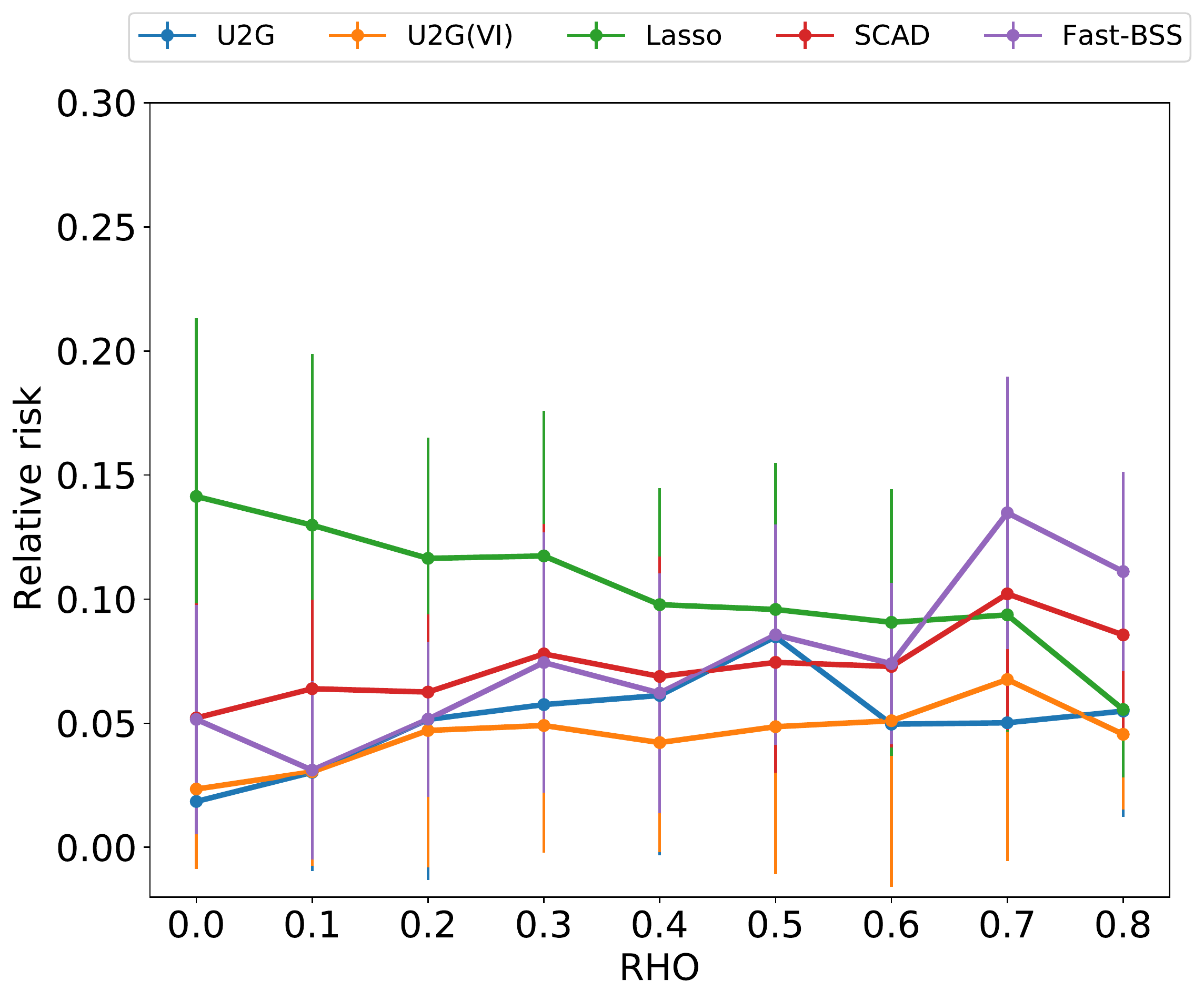} }} 
		\caption{  \small  Results for Experiment 1. \textbf{Top}: Variable selection across $\text{SNR}$ with $\rho=0.5$. The best subset methods based on $L_0$ penalty has significantly higher F1 scores in finding the true predictor set. For $\text{SNR} \geq 1.5$, U2G, U2G-VI have the highest F1 scores and the lowest prediction risks. \textbf{Bottom}: Variable selection across covariates correlation $\rho$ with $\text{SNR}=3$. The F1 scores of the best subset methods decrease with an increasing $\rho$ but are higher than that of SCAD and \textsc{LASSO} for all $\rho$. U2G and U2G-VI have the highest F1 scores in the high correlation regime.  U2G-VI has the lowest predictive risk for most of the $\rho$ values.}
		\label{fig:EXP1}
	}
\end{figure}

We show the regularization path of $L_0$
regression in Figure~\ref{fig:regu_path}, with $n=60, p=200, \sigma=1$, and $\rho=0$. When $\lambda$ decreases, the number of selected variable increases. The test error first decreases when the correct covariates join the selection, and then increases as additional incorrect covariates are selected. As the top panel shows, for a wide range of $\lambda$ values, the $L_0$-regularized regression recovers the true active set and has an estimated coefficient close to its true value without shrinkage.

\begin{figure}[ht]
\centering
\includegraphics[width=0.8\textwidth]{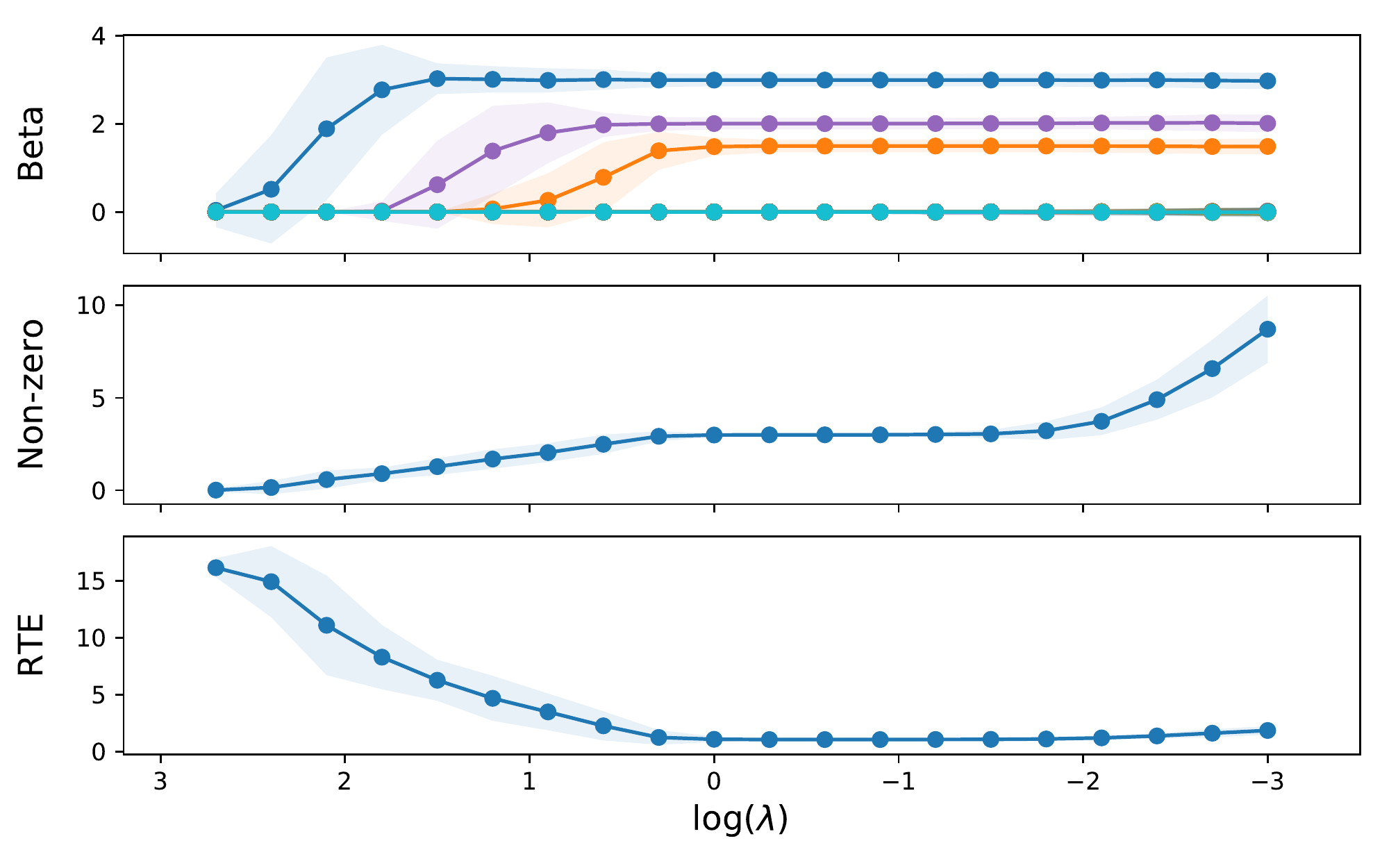}
\caption{\small Regularized path for $L_0$-regularized regression estimated by U2G gradient, with $n=60, p=200, \sigma=1$. The dotted curves are the mean of 100 independent trials and the shaded areas represent the standard deviation. }
\label{fig:regu_path}
\end{figure}

\vspace{10pt}
\noindent\textbf{\large Experiment 2: Synthetic Data with Independent Covariates} 
\vspace{5pt}

We consider the experiment in \citet{bertsimas2016best} and \citet{ hastie2020best}. The true coefficients have the first $10$ elements equal to $1$ as $\betav^* = (\underbrace{1, \ldots, 1}_{\text{10}}, \underbrace{0, \ldots, 0}_{\text{990}})$ and $p=1000$, $S=10$. The covariates $\xv_i \in \bR^{1000}, i \in [n]$ are sampled i.i.d. from a zero mean isotropic Gaussian distribution. $\text{SNR}=5$ in this example.

Table \ref{tab:exp2} contains the numerical results for the simulations. \textsc{LASSO} produces a large active set and high RR. SCAD has low prediction error but estimates an excessively large active set. Fast-BSS has high precision and recall but low prediction accuracy. In comparison, U2G and U2G-VI perform well for both target prediciton and set estimation.

\begin{table}[ht]
\centering
\caption{ \small Results of Experiment 2 with $n=100$, $p=1000$, $S=10$, $\text{SNR}=5$. Reported results are the mean of 100 independent trials.}  \label{tab:exp2} 
\begin{tabular}{cccccccc}
\toprule
&Precision&Recall& F1&Nonzero& RR& RTE&PVE \\
 \midrule

 \textsc{Lasso} &    0.162 &  0.995&  0.277 & 63.70  & 0.259 &  2.299 &  0.616 \\
   
 SCAD  &    0.272  &  1.000  &  0.422 &  39.10  &  0.050  &  1.245  &  0.792  \\

 Fast-BSS & 0.945  &  0.940  &  0.942  &  9.90 &  0.109  &  1.547  &  0.742\\

 U2G & 0.896 &  0.980   &      0.934  &  11.05  &       0.078  &  1.391 &  0.768 \\

 U2G(VI)&  0.923 &  1.000     &      0.950 &  10.90  &        0.050 &  1.256 &  0.791 \\
\bottomrule
\end{tabular}
\end{table}

\Cref{fig:EXP2} explores the influence of the  sample size $N$. For fixed SNR and dimension $p$, the sample size reflects the level of information contained in the observed data. For small $N$, the non-$L_0$-based methods have a higher F1 score and a lower RR. This indicates that the non-$L_0$-based methods are less affected by the scarcity of data, potentially because of the relaxation in the sparsity penalty. When $N$ increases, the best subset methods outperform the  non-$L_0$-based methods on the F1 score and are on par with SCAD on the RR. In Appendix \Cref{fig:simu2-K}, we study how the number of samples $K$ in estimating the gradient in Eq.~\eqref{eq:grads-multi-mc} influences the performance of U2G and  U2G-VI. We find the performance of U2G improves when $K$ increases from 1 to 10 and stays similar when $K$ further increases. The performance of U2G-VI is similar across different values of $K$. 

\begin{figure}[ht]
	\centering{
		{{\includegraphics[width=0.47\textwidth]{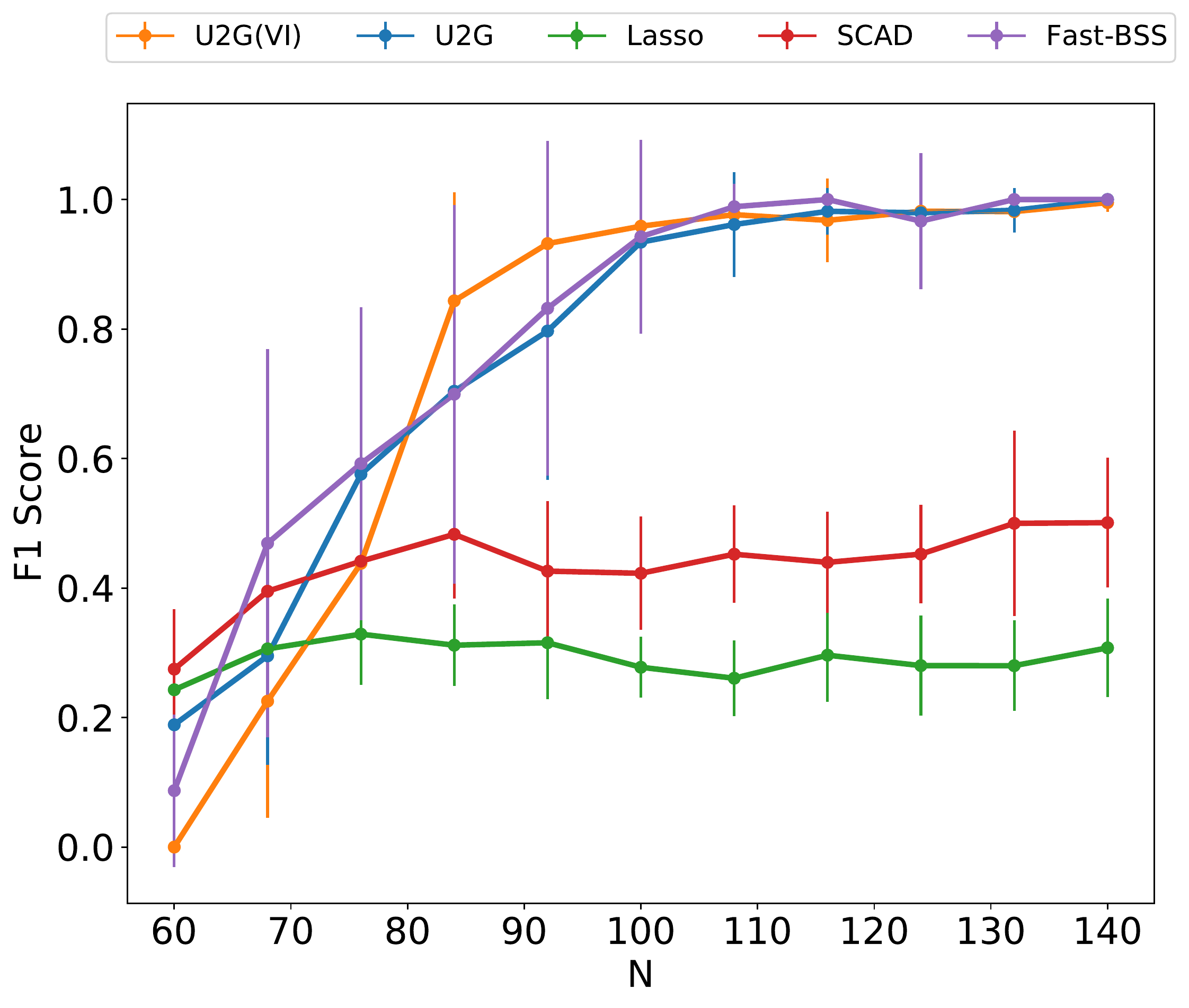} }} ~
		{{\includegraphics[width=0.47\textwidth]{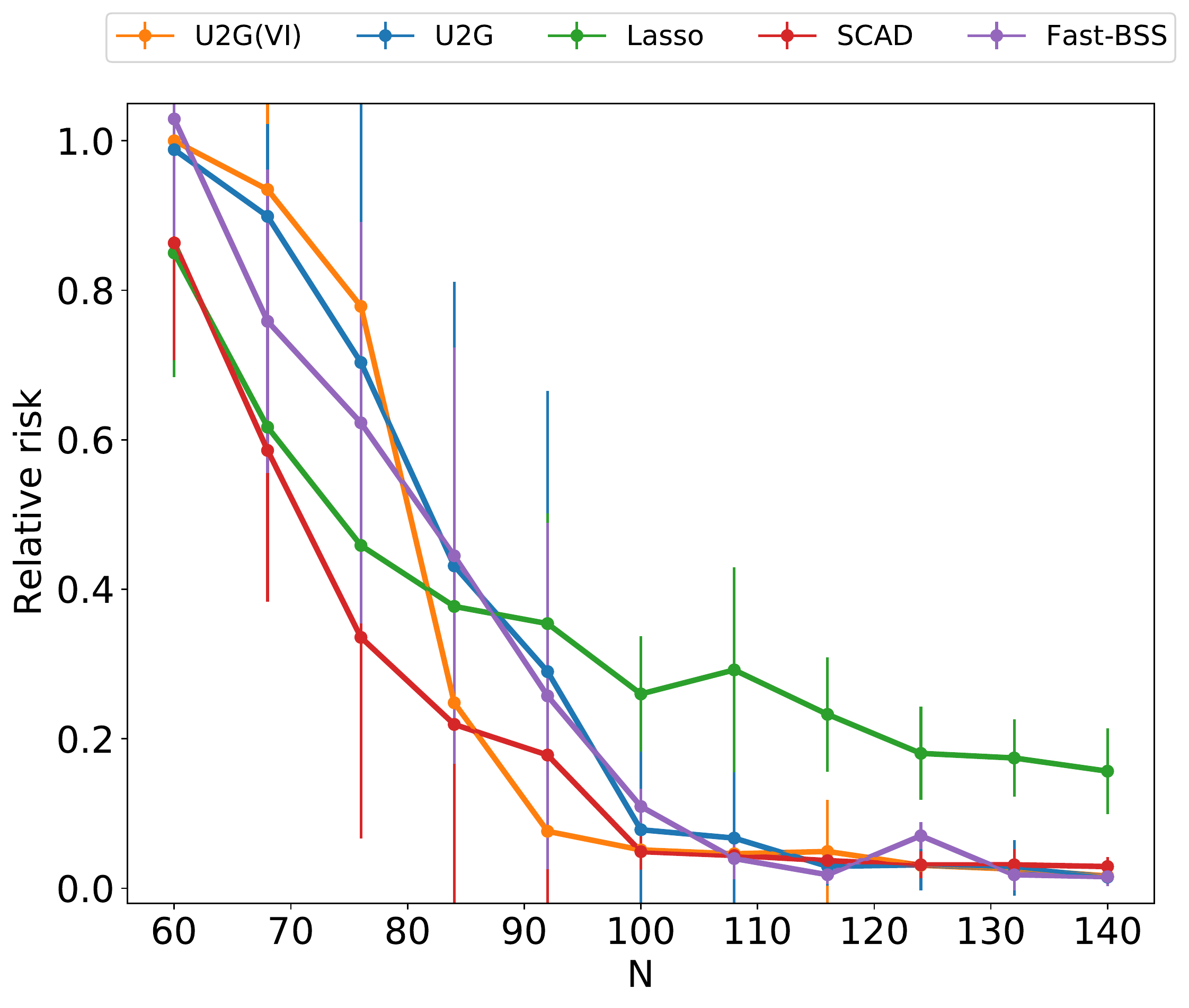} }} 
		\caption{  \small Results for Experiment 2. The change of F1 score and RR with the number of samples with $\text{SNR}=5$. For a small value of $n$, non-$L_0$-based methods have higher F1 scores and lower RR. For relatively large $N$ the F1 scores of U2G and U2G-VI are close to 1 and xare much higher than those of the non-$L_0$-based methods. When $N \geq 100$, the RRs of all compared methods  except \textsc{LASSO} are similarly low. }
		\label{fig:EXP2}
	}
\end{figure}

\subsection*{Experiment 3: Semi-synthetic Data}

We further benchmark our methods on Prostate, a real-world microarray dataset about prostate cancer \citep{singh2002gene}. The regularized models are widely used for gene selection as biomarkers in analyzing microarray data, which is often high-dimensional with a large number of genes and a small number of samples. The original Prostate dataset contains the expression profiles of 12,600 genes for 50 normal tissues and 52 prostate tumor tissues. Similar to \citet{bertsimas2016best}, we reduce the number of covariates by choosing 1000 genes that maximally correlate (in absolute value) with the tumor type. For the active set, we first choose five gene biomarkers correlated the most with the tumor type, which induces high multi-collinearity in the chosen genes. We also choose five gene biomarkers with pairwise correlation in $(-0.7,0.7)$ so that the multi-collinearity is moderate. The pairwise correlations are shown in Appendix Figure~\ref{fig:corr}. For each type of active set, we separately create a semi-synthetic data set $\yv \sim \cN(\Xmat\betav^*, \sigma^2\Imat)$, $\Xmat \in \bR^{102\times 1000}$, where the coefficients are one for the chosen covariates and zero for the other ones. %

We compare U2G and U2G-VI with \textsc{LASSO}, SCAD and Fast-BSS, as shown in Table~\ref{tab:prostate} and Appendix Table~\ref{tab:prostate2}. When the multi-collinearity is moderate in the true active set, the gradient-based method can recover the true active set with high probability. When the multi-collinearity is high,  
U2G and U2G-VI have much higher F1 scores and lower RR than the compared methods.  %

\begin{table}[t]
\centering
\caption{\small Results of the Prostate cancer dataset, with $n=102$, $p=1000$, $S = 5$, $\text{SNR}=5$, \emph{moderate} collinearity. Reported results are the average of 100 independent trials.}
\begin{tabular}{cccccccc}
\toprule
&Precision&Recall & F1&Nonzero& RR& RTE&PVE \\
 \midrule
 
 \textsc{Lasso} &  0.167 & 0.980 & 0.284 & 31.2 &  0.049 & 1.296 & 0.814 \\
 
 SCAD & 0.668 & 0.940 & 0.771 & 7.50 & 0.020 & 1.122 & 0.839 \\
 Fast-BSS & 0.823 & 0.900 & 0.853 & 5.73 & 0.037 & 1.186 & 0.802  \\

 U2G & 0.926 & 0.960  &      0.942 &  5.20     &    0.020 &  1.12
 &  0.840 \\
 U2G(VI) & 0.963 & 0.980  &      0.971 &  5.10    &     0.011 & 1.067 &  0.847 \\
\bottomrule
\end{tabular}
\label{tab:prostate}
\end{table}

\Cref{fig:EXP3} exhibits the performance over a path of SNR from 1 to 10. When the SNR is less than 2, the non-$L_0$-based methods have lower RR. When the SNR is moderately large, U2G and U2G-VI have the highest F1 score and  U2G-VI has the lowest RR. This is consistent with the observations with the synthetic data in Experiment 1.   ~\looseness=-1

\begin{figure}[ht]
	\centering{
		{{\includegraphics[width=0.47\textwidth]{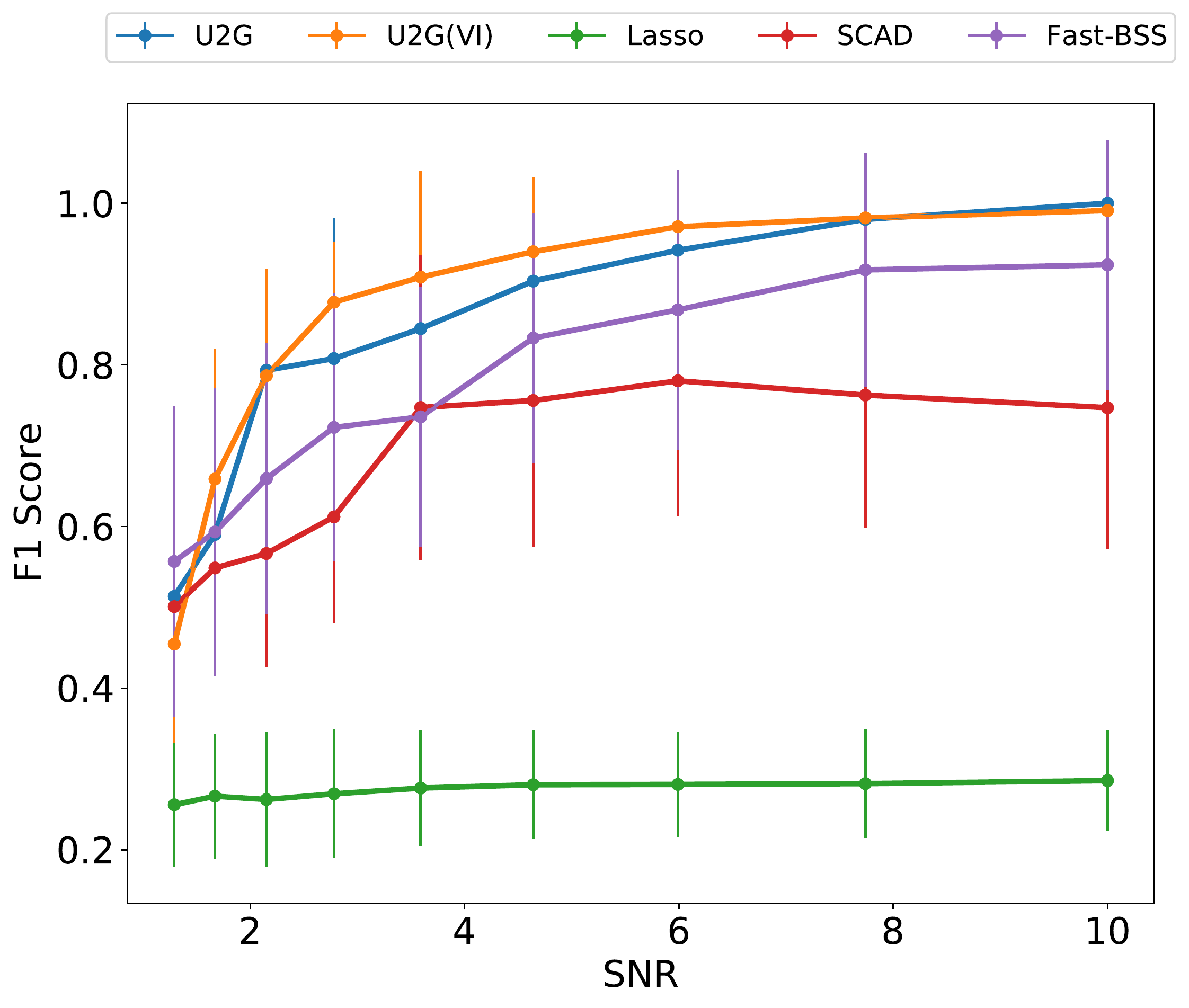} }} ~
		{{\includegraphics[width=0.48\textwidth]{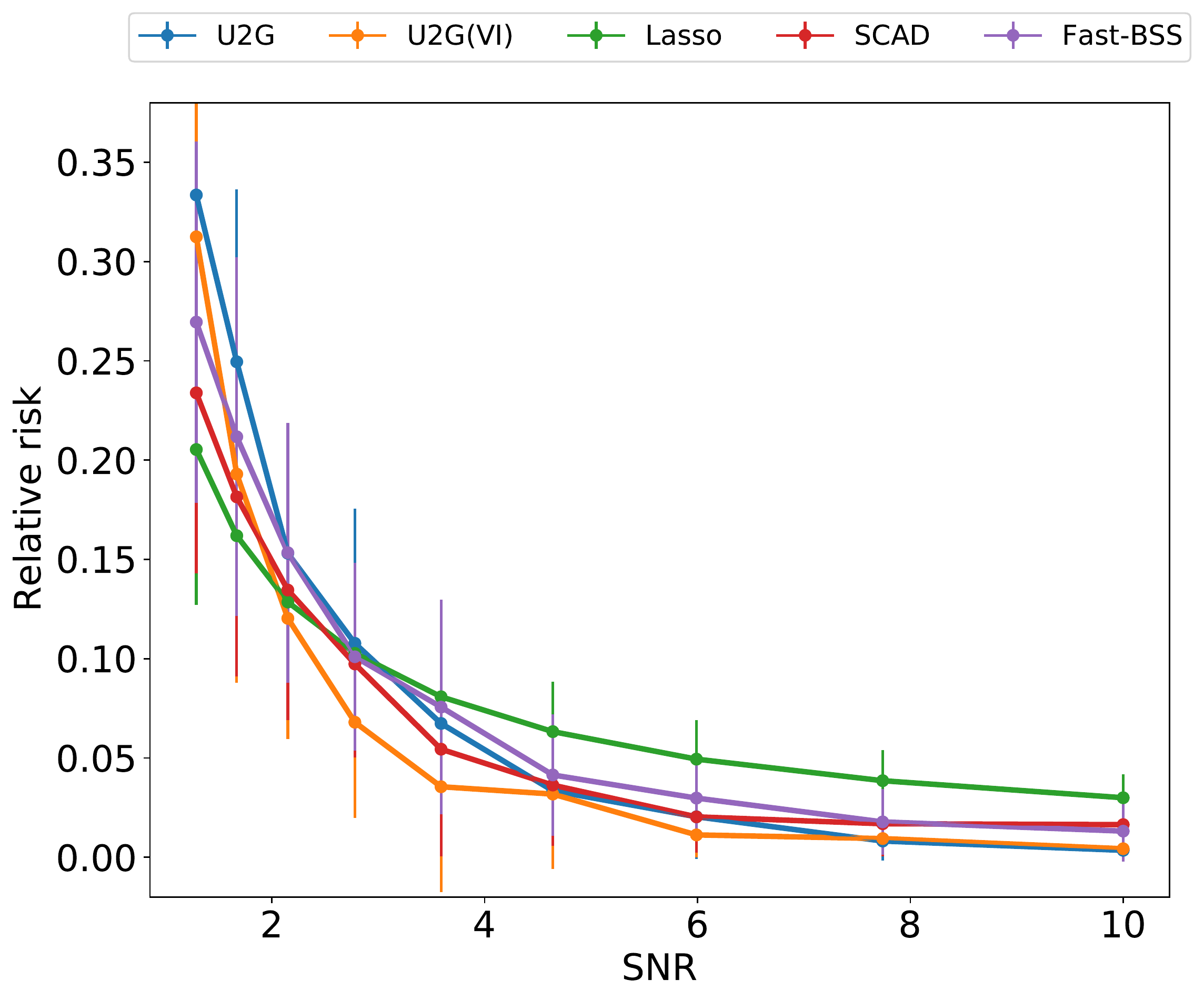} }} 
		\caption{  \small Performance measures on the Prostate cancer data with varying SNR from 1 to 10 with moderate collinearity. U2G-VI has the highest F1 score and the lowest RR when SNR$\geq 2$.}
		\label{fig:EXP3}
	}
\end{figure}

\subsection*{Experiment 4: Compressive Sensing}

We further study whether the $L_0$-based method improves the sparse signal recovery in compressive sensing. The traditional compressive sensing combines the random projection method with $L_1$-relaxation \citep{wainwright2019high}. It finds the sparse pattern of the observed signal under a set of orthonormal bases while maintaining the exact reconstruction under random projection by a measurement matrix. Following \citet{ji2008bayesian}, we consider $\thetav$ as the coordinates of observations in transformed space with length $p = 1024$, where $10$ elements are randomly picked as the signal with magnitude $\pm 1$.  In this example, most of the entries in the true signal are identically zero, which is called \emph{strong sparsity} \citep{carvalho2010horseshoe}. Construing $\Amat \in \bR^{n \times p}$ as a multiplication of the random projection matrix and orthonormal transformation matrix, each row of $\Amat$ is generated from isotropic Gaussian distribution $\cN(0,\Imat_p)$ and normalized to have the unit norm. We add the Gaussian white noise with a standard deviation $\sigma$ to the measurements $\yv$. Aligned with our probabilistic objective, we solve the Lagrangian form of 
\ba{
\min_{\thetav \in \bR^p}~\|\thetav\|_0, \qquad\text{such that}~ \Amat\thetav = \yv.
}

We compare U2G with basis pursuit (BP) \citep{chen2001atomic} and Bayesian compressive sensing (BCS) \citep{ji2008bayesian} in different $\text{SNR}$ settings by changing the magnitude of $\sigma$. 
\begin{table}[t]
\centering
\caption{\small Results of the signal reconstruction when SNR$_d$ is low, with $n=500$, $p=1000$, $S=10$, $\sigma =0.1$. %
$L_0$-based method has the best performance in recovering strong sparsity in the signal. }
\begin{tabular}{cccccccc}
\toprule
&Precision&Recall & F1&Nonzero& RR& RTE&PVE \\
 \midrule
BP &0.009&1.000&0.019&1024&0.298&298&0.702\\
BCS &0.029&1.000&0.057&336.1&2.380&2381& - \\
U2G &1.000&1.000&1.000&10.00&0.014&15.1&0.985\\
U2G(VI) & 1.000 & 1.000 & 1.000 & 10.00 & 0.018 & 18.6 & 0.981\\
\bottomrule
\end{tabular}
\label{tab:cs}
\end{table}
For U2G, we use $K = 5$ Monte Carlo samples in the gradient estimation. %
The numerical results are summarized in Table \ref{tab:cs} and Appendix Table~\ref{tab:cs2}. As shown in Appendix Figure \ref{fig:cs_high}, in the high $\text{SNR}$ regime, all three methods can reconstruct the sparse signal reasonably well, but the probabilistic best subset method can identify the locations of true signals, while BP and BCS identify excessively large active sets.  Consequently, the $L_0$-regularized method has higher predictive precision. This phenomenon is amplified when $\text{SNR}$ drops. When $\text{SNR}$ is low, as shown in Figure~\ref{fig:cs_low}, BP and BCS only recover \emph{weak sparsity} where the signals are dense yet most of the entries are small compared to several large ones.
In both high and low $\text{SNR}$ regimes, the gradient-based methods accurately recover the strong sparsity in the signal, and improve the predictive accuracy of BP and BCS by several orders of magnitude.

\begin{figure}[t]
\centering
 \includegraphics[width=0.8\textwidth]{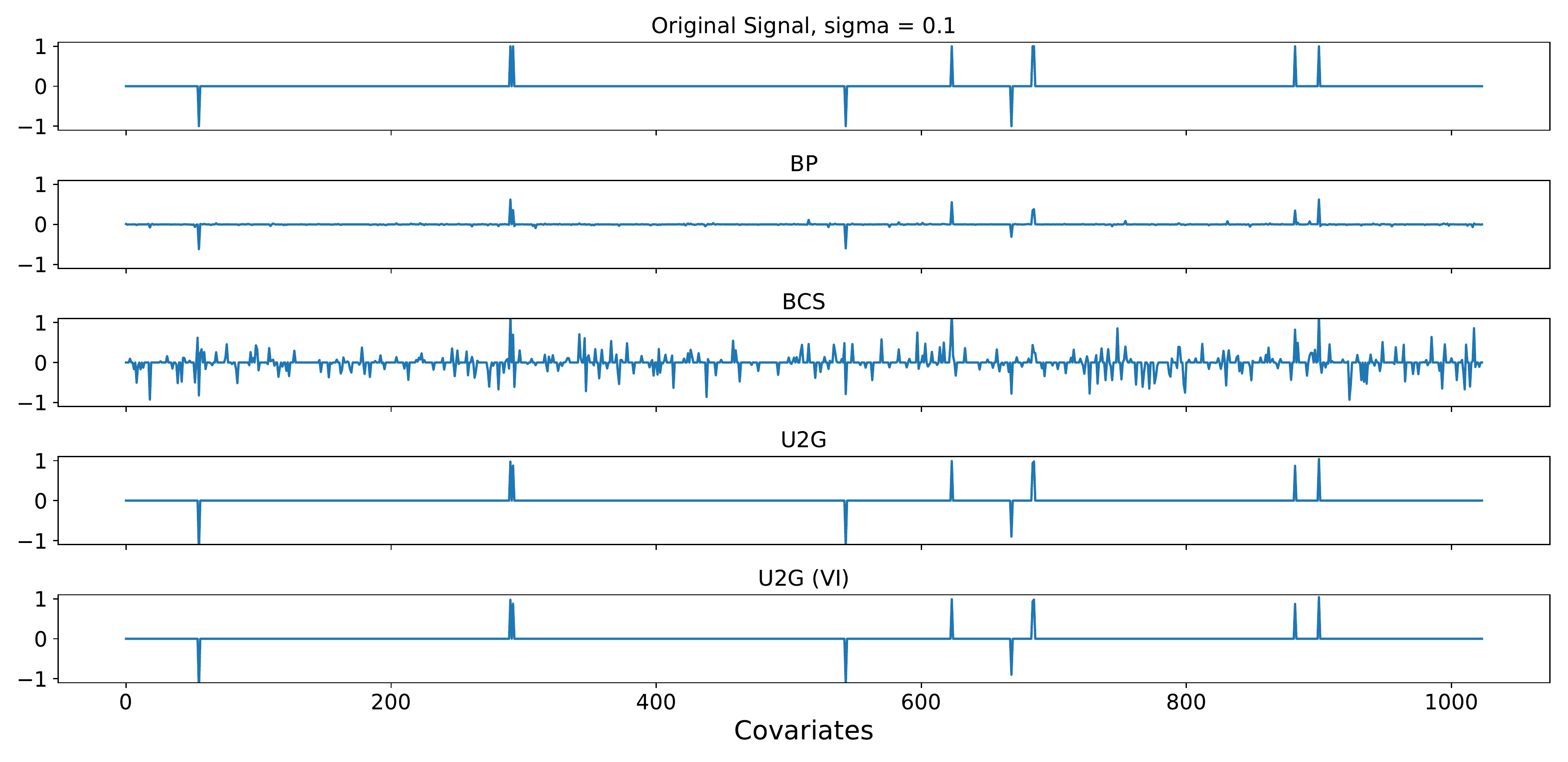}
\caption{\small The reconstruction of signals when $\text{SNR}$ is low, with $n = 500$, $p = 1024$, $\sigma = 0.1$.}
\label{fig:cs_low}
\end{figure}

\section{Discussion}
\label{sec:discussion}
We propose a probabilistic reformulation to solve the exact best subset selection problem using gradient-based optimization. In order to efficiently solve the $L_0$-regularized regression in high dimensional settings, a family of unbiased gradient estimators is proposed to approximate the exact gradient. Within this estimator family, we identify the U2G estimator as the one with minimal variance. Theoretically, the U2G estimator recovers the true sparse pattern in expectation. We also developed a variational method to solve the Bayesian best subset selection with the proposed gradient estimators. Empirically, the proposed gradient-based methods improve the sparsity estimation and predictive accuracy over convex and nonconvex relaxation methods and existing best subset selection tools. 

There are several future directions arising naturally. First, the proposed gradient-based methods are highly flexible. A future direction is to analyze the theoretical and empirical properties of the proposed best subset methods on statistical models such as the generalized linear models and deep neural networks. Second, the probabilistic reformulation in this paper is developed for binary variables. An important direction is to generalize it to categorical latent variables \citep{jang2016categorical,tucker2017rebar, yin2019arsm} and find the minimal variance estimator. Third, our theoretical analysis takes a first step with independent covariates and convergence under expectation. Future work is to study the convergence under more general design matrix assumptions.
Finally, in this paper, the unbiased gradient estimator is applied in the stochastic gradient descent framework. One future direction is to analyze the update rules with accelerated gradient and momentum  \citep{kingma2014adam,Ho_instability}. ~\looseness=-1

\appendix

\FloatBarrier

\bibliographystyle{plainnat}
\bibliography{reference}

\newpage 
\begin{center}
{\Large\bf Supplementary Material for ``Probabilistic Best Subset Selection via Gradient-Based Optimization''}
\end{center}

\vspace{3mm}
This supplementary document contains detailed proofs and derivations of the theoretical results presented in the main paper, and additional experimental results. In particular, Appendix~\ref{sec:aux_lemma} contains a necessary lemma for the proof of convergence properties. Appendix~\ref{sec:main_proof}  contains  proofs of the theoretical results presented in the main paper. Finally, Appendix~\ref{sec:additional} contains additional simulation results on the semi-synthetic data and compressive sensing. 

\section{Auxiliary Lemmas}
\label{sec:aux_lemma}
The following lemma describes the linear independence between random Gaussian vectors, which is useful when the sample size exceeds the number of covariates. Closely following the proof in \citet{taosingularity}, which contains a thorough discussion on the singularity of random matrix ensembles, we have the following result:
\begin{lemma}
\label{lemma:tao}
Let $X_j \in \bR^n$ are i.i.d. random Gaussian vectors with distribution $\cN(0,\Imat_n)$, for $j=1,\cdots,k$, $k\leq n$. Then $\{X_1, \cdots, X_k\}$ are linearly independent with probability one. 
\end{lemma}
\begin{proof}
Let event $\mathcal{E}$ be the event that $\{X_1, \cdots, X_k\}$ are linearly dependent. Then $\mathcal{E}$ is equivalent to that $X_j$ lies in the span of $X_1, \cdots, X_{j-1}$ for some $j$. Thus
\bas{
p(\mathcal{E}) \leq \sum_{j=2}^k p(X_j \in V_j),
}
where $V_j\coloneqq \text{span}(X_1,\cdots, X_{j-1})$. For each $2\leq j \leq k$, conditional on vectors $X_1, \cdots, X_{j-1}$, the vector space $V_j$ is fixed, has positive codimension, and thus has measure zero. Since the distribution of $X_j$ is absolutely continuous, and is independent of $X_1, \cdots, X_{j-1}$, we have
$$p(X_j \in V_j | X_1, \cdots, X_{j-1})=0$$
for all $(X_1, \cdots, X_{j-1})$. Integrating over $(X_1, \cdots, X_{j-1})$, we have $p(X_j \in V_j) = 0$; therefore
\bas{
p(\mathcal{E}) \leq \sum_{j=2}^k p(X_j \in V_j)=0,
}
which proves that $\{X_1, \cdots, X_k\}$ are linearly independent with probability 1. As a consequence, we obtain the conclusion of the lemma.
\end{proof}

\section{Proofs}
\label{sec:main_proof}
In this appendix, we provide proofs for theoretical results in the paper.
\subsection{Proof of Theorem \ref{thm:continuous}} 
\label{subsec:proof:thm:continuous}
\begin{proof}
On the one hand, we show the optimal solution to problem \eqref{eq:bss2} is in the set of feasible solutions to problem \eqref{eq:bss3}. Assuming $(\alphav^*, \zv^*)$ is the optimal solution to problem \eqref{eq:bss2}, setting $p_j = z^*_j,~ j \in[p]$ and $\alphav = \alphav^*$ would be a feasible solution to problem \eqref{eq:bss3} which gives the same object value as what $(\alphav^*, \zv^*)$ achieves in problem \eqref{eq:bss2}.

On the other hand, we show the optimal solution to problem \eqref{eq:bss3} is in the set of feasible solutions to problem \eqref{eq:bss2}. Let $h(\zv) = \frac{1}{n}\norm{\yv - \Xmat (\alphav\odot \zv)}^2 + \lambda \norm{\zv}_0$, $f(\zv) = \min_{\alphav}h(\zv)$, $g(\zv) = \argmin_{\alphav}h(\zv)$, and assume $\piv^*$ is the optimal $\piv$ in problem \eqref{eq:bss3}. First, if $p(\zv | \piv^*)$ is a point mass density $\delta_{\zv^*}$ with $ \pi^*_j = z^*_j, j\in [p]$ and $\alphav^* = g(\zv^*)$, 
by setting $\zv = \piv^*$, $\alphav = \alphav^*$, it would give a feasible solution to problem \eqref{eq:bss2} with the same objective value as problem \eqref{eq:bss3}. 
Second, if $p(\zv | \piv^*)$ is not a point mass density, assume $\text{supp}[p(\zv|\piv^*)] = \{\zv_1, \cdots, \zv_K\}$.
We show by contradiction that all the points in $\text{supp}[p(\zv|\piv^*)]$ would give the same objective value $f(\zv)$. Otherwise there exist $\zv_s, \zv_l \in \text{supp}[p(\zv | \piv^*)]$ with $f(\zv_s) < f(\zv_l)$ and $f(\zv_s) \leq f(\zv_k)$ for $k \neq s,l$. By setting %
$\hat{\pi}_j = z_{ij}$, we would have $\E_{\zv \sim p(\zv |\hat{\piv})} f(\zv) < \E_{\zv \sim p(\zv |\piv^*)} f(\zv)$ which contradicts with the assumption that $\piv^*$ is optimal. 
Therefore, we have $f(\zv_1) = f(\zv_2)=\cdots = f(\zv_K)$. Hence all points in $\{\zv_k, g(\zv_k)\}_{k \in [K]}$ are feasible solutions to problem \eqref{eq:bss2} which give the same objective value as what the optimal solution gives in problem \eqref{eq:bss3}. 

In summary, we show problem \eqref{eq:bss2} and \eqref{eq:bss3} have the same global optima and the same objective value at such points, so they are equivalent problems.
\end{proof}

\subsection{Proof of Proposition~\ref{prop:R-ARM-U2G}} 
\label{subsec:proof:prop:R-ARM-U2G}
\begin{proof}
Without loss of generality, we assume that $f(1), f(0)>0$, $\pi = \sigma(\phi)\geq 1/2$, and let $\Delta = f(1) - f(0)$. \\
For the first inequality, direct calculation shows that
\bas{
\var[g_{\text{ARM}}] - \var[g_{R}] =& \E_u[g_{\text{ARM}}^2] - \E_u[g_{R}^2] \\
=& s_1 f(1)^2 + s_2 f(0)^2 + s_3 f(1)f(0) \\
=& (s_1 + s_2 + s_3) f(0)^2 + s_1 \Delta^2 + (2s_1 + s_3) f(0)\Delta \\
\leq& (s_1 + s_2 + s_3) f(0)^2 + (3s_1 + s_3) f(0)^2
}
where $s_1 = -\frac{5}{3}\pi^3 + 3\pi^2 - \frac{3}{2}\pi + \frac{1}{6}$, $s_2 = \frac{1}{3}\pi^3 - \frac{1}{2}\pi + \frac{1}{6}$, $s_3 = \frac{4}{3}\pi^3 - 2\pi^2 + \pi - \frac{1}{3}$.\\
Re-organizing the coefficients, we have
\bas{
\var[g_{\text{ARM}}] - \var[g_{R}] \leq -\big[\pi(1-\frac{\pi}{6})+\frac{21\pi+1}{6}(1-\pi)\big]f(0)^2 \le 0.
}
For the second inequality, we find that
\bas{
\var[g_{\text{U2G}}] - \var[g_{\text{ARM}}] =& \E_u[g_{\text{U2G}}^2] - \E_u[g_{\text{ARM}}^2] = - \frac{(1-\pi)^3}{6}(f(1)-f(0))^2 \le 0.
}
As a consequence, we obtain the conclusion of the proposition.
\end{proof}

\subsection{Proof of Proposition \ref{prop:umvue}}
\label{sec:umvue}
\begin{proof}
We consider a constrained optimization problem
\bas{
\min_g \int_0^1 g^2(u)du, \quad \text{subject to }~\E[g(u)]= \mu,
}
where $\mu = \pi(1-\pi)(f_1 - f_0)$. For simplicity, we omit the conditional notation on $\pi$ if it is clear. The integration can be decomposed into three intervals $[0, 1-\pi], (1-\pi, \pi], (\pi, 1]$, so we can rewrite Eq.~\eqref{eq:linear_form_estimator} into a piece-wise function.
\[
g(u)=
\begin{cases}
g_1(u):=a(u)f_1+b(u)f_0, & u\in [0, 1-\pi]\\ 
g_2(u):=a(u)f_1+b(u)f_1, & u\in (1-\pi, \pi]\\ 
g_3(u):=a(u)f_0+b(u)f_1, & u\in (\pi, 1]
\end{cases}
\]
And we would like to minimize
\[ \bE[g^2(u)] = \bE_{u\in [0, 1-\pi]} [g_1^2(u)] + \bE_{u\in [1-\pi, \pi]} [g_2^2(u)] + \bE_{u\in [\pi, 1]} [g_3^2(u)].\]
If there exists $\pi$ and a positive measure subset $\mathcal{S}_{\pi} \subset (1-\pi, \pi]$ where $|a(u;\pi)+b(u;\pi)| \ge \epsilon_{\pi} > 0$ and $f_1 \neq 0$, then
\bas{
\bE [g^2(u;\pi)] \ge \bE [g_2^2(u;\pi)] \ge \epsilon_{\pi}^2 |\mathcal{S}_{\pi}| f_1^2. 
}
This means that
\bas{
\lim_{\Delta \to 0}~\var[g(u; \pi)] =& \lim_{\Delta \to 0}~\bE [g^2(u; \pi)] - \mu^2 \\
=& \lim_{\Delta \to 0}~\bE [g^2(u; \pi)] \ge \epsilon_{\pi}^2 |\mathcal{S}_{\pi}| f_1^2 > 0,
}
which contradicts the condition~\eqref{eq:var_condition}. Therefore if the estimator has positive SNR for $\pi \in (0,1)$, it has to satisfy $a(u)+b(u)=0$ almost surely (a.s.) in $(1-\pi, \pi]$ or has $f_1 = 0$, where in both cases $g(u)=0$ a.s. for $u \in (1-\pi, \pi]$. Assume $g(u)=0$ for $u \in (1-\pi, \pi]$ a.s., we have
\bas{
\var[g(u;\pi)] =& (1-\pi)[\int_0^{1-\pi} \frac{1}{1-\pi}g^2(u)du + \int_\pi^{1} \frac{1}{1-\pi} g^2(u)du] - \mu^2 \\
\geq & \frac{1}{1-\pi} \Big\{[\int_0^{1-\pi} g(u) du]^2 + [\int_\pi^1 g(u) du]^2\Big\} - \mu^2 \\
=& \frac{1}{1-\pi} (2s^2 -2\mu s + \mu^2) - \mu^2 \\
\geq & \frac{2\pi-1}{2(1-\pi)}\mu^2,
}
where $s = \int_0^{1-\pi} g(u) du$; both inequalities are equalities if and only if $g(u) = \frac{\mu}{2(1-\pi)}$ for $u \in [0,1-\pi] \cup (\pi,1]$. Together with the premise $g(u) = 0$ for $u \in (1-\pi, \pi]$ a.s., we get the U2G estimator. The same argument holds for $\pi < 0.5$ because of the symmetry. 
\end{proof}

\subsection{Proof of Lemma~\ref{lemma:expect_grad}} 

\begin{proof}
\sloppy
In order to ease the presentation, we denote $\zv = \hpgreaterv, \tilde{\zv}=\hplessv$, $\pi_j = \sigma(\phi_j)~\forall j \in \{1,\cdots,p\}$. Then we have $\gv_{\text{U2G}}(\uv;\sigma(\phiv)) = \frac{f(\zv) - f(\tilde{\zv})}{2}\sigma(|\phiv|)\odot(\hpgreaterv - \hplessv)$.
Construct a sequence of binary code $\zv^0=\zv, \zv^1, \cdots, \zv^p=\tilde{\zv}$ by flipping one dimension of $\zv$ to the value in $\tilde{\zv}$ at a time, $i.e.$,
$
\zv^i = (\tilde{z}_1, \cdots, \tilde{z}_i, z_{i+1}, \cdots, z_p)'.
$
Hence $f_{\Xmat, \yv}(\zv)-f_{\Xmat, \yv}(\tilde{\zv}) = \sum_{i=1}^{p} (f_{\Xmat, \yv}(\zv^{i-1})-f_{\Xmat, \yv}(\zv^{i}))$.
We prove the statement for the gradient vector element-wisely. Consider the $j^{th}$ dimension of the gradient vector
\ba{
\resizebox{.9\hsize}{!}{
$ \bE_{\uv} [g(\uv)_j] = \frac{\sigma(|\phi_j|)}{2}  \bE_{\uv} \sum_{i=1}^{p} (f_{\Xmat, \yv}(\zv^{i-1})-f_{\Xmat, \yv}(\zv^{i})) (\mathbf{1}_{[u_j > \sigma(-\phi_j)]} - \mathbf{1}_{[u_j < \sigma(\phi_j)]})$
}
 \label{eq:grad_element}
}
Note that $\zv^{i-1}$ and $\zv^{i}$ only differ on the $i^{th}$ dimension, and different dimensions of $\uv$ are independent. Consider the $i^{th}$ element of the  summation in Eq.\eqref{eq:grad_element} and W.L.O.G. we first assume the logit $\phi_i \geq 0$. For $i\ne j$, due to the symmetry of the sigmoid function, we have
\bas{
& ~  \bE_{\uv} \frac{\sigma(|\phi_j|)}{2}(f_{\Xmat, \yv}(\zv^{i-1})-f_{\Xmat, \yv}(\zv^{i}))
(\mathbf{1}_{[u_j > \sigma(-\phi_j)]} - \mathbf{1}_{[u_j < \sigma(\phi_j)]}) \\
=&~ \frac{\sigma(|\phi_j|)}{2}  \bE_{\uv_{-i}} \left[ \bE_{u_i}[ f_{\Xmat, \yv}(\zv^{i-1})-f_{\Xmat, \yv}(\zv^{i})) (\mathbf{1}_{[u_j > \sigma(-\phi_j)]} - \mathbf{1}_{[u_j < \sigma(\phi_j)]})|\uv_{-i} ] \right] \\
=&~ \frac{\sigma(|\phi_j|)}{2}  \bE_{\uv_{-i}} \left[ (\mathbf{1}_{[u_j > \sigma(-\phi_j)]} - \mathbf{1}_{[u_j < \sigma(\phi_j)]})\bE_{u_i}[ f_{\Xmat, \yv}(\zv^{i-1})-f_{\Xmat, \yv}(\zv^{i})) |\uv_{-i} ] \right] \\
=&~   \bE_{\uv_{-i}}  \Big[ (\mathbf{1}_{[u_j > \sigma(-\phi_j)]} - \mathbf{1}_{[u_j < \sigma(\phi_j)]})  \Big( \int_{0}^{\sigma(-\phi_i)} \big(f_{\Xmat, \yv}(\zv^{i-1}|z_i=0) - f_{\Xmat, \yv}(\zv^i|z_i=1)\big)du_i\\
&~  + \int_{\sigma(\phi_i)}^{1}\big(f_{\Xmat, \yv}(\zv^{i-1}|z_i=1) - f_{\Xmat, \yv}(\zv^i|z_i=0)\big)du_i  \Big) | \uv_{-i}   \Big]\frac{\sigma(|\phi_j|)}{2} \\
=&~  \bE_{\uv_{-i}}  \Big[ (\mathbf{1}_{[u_j > \sigma(-\phi_j)]} - \mathbf{1}_{[u_j < \sigma(\phi_j)]})  \Big(  \big(f_{\Xmat, \yv}(\zv^{i-1}|z_i=0) - f_{\Xmat, \yv}(\zv^{i}|z_i=1)\big)(1-\sigma(\phi_i)) \\
&~    +  \big(f_{\Xmat, \yv}(\zv^{i-1}|z_i=1) - f_{\Xmat, \yv}(\zv^{i}|z_i=0)\big)(1-\sigma(\phi_i))  \Big)| \uv_{-i}  \Big]\frac{\sigma(|\phi_j|)}{2} \\
=&~ 0.
}
 Whereas for $i=j$, we have
\bas{
&~\frac{\sigma(|\phi_j|)}{2}   \bE_{\uv} (f_{\Xmat, \yv}(\zv^{j-1})-f_{\Xmat, \yv}(\zv^{j})) (\mathbf{1}_{[u_j > \sigma(-\phi_j)]} - \mathbf{1}_{[u_j < \sigma(\phi_j)]}) \\
=&~ \frac{\sigma(\phi_j)}{2}  \bE_{\uv_{-j}} \left[ \bE_{u_j}[ (f_{\Xmat, \yv}(\zv^{j-1})-f_{\Xmat, \yv}(\zv^{j})) (\mathbf{1}_{[u_j > \sigma(-\phi_j)]} - \mathbf{1}_{[u_j < \sigma(\phi_j)]})|\uv_{-j} ] \right] \\
=&~  \frac{\sigma(\phi_j)}{2}  \bE_{\uv_{-j}} \left[ \left( \int_{0}^{\sigma(-\phi_j)} (f_{\Xmat, \yv}(\zv^{j-1}|z_j=0) - f_{\Xmat, \yv}(\zv^j|z_j=1))(-1)du_j \right. \right.\\
    &+\left.\left.\left. \int_{\sigma(\phi_j)}^{1} (f_{\Xmat, \yv}(\zv^{j-1}|z_j=1) - f_{\Xmat, \yv}(\zv^j|z_j=0))du_j \right) \right| \uv_{-j}  \right] \\
=&~ \bE_{\uv_{-j}}  \Big[ \sigma(\phi_j)(1-\sigma(\phi_j))  [f_{\Xmat, \yv}(\zv^{j-1}|z_j=1) - f_{\Xmat, \yv}(\zv^j|z_j=0)]  \Big]\\
=&~ \pi_j (1-\pi_j)\E_{\uv}[\Delta_{\zv,j}f ]. 
}
The same derivation holds true when the logit $\phi_i \leq 0$. Hence for each dimension there is only one non-zero element in the summation of Eq.\eqref{eq:grad_element}. 
Rewriting the result in vector form proves the lemma.
\end{proof}

\subsection{ Proof of Lemma~\ref{lemma:expectation_u2g}}

\begin{proof}
Based on Lemma~\ref{lemma:expect_grad}, it suffices to compute the expectation of $\E_{\Xmat, \yv, \uv}[\Delta_{\zv}f ]$ to obtain the conclusion of Lemma~\ref{lemma:expectation_u2g}, where $\Delta_{\zv}f = (\Delta_{\zv,1}f, \cdots, \Delta_{\zv,p}f)$. For any $k \in [p]$, direct application of conditional expectation formulations leads to
\begin{align}
    \E_{\Xmat, \yv, \uv}[\Delta_{\zv, k}f ] = \E[\Delta_{\zv, k}f \Big| \|\zv\|_0 < n - 1] p(\|\zv\|_0< n-1) \label{eq:key_conditional_expectation} \\
    & \hspace{-24em} + \E[\Delta_{\zv, k}f \Big| \|\zv\|_0 = n - 1] p(\|\zv\|_0 = n-1) + \E[\Delta_{\zv, k}f \Big| \|\zv\|_0 \geq n] p(\|\zv\|_0 \geq n). \nonumber
\end{align}
Conditioned on the event $\|\zv\|_{0} < n - 1$, we denote projection matrix $P_{\zv} = \Xmat_{\zv} (\Xmat_{\zv}^{\top} \Xmat_{\zv})^{-1} \Xmat_{\zv}^{\top}$ for any $\zv$ and the OLS estimator $\hat{\alphav}_{\zv} = (\Xmat_{\zv}^{\top} \Xmat_{\zv})^{-1} \Xmat_{\zv}^{\top}\yv$.
Under that event, given the definition of $\deltaf$ we obtain that
\begin{align}
    \deltaf & = \lambda + \frac 1n \norm{ \yv - \Xmat_{\tilde{\zv}}\hat{\alphav}_{\tilde{\zv}}}_{2}^2 - \frac 1n \norm{ \yv - \Xmat_{\zv}\hat{\alphav}_{\zv}}_{2}^2 \nonumber \\
    & = \lambda + \underbrace{\frac 1n \norm{ \yv - \Xmat_{\tilde{\zv}}\hat{\alphav}_{\tilde{\zv}}}_{2}^2 - \frac{1}{n} \norm{ \yv - \Xmat_{\tilde{\zv}}\betav^*_{\tilde{\zv}}}_{2}^2}_{: = T_{1}}  \nonumber \\
    & + \underbrace{\frac 1n \norm{ \yv - \Xmat_{\tilde{\zv}}\betav^*_{\tilde{\zv}}}_{2}^2 - \frac 1n  \norm{ \yv - \Xmat_{\zv}\betav^*_{\zv}}_{2}^2}_{: = T_{2}} + \underbrace{\frac 1n  \norm{ \yv - \Xmat_{\zv}\betav^*_{\zv}}_{2}^2 - \frac 1n \norm{ \yv - \Xmat_{\zv}\hat{\alphav}_{\zv}}_{2}^2}_{: = T_{3}}, \label{eq:lemma_key_equation}
\end{align} 
where $\tilde{\zv} \in \{0,1\}^{ p}$ is such that $\zv$ and $\tilde{\zv}$ only differ on dimension $k$, $i.e.$, $\zv_k=0, \tilde{\zv}_k=1, \zv_j=\tilde{\zv}_j, \forall j\ne k$. Note that, $\|\tilde{\zv}\|_{0} = \|\zv\|_{0} + 1 \leq n - 1$ when $\|\zv\|_{0} < n - 1$; therefore, the OLS estimator $\hat{\alphav}_{\widetilde{\zv}} = (\Xmat_{\widetilde{\zv}}^{\top} \Xmat_{\widetilde{\zv}})^{-1} \Xmat_{\widetilde{\zv}}^{\top}\yv$ is valid. Regarding term $T_{1}$ in equation~\eqref{eq:lemma_key_equation}, direct computation shows that
\begin{align}
    T_{1} = - \frac 1n (\Xmat_{-\tilde{\zv}} \betav^*_{-\tilde{\zv}} + \epsilonv)^{\top} P_{\tilde{\zv}} (\Xmat_{-\tilde{\zv}} \betav^*_{-\tilde{\zv}} + \epsilonv) = -\frac 1n \norm{P_{\tilde{\zv}} (\Xmat_{-\tilde{\zv}} \betav^*_{-\tilde{\zv}} + \epsilonv)}_2^2.
\end{align}
By Lemma~\ref{lemma:tao}, conditioned on $\uv$, the rank of $X_{\tilde{\zv}}$ equals $\|\tilde{\zv}\|_0$ with probability one. Taking the expectation of $T_{1}$ with respect to $\Xmat, \yv$, a key observation is that for any $X_{\tilde{\zv}}$ with full column rank,
\begin{align}
   \bE_{\Xmat, \yv}[T_1|X_{\tilde{\zv}}, \|\zv\|_{0} < n - 1] = - \frac{(\sigma^2 + \norm{\betav_{-\tilde{\zv}}^{*}}_{2}^2)\norm{\tilde{\zv}}_{0}}{n}.
   \label{eq:condexpT1_Xz}
\end{align}
Therefore, we obtain that 
\begin{align}
    \mathbb{E}_{\Xmat, \yv}[T_1\Big| \|\zv\|_0 < n - 1] = \bE_{\Xmat, \yv}[\bE[T_1|X_{\tilde{\zv}}, \|\zv\|_{0} < n - 1]] = - \frac{(\sigma^2 + \norm{\betav_{-\tilde{\zv}}^{*}}_{2}^2)\norm{\tilde{\zv}}_{0}}{n}.
    \label{eq:expT1}
\end{align}
The above result leads to
\begin{align*}
    \mathbb{E}_{\Xmat, \yv, \uv}[T_1\Big| \|\zv\|_0 < n - 1] = - \mathbb{E}_{\uv} \biggr[\frac{(\sigma^2 + \norm{\betav_{-\tilde{\zv}}^{*}}_{2}^2)\norm{\tilde{\zv}}_{0}}{n} \ | \ \|\zv\|_{0} < n -1 \biggr].
\end{align*}
For the term $T_{3}$ in equation~\eqref{eq:lemma_key_equation}, similar argument proves that 
\begin{align*}
    \mathbb{E}_{\Xmat, \yv, \uv}[T_3\Big| \|\zv\|_0 < n - 1] = \mathbb{E}_{\uv} \biggr[\frac{(\sigma^2 + \norm{\betav_{-\zv}^{*}}_{2}^2)\norm{\zv}_{0}}{n} \ | \ \|\zv\|_{0} < n -1 \biggr].
\end{align*}
For the term $T_{2}$ in equation~\eqref{eq:lemma_key_equation}, direct calculation shows that
\begin{align*}
    T_{2} = -(\beta_{k}^{*})^2 \frac{1}{n} \|X_{k}\|^2 - \frac{2}{n} \beta_{k}^{*} X_{k}^{\top} R_{\tilde{\zv}},
\end{align*}
where $R_{\tilde{\zv}} = \Xmat_{-\tilde{\zv}} \betav^*_{-\tilde{\zv}} + \epsilonv$. Therefore, %
$\E_{\Xmat, \yv, \uv}[T_2\Big| \|\zv\|_0 < n - 1] = -(\beta_{k}^{*})^2$.

Putting the above results together, we find that
\begin{align}
    \E_{\Xmat, \yv, \uv}[\Delta_{\zv, k}f \Big| \|\zv\|_0 < n-1] = \lambda - \frac{1}{n} \biggr( (\beta_{k}^{*})^2 (n - \E_{\uv} [ \|\zv\|_{0} \ | \ \|\zv\|_{0} < n-1] - 1) & \nonumber \\
    & \hspace{-18 em} + \sigma^2 + \E_{\uv}[\norm{\betav_{-\zv}^{*}}_{2}^2\ | \ \|\zv\|_{0} < n-1] \biggr). \label{eq:first_expectation}
\end{align}
Conditioned on the event $\|\zv\|_{0} = n - 1$, recall that
\begin{align}
    \deltaf = \lambda - \frac 1n \norm{ \yv - \Xmat_{\zv}\hat{\alphav}_{\zv}}_{2}^2 = \lambda - \frac{1}{n} \norm{ (I_{n} - P_{z}) W_{\zv}}_{2}^2, \nonumber
\end{align}
where $W_{z} = \Xmat_{-\zv} \betav^*_{-\zv} + \epsilonv$. Conditioned on $\uv$, $W_{z} \sim \mathcal{N}(0, (\sigma^2+\|\betav_{-\zv}^*\|_2^2)I_{n})$. Therefore, we obtain that
\begin{align*}
    \E_{\Xmat, \yv} \biggr[ \frac{1}{n} \norm{ (I_{n} - P_{z}) W_{\zv}}_{2}^2 \ | \ \|\zv\|_{0} = n - 1 \biggr] = \frac{\sigma^2 + \| \betav_{-\zv}^{*}\|_{2}^2}{n}. 
\end{align*}
The above inequality shows that
\begin{align}
    \E_{\Xmat, \yv, \uv}[\Delta_{\zv, k}f \Big| \|\zv\|_0 = n - 1] = \lambda - \frac{\sigma^2 + \E_{\uv} [\| \betav_{-\zv}^{*}\|_{2}^2 \ | \ \|\zv\|_{0} = n - 1]}{n}. \label{eq:second_expectation}
\end{align}
Finally, we compute $\E_{\Xmat, \yv, \uv}[\Delta_{\zv, k}f \Big| \|\zv\|_0 \geq n]$. Under the setting $\|\zv\|_{0} \geq n$, we have $\|\tilde{\zv}\|_{0} \geq n + 1$. By Lemma~\ref{lemma:tao}, conditioned on $\uv$, with probability one, we have $\min_{\alphav}\frac{1}{n}~ \norm{\yv - \Xmat_{\zv}\alphav}_2^2 = 0$ and $\min_{\alphav}\frac{1}{n}~ \norm{\yv - \Xmat_{\tilde{\zv}}\alphav}_2^2 = 0$. Therefore, conditioned on $\uv$ and $\|\zv\|_{0} \geq n$, with probability one, we obtain that
\begin{align*}
    \deltaf & = \lambda.
\end{align*}
It directly leads to
\begin{align}
    \E_{\Xmat, \yv, \uv}[\Delta_{\zv, k}f \Big| \|\zv\|_0 \geq n] = \lambda. \label{eq:third_expectation}
\end{align}
Plugging the results of equations~\eqref{eq:first_expectation},~\eqref{eq:second_expectation}, and~\eqref{eq:third_expectation} into the equation~\eqref{eq:key_conditional_expectation}, we obtain the conclusion of the lemma. 
\end{proof}

\subsection{Proof of Proposition~\ref{prop:expected_grad}}

\begin{proof}
We consider cases whether the covariates are in the true active set separately. 

\textbf{Case 1 - $\{j:\beta_j^*=0\}$:} when $\|\zv\|_0< n$, by Lemma~\ref{lemma:expectation_u2g}, we have
\bas{
\E_{\Xmat, \yv, \uv}[g_{\emph{\text{U2G}}}(\uv; \sigma(\phiv))_j] &= \biggr(\lambda - \frac{\sigma^2 + \E_{\uv}[\norm{\betav_{-\zv}^{*}}_{2}^2\Big| \|\zv\|_0 < n]}{n} \biggr) \pi_{j} ( 1 - \pi_{j}) \\ %
 &\geq \biggr( \lambda - \frac{\sigma^2 + \|\betav^{*}\|_{2}^2}{n} \biggr) \pi_{j} ( 1 - \pi_{j}),
}
where the  inequality is due to $\|\betav_{- \zv}^{*}\|_{2}^2 \leq \|\betav^{*}\|_{2}^2$ for all $\zv$. %
Therefore, if $\lambda > (\sigma^2 + \norm{\betav^*}_2^2)/n$, we have $\E_{\Xmat,\yv,\uv}[g_{\text{U2G}}(\uv; \sigma(\phiv))_j] > 0$ for all $j\notin \mathcal{A}$. 

\textbf{Case 2 - $\{j:\beta_j^* \neq 0\}$:} by Hoeffding's inequality for sub-Gaussian random variables, and by the assumption $ p \leq 2\eta^2(n - 1)^2 / \log(n)$, we have
\bas{
p(\|\zv\|_0 \geq n - 1) \leq \exp\left(-\frac{2(n - 1 - \sum_{k=1}^p \pi_j)^2}{p}\right) \leq \frac 1n.
}
Furthermore, simple algebra shows that $$\E\Big[\|\zv\|_0 \Big| \|\zv\|_0< n - 1\Big]  \leq \E[\|\zv\|_0] = \sum_{k=1}^p\pi_k.$$ 
Collecting the previous results and using the result of Lemma~\ref{lemma:expectation_u2g}, we find that

\begin{align}
\E_{\Xmat, \yv, \uv}[g_{\emph{\text{U2G}}}(\uv; \sigma(\phiv))_j] & \leq \biggr( \lambda - \frac{(\beta_{j}^{*})^2 (n - 1)}{n} (1 - \frac{1}{n}) + \frac{(\beta_{j}^{*})^2 (\sum_{k = 1}^{p} \pi_{k})}{n} \biggr) \pi_{j} (1 - \pi_{j}) \notag \\
& \leq \biggr(\lambda - (\beta_{j}^{*})^2 \biggr[\frac{(n-1)^2}{n^2} - \frac{\sum_{k = 1}^{p} \pi_{k}}{n}\biggr]\biggr) \pi_{j} (1 - \pi_{j})  \notag \\
& \leq \biggr( \lambda - (\beta_{j}^{*})^2 \frac{n - 1}{n} (\eta - \frac{1}{n}) \biggr) \pi_{j}(1 - \pi_{j}).
\end{align}

Therefore, as long as $\lambda <\frac{n-1}{n} (\eta - \frac{1}{n}) \min_{k \in \mathcal{A}}{(\beta_k^*)^2} $, we have $\E_{\Xmat, \yv, \uv}[g_{\emph{\text{U2G}}}(\uv; \sigma(\phiv))_j] < 0$ for all $j \in \mathcal{A}$.

Combining the two cases,  by setting
\ba{
\lambda \in \left(\frac{\norm{\betav^*}_2^2 + \sigma^2}{n}, \frac{n-1}{n} (\eta - \frac{1}{n}) \min_{k \in \mathcal{A}}{(\beta_k^*)^2} \right) : = \mathcal{I},
\label{eq:lambda_range}
}
the proposition is proved.
\end{proof}

\subsection{Proof of Theorem \ref{theorem:rate_population}}

\begin{proof}
For the simplicity of the presentation, we denote 
\begin{align*}
    G(\phiv) := \E_{\Xmat, \yv, \uv}[g_{\emph{\text{U2G}}}(\uv; \sigma(\phiv))]/ [\piv (1 - \piv)],
\end{align*}
where $\piv = \sigma(\phiv)$ and the division is element-wise. Furthermore, we denote 
\begin{align*}
    \eta = \min\left\{\sqrt{\frac{p\log(n)}{2 (n - 1)^2}},\frac{\norm{\betav^*}_2^2+ \sigma^2 }{(n-1)\min_{k \in \mathcal{A}}{(\beta_k^*)^2}}
\right\}.
\end{align*}
We now state our proof with two parts.

(a) Since $\lambda \in \mathcal{I}$ and $\sum_{j=1}^p \sigma(\phi_j^{(0)}) \leq \varpi - S$, we have $G(\phi_{j}^{(t)}) < 0$ for all $j \in \cA$. It shows that the updates $\phi_{j}^{(t+1)}$ are monotonically increasing, $i.e.$, $\phi_{j}^{(t+1)} > \phi_{j}^{(t)}$ for any $t \geq 0$ and $j \in \cA$. For any $t \geq 0$, we obtain that
\begin{align*}
    \phi_{j}^{(t+1)} = \phi_{j}^{(0)} - \rho \biggr(\sum_{i = 0}^{t} \E_{\Xmat, \yv, \uv}[g_{\emph{\text{U2G}}}(\uv; \sigma(\phiv^{(i)}))_{j}] \biggr).
\end{align*}
We now study the lower bound of $t$ such that $\phi_{j}^{(t+1)} \geq 0$ for the first time. In order to study that, we assume that $\phi_{j}^{(i)} < 0$ for all $0 \leq i \leq t$. Based on the proof of Proposition~\ref{prop:expected_grad}, for any $0 \leq i \leq t$, we have
\begin{align*}
    \E_{\Xmat, \yv, \uv}[g_{\emph{\text{U2G}}}(\uv; \sigma(\phiv^{(i)}))_{j}] & \leq \biggr( \lambda - (\beta_{j}^{*})^2 \frac{n - 1}{n} (\eta - \frac{1}{n}) \biggr) \sigma(\phi_{j}^{(i)})(1 - \sigma(\phi_{j}^{(i)})) \\
    & \leq \biggr( \lambda - (\beta_{j}^{*})^2 \frac{n - 1}{n} (\eta - \frac{1}{n}) \biggr) \sigma(\phi_{j}^{(0)})(1 - \sigma(\phi_{j}^{(0)})),
\end{align*}
where the second inequality is due to the fact that $\sigma_{j}^{(0)} \leq \sigma_{j}^{(i)} < 0$ for all $0 \leq i \leq t$. Collecting the above results, we find that
\begin{align*}
    \phi_{j}^{(t+1)} \geq \phi_{j}^{(0)} - \rho t \biggr( \lambda - (\beta_{j}^{*})^2 \frac{n - 1}{n} (\eta - \frac{1}{n}) \biggr) \sigma(\phi_{j}^{(0)})(1 - \sigma(\phi_{j}^{(0)})).
\end{align*}
Therefore, as long as $t \geq \frac{\phi_{j}^{(0)}}{\rho \biggr( \lambda - (\beta_{j}^{*})^2 \frac{n - 1}{n} (\eta - \frac{1}{n}) \biggr) \sigma(\phi_{j}^{(0)})(1 - \sigma(\phi_{j}^{(0)}))} : = T_{1}$, we have $\phi_{j}^{(t+1)} \geq 0$. 

By using the inequality $\sigma(x + y) \leq \sigma(x) + y \sigma'(x)$ for any $x, y > 0$, for any $j \in \cA$ and $t \geq T_{1}$ we obtain that
\begin{align}
    1 - \sigma(\phi_{j}^{(t+1)}) & \geq 1 - \sigma(\phi_{j}^{(t)}) + \rho G(\phi_{j}^{(t)}) \sigma(\phi_{j}^{(t)}) (1 - \sigma( \phi_{j}^{(t)})) \sigma'(\phi_{j}^{(t)}) \nonumber \\
    & = \parenth{ 1 - \sigma(\phi_{j}^{(t)})} \brackets{ 1 + \rho \sigma^2(\phi_{j}^{(t)}) G(\phi_{j}^{(t)}) \parenth{1 - \sigma(\phi_{j}^{(t)})}}. \label{eq:lower_bound_first}
\end{align}
Based on the result of Lemma~\ref{lemma:expectation_u2g}, we find that
\begin{align}
    G(\phi_{j}^{(t)}) \geq \lambda - \frac{\sigma^2 + \|\betav^{*}\|_{2}^2 + (n - 1) (\beta_{i}^{*})^2}{n}, \label{eq:useful_bound_first}
\end{align}
which is due to $\norm{\betav_{-\zv}^{*}}_{2}^2\leq \norm{\betav^*}$ and $\norm{\zv}_0 \geq 0$ for all $\zv$.
Plugging this inequality into equation~\eqref{eq:lower_bound_first}, we achieve the conclusion of the lower bound in part (a) with $c_{1} = \rho[(\sigma^2 + \|\betav^{*}\|_{2}^2 + (n - 1) (\beta_{i}^{*})^2)/n - \lambda] > 0$.

Regarding the upper bound, we first prove that 
\begin{align}
    \sigma(x + y) \geq \sigma(x) + \exp(-y)y \sigma'(x) \label{eq:lemma_pop_active}
\end{align}
for all $x, y > 0$. In fact, from the mean-value theorem, there exists $\xi \in (x, x+ y)$ such that $\sigma'(\xi) = \parenth{ \sigma(x + y) - \sigma(x)}/y$. Since the function $\sigma'(.)$ is monotonically decreasing in $(0, \infty)$, we have $\sigma'(\xi) \geq \sigma'(x+y)$. Simple calculation yields that $\sigma'(x+y)/\sigma'(x) \geq \exp(-y)$ for all $x, y > 0$. Therefore, we obtain the conclusion of inequality~\eqref{eq:lemma_pop_active}. Given the inequality~\eqref{eq:lemma_pop_active}, for any $t \geq T_{1}$ we have
\begin{align}
    1 - \sigma(\phi_{j}^{(t+1)}) & \leq 1 - \sigma(\phi_{j}^{(t)}) \nonumber \\
    & + \exp \parenth{ \rho \sigma(\phi_{j}^{(t)}) (1 - \sigma(\phi_{j}^{(t)})) G(\phi_{j}^{(t)})} \rho G(\phi_{j}^{(t)}) \sigma^2(\phi_{j}^{(t)}) (1 - \sigma(\phi_{j}^{(t)}))^2. \label{eq:upper_bound_first} 
\end{align}
Since $G(\phi_{j}^{(t)}) < 0$ and $\phi_{j}^{(t)} \geq 0$, an application of the bound~\eqref{eq:useful_bound_first} and standard inequality $\sigma(\phi_{j}^{(t)})(1 - \sigma(\phi_{j}^{(t)})) \leq 1/4$ leads to
\begin{align}
    \sigma(\phi_{j}^{(t)}) (1 - \sigma(\phi_{j}^{(t)})) G(\phi_{j}^{(t)}) \geq \frac{1}{4}\biggr(\lambda - \frac{\sigma^2 + \|\betav^{*}\|_{2}^2 + (n - 1) (\beta_{i}^{*})^2}{n}\biggr). \label{eq:another_bound_first}
\end{align}
Furthermore, using the proof argument of Proposition~\ref{prop:expected_grad}, we find that
\begin{align}
    G(\phi_{j}^{(t)}) \leq \lambda - (\beta_{i}^{*})^2 \frac{n - 1}{n} (\eta - \frac{1}{n}), \label{eq:another_bound_second}
\end{align}
where $\eta = \min\left\{\sqrt{p\log(n)/2n^2},(\norm{\betav^*}_2^2+ \sigma^2 )/(n-1)\min_{k \in \mathcal{A}}{(\beta_k^*)^2}
\right\}$. Plugging the bounds~\eqref{eq:another_bound_first} and~\eqref{eq:another_bound_second} into the equation~\eqref{eq:upper_bound_first}, we obtain the conclusion of the upper bound in part (a) with $C_{1}$ is given by
\begin{align*}
    C_{1} = \rho \exp \parenth{ \frac{\rho}{4} \biggr(\lambda - \frac{\sigma^2 + \|\betav^{*}\|_{2}^2 + (n - 1) (\beta_{i}^{*})^2}{n}\biggr)} \biggr( (\beta_{i}^{*})^2 \frac{n - 1}{n} (\eta - \frac{1}{n}) - \lambda \biggr) > 0.
\end{align*}
As a consequence, we obtain the conclusion of the upper bound in part (a). \\

(b) The proof of part (b) follows the same line of proof argument as that in part (a). Indeed, since $\lambda \in \mathcal{I}$ and $\sum_{j=1}^p \sigma(\phi_j^{(0)}) \leq \varpi - S$, we have $G(\phi_{j}^{(t)}) > 0$ for all $j \notin \cA$. It shows that the updates $\phi_{j}^{(t+1)}$ are monotonically decreasing, $i.e.$, $\phi_{j}^{(t+1)} < \phi_{j}^{(t)}$ for any $t \geq 0$ and $j \notin \cA$. Similar to part (a), we first find a lower bound on $t$ such that $\phi_{j}^{(t+1)} \leq 0$. In fact, we assume that $\phi_{j}^{(i)} > 0$ for all $0 \leq i \leq t$. Then, based on the proof of Proposition~\ref{prop:expected_grad}, for any $0 \leq i \leq t$, we find that
\begin{align*}
    \E_{\Xmat, \yv, \uv}[g_{\emph{\text{U2G}}}(\uv; \sigma(\phiv^{(i)}))_{j}] & \geq \biggr( \lambda - \frac{\sigma^2 + \|\betav^{*}\|_{2}^2}{n} \biggr) \sigma(\phi_{j}^{(i)})(1 - \sigma(\phi_{j}^{(i)})) \\
    & \geq \biggr( \lambda - \frac{\sigma^2 + \|\betav^{*}\|_{2}^2}{n} \biggr) \sigma(\phi_{j}^{(0)})(1 - \sigma(\phi_{j}^{(0)})),
\end{align*}
where the second inequality is due to the fact that $\sigma_{j}^{(0)} \geq \sigma_{j}^{(i)} > 0$ for any $0 \leq i \leq t$. Therefore, we obtain that
\begin{align*}
    \phi_{j}^{(t+1)} \leq \phi_{j}^{(0)} - \rho t \biggr( \lambda - \frac{\sigma^2 + \|\betav^{*}\|_{2}^2}{n} \biggr) \sigma(\phi_{j}^{(0)})(1 - \sigma(\phi_{j}^{(0)})).
\end{align*}
The above inequality demonstrates that as long as $t \geq \frac{\phi_{j}^{(0)}}{\rho \biggr( \lambda - \frac{\sigma^2 + \|\betav^{*}\|_{2}^2}{n} \biggr) \sigma(\phi_{j}^{(0)})(1 - \sigma(\phi_{j}^{(0)}))} : = T_{2}$, we have $\phi_{j}^{(t+1)} \leq 0$.

Now, we can check that $\sigma(x - y) \geq \sigma(x) - y \sigma'(x)$ for all $x < 0$ and $y >0$. Using this inequality, for any $j \notin \mathcal{A}$ and $t \geq T_{2}$ we obtain that
\begin{align*}
    \sigma(\phi_{j}^{(t+1)}) \geq \sigma(\phi_{j}^{(t)}) - \rho G(\phi_{j}^{(t)}) \sigma(\phi_{j}^{(t)}) (1 - \sigma( \phi_{j}^{(t)})) \sigma'(\phi_{j}^{(t)}). 
\end{align*}
From the result of Lemma~\ref{lemma:expectation_u2g}, it is clear that $G(\phi_{j}^{(t)}) \leq \lambda$ for all $j \notin \mathcal{A}$. Therefore, we have
\begin{align*}
    \sigma(\phi_{j}^{(t+1)}) \geq \sigma(\phi_{j}^{(t)}) \biggr[1 - \lambda \rho (1 - \sigma(\phi_{j}^{(t)}))^2 \sigma(\phi_{j}^{(t)}) \biggr].
\end{align*}
We achieve the conclusion of the lower bound in part (b) with $c_{2} = \rho \lambda > 0$. 

Moving to the upper bound, with similar argument as that of equation~\eqref{eq:lemma_pop_active}, we can check that
\begin{align*}
    \sigma(x - y) \geq \sigma(x) - \exp(-y)y \sigma'(x),
\end{align*}
for any $x < 0$ and $y > 0$. Given that inequality, for any $t \geq T_{2}$ we obtain that
\begin{align}
    \sigma(\phi_{j}^{(t+1)}) & \leq \sigma(\phi_{j}^{(t)}) \nonumber \\
    & \hspace{-2 em} - \exp \parenth{ - \rho \sigma(\phi_{j}^{(t)}) (1 - \sigma(\phi_{j}^{(t)})) G(\phi_{j}^{(t)})} \rho G(\phi_{j}^{(t)}) \sigma(\phi_{j}^{(t)}) (1 - \sigma( \phi_{j}^{(t)})) \sigma'(\phi_{j}^{(t)}). \label{eq:upper_bound_second} 
\end{align}
Based on the result of Lemma~\ref{lemma:expectation_u2g}, we can check that $\sigma(\phi_{j}^{(t)}) (1 - \sigma(\phi_{j}^{(t)}) G(\phi_{j}^{(t)}) \leq \lambda/ 4$ and $G(\phi_{j}^{(t)}) \geq \lambda - (\sigma^2 + \|\betav^{*}\|_{2}^2)/n$. Putting the above results together, we have
\begin{align*}
    \sigma(\phi_{j}^{(t+1)}) \leq \sigma(\phi_{j}^{(t)}) \brackets{1 - C_{2} (1 - \sigma(\phi_{j}^{(t)}))^2 \sigma(\phi_{j}^{(t)})},
\end{align*}
for any $t \geq T_{2}$ where $C_{2} = \rho \exp \parenth{ - \rho \lambda/ 4} \biggr(\lambda - (\sigma^2 + \|\betav^{*}\|_{2}^2)/n\biggr)>0$. As a consequence, we reach the conclusion of the upper bound in part (b).
\end{proof}

\section{Additional Results}
\label{sec:additional}
In this appendix, we provide additional experimental results. For the Prostate cancer dataset, Figure~\ref{fig:corr} shows the pairwise correlations of the two types of chosen covariates in the true active set. Table~\ref{tab:prostate2} summarizes the results for the variable selection methods when the multi-collinearity in the true active set is high.

\begin{figure}[ht]
    \centering
    \subfloat[Moderate multi-collinearity]
    {{\includegraphics[width=7cm]{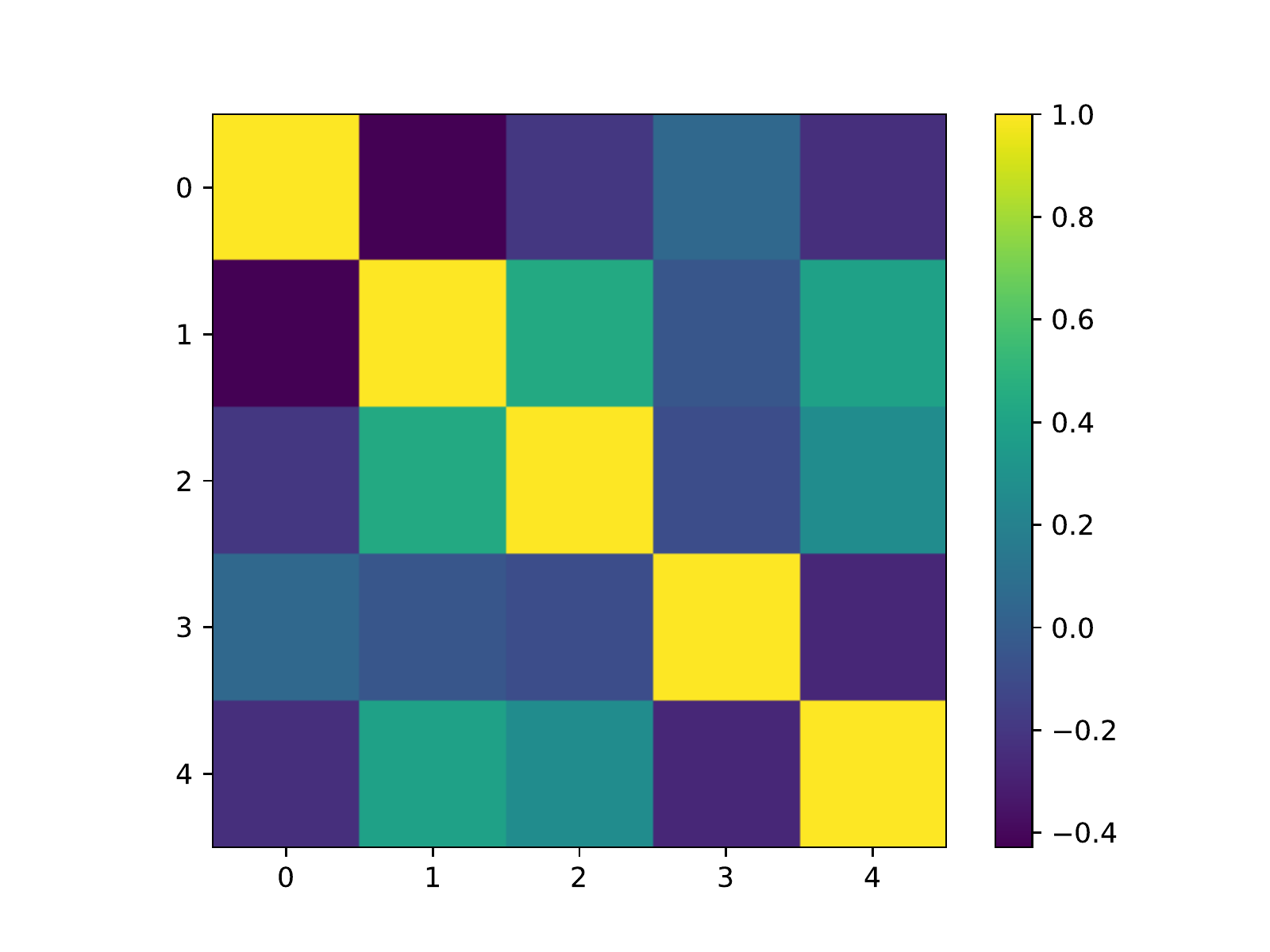} }} \hspace{-12mm}
    \subfloat[High multi-collinearity] 
    {{\includegraphics[width=7cm]{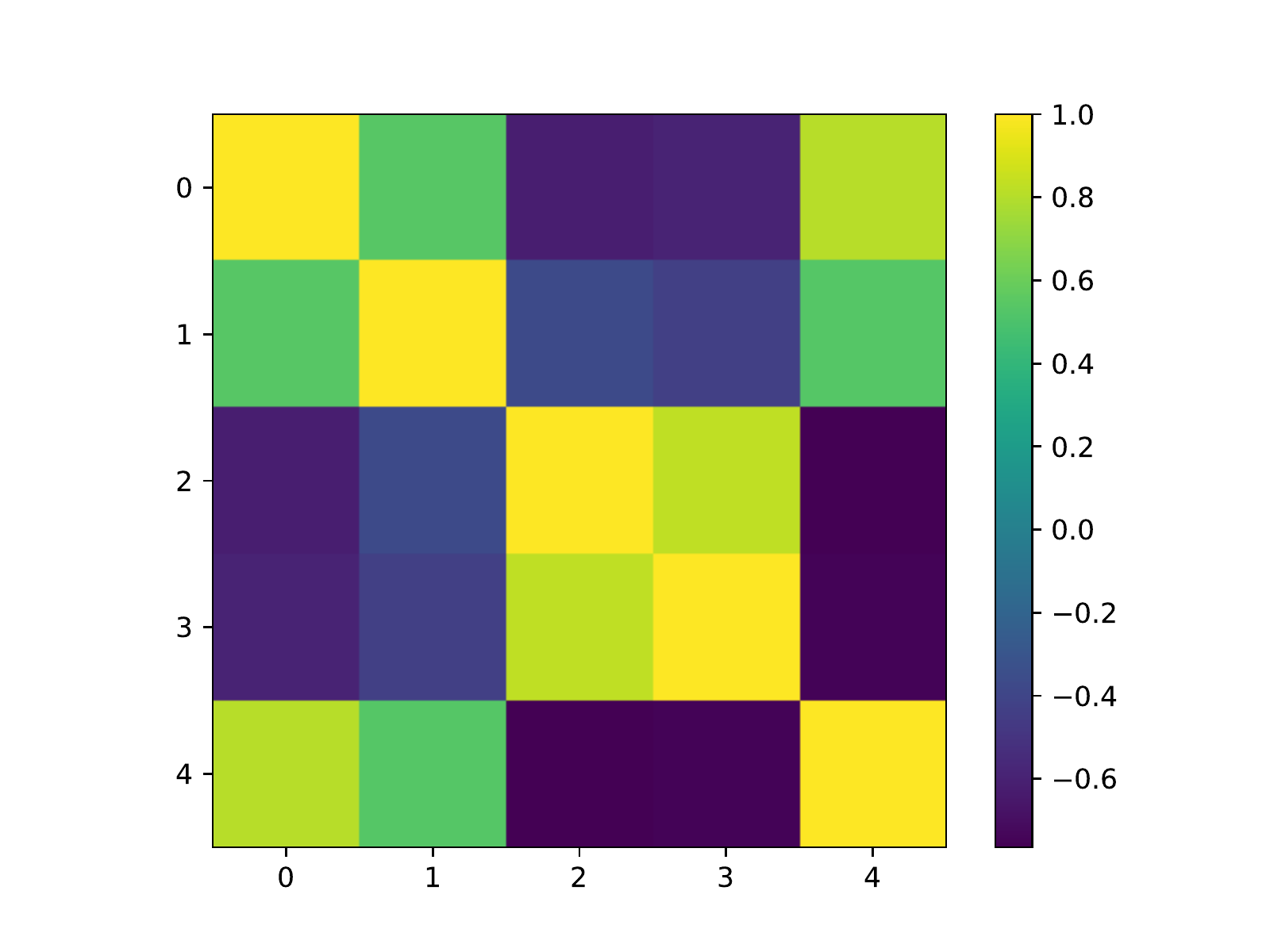} }}%
    \caption{\small Two types of correlations between the covariates in the true active sets, Prostate cancer dataset. }%
    \label{fig:corr}%
\end{figure}

\begin{figure}[ht]

	\centering{
		{{\includegraphics[width=0.47\textwidth]{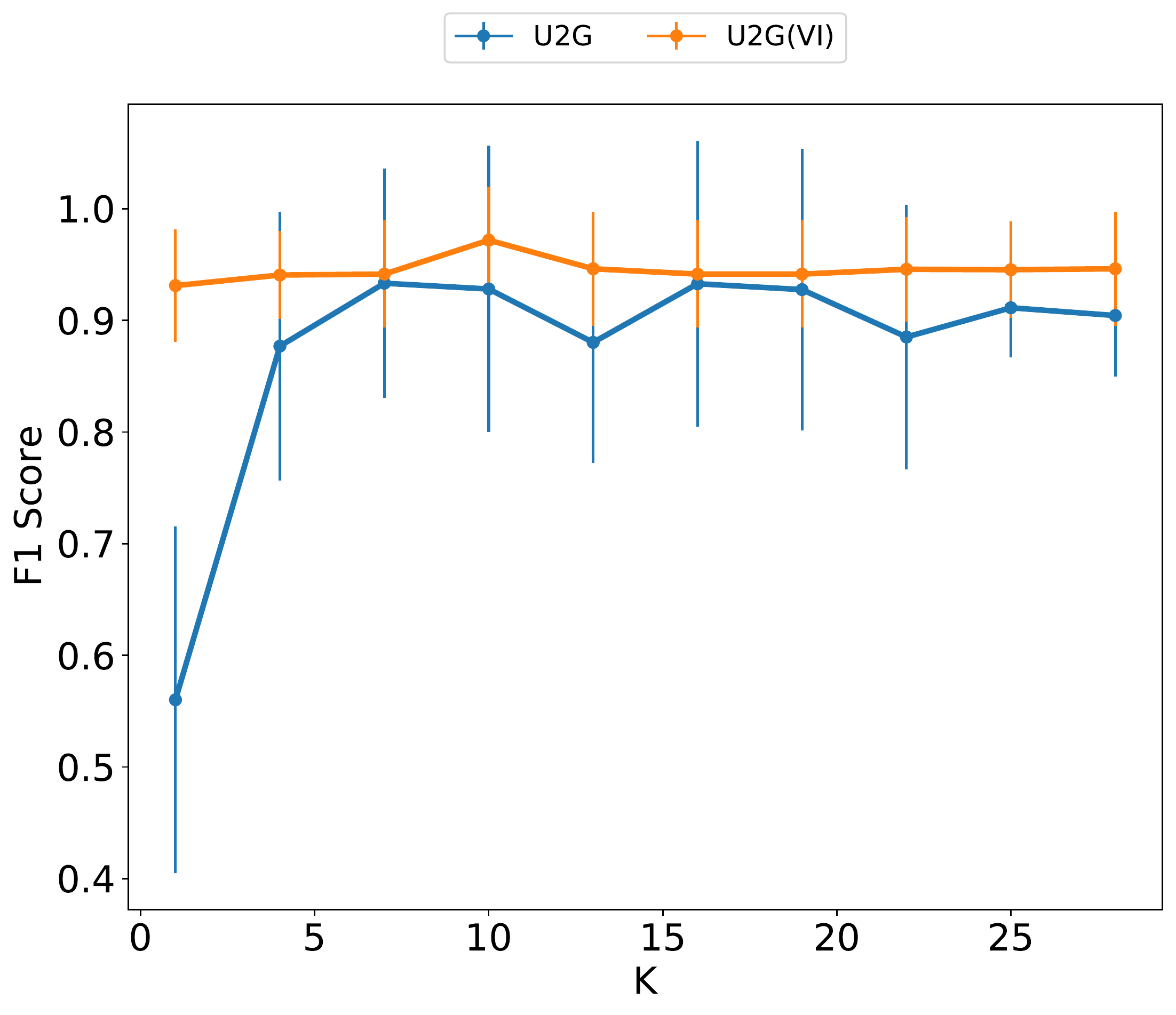} }} ~
		{{\includegraphics[width=0.47\textwidth]{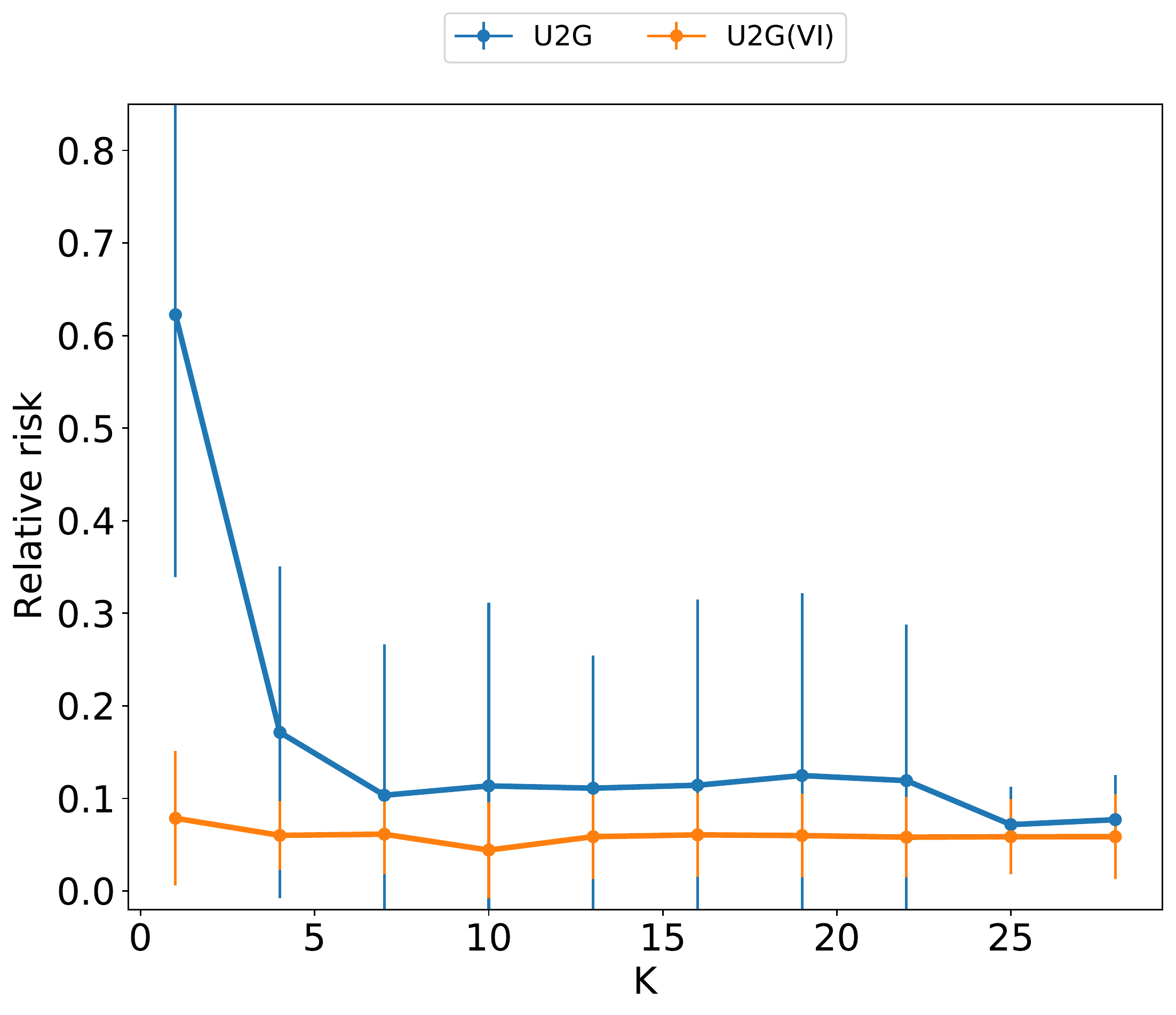} }} 
		\caption{  \small The F1 score and RR of U2G and U2G-VI over different number of samples K in estimating the gradient. U2G performance increases with K when K is small and U2G-VI has similar performance across different K. }
		\label{fig:simu2-K}
	}
\end{figure}

\begin{table}
\centering
\caption{ \small Results of the variable selection simulation study, with $n=100$, $p=1000$, $S=10$, $\text{SNR}=5$. Reported results are the mean of 100 independent trials.}  \label{tab:exp2-appendix} 
\begin{tabular}{cccccccc}
\toprule
&Precision&Recall& F1&Nonzero& RR& RTE&PVE \\
 \midrule

 \textsc{Lasso} &    0.162 &  0.995&  0.277 & 63.70  & 0.259 &  2.299 &  0.616 \\
   
 SCAD  &    0.272  &  1.000  &  0.422 &  39.10  &  0.050  &  1.245  &  0.792  \\

 MIO &0.674&0.674&0.674&10.00&0.444&3.220&0.463 \\
 
 Fast-BSS & 0.945  &  0.940  &  0.942  &  9.90 &  0.109  &  1.547  &  0.742\\ 
 
  REINFORCE & 0.011 & 0.497 & 0.021 & 471.6 & 1.085 & 6.426 & - \\ 

 ARM$_0$ & 0.906 & 0.906 & 0.904 & 10.01 & 0.154 & 1.771 & 0.705 \\

 U2G & 0.896 &  0.980   &      0.934  &  11.05  &       0.078  &  1.391 &  0.768 \\

 ARM$_0$(VI) & 0.954 & 0.967 & 0.960 & 10.13 & 0.075 & 1.374 & 0.771 \\

 U2G(VI)&  0.923 &  1.000     &      0.950 &  10.90  &        0.050 &  1.256 &  0.791 \\
\bottomrule
\end{tabular}
\end{table}

\begin{table}[ht]
\centering
\caption{\small Results of the Prostate cancer dataset, with $n=102, p=1000, S = 5$, \emph{high} multi-collinearity. Reported results are the average of 100 independent trials.}
\begin{tabular}{cccccccc}
\toprule
&Precision&Recall & F1&Nonzero& RR& RTE&PVE \\
 \midrule
 \textsc{LASSO} & 0.091 & 0.712 & 0.160 & 44.9 & 0.090 & 1.454 & 0.759 \\  

  SCAD & 0.270 & 0.368 & 0.306 & 7.04 & 0.427 & 3.144 & 0.478 \\

Fast-BSS & 0.336 & 0.302 & 0.312 & 4.58 & 0.194 & 1.982 & 0.672 \\
  U2G & 0.852 & 0.796 & 0.818 & 4.73 & 0.058 & 1.292 & 0.785 \\  

  U2G(VI) & 0.936 & 0.812 & 0.863 & 4.34 & 0.057 & 1.286 & 0.786 \\ 
\bottomrule
\end{tabular}
\label{tab:prostate2}
\end{table}

For the compressive sensing, as shown in Figure~\ref{fig:cs_high} and Table~\ref{tab:cs2}, when $\text{SNR}$ is high, all methods can achieve low prediction error but BP and BCS cannot set the coefficients of the non-signal covariates as exact zero. As shown in Figure~\ref{fig:cs_low}, when $\text{SNR}$ is low, methods with $L_1$ penalty cannot recover strong sparsity, but our method with $L_0$ penalty can.

\begin{figure}[ht]
\centering
 \includegraphics[width=0.8\textwidth]{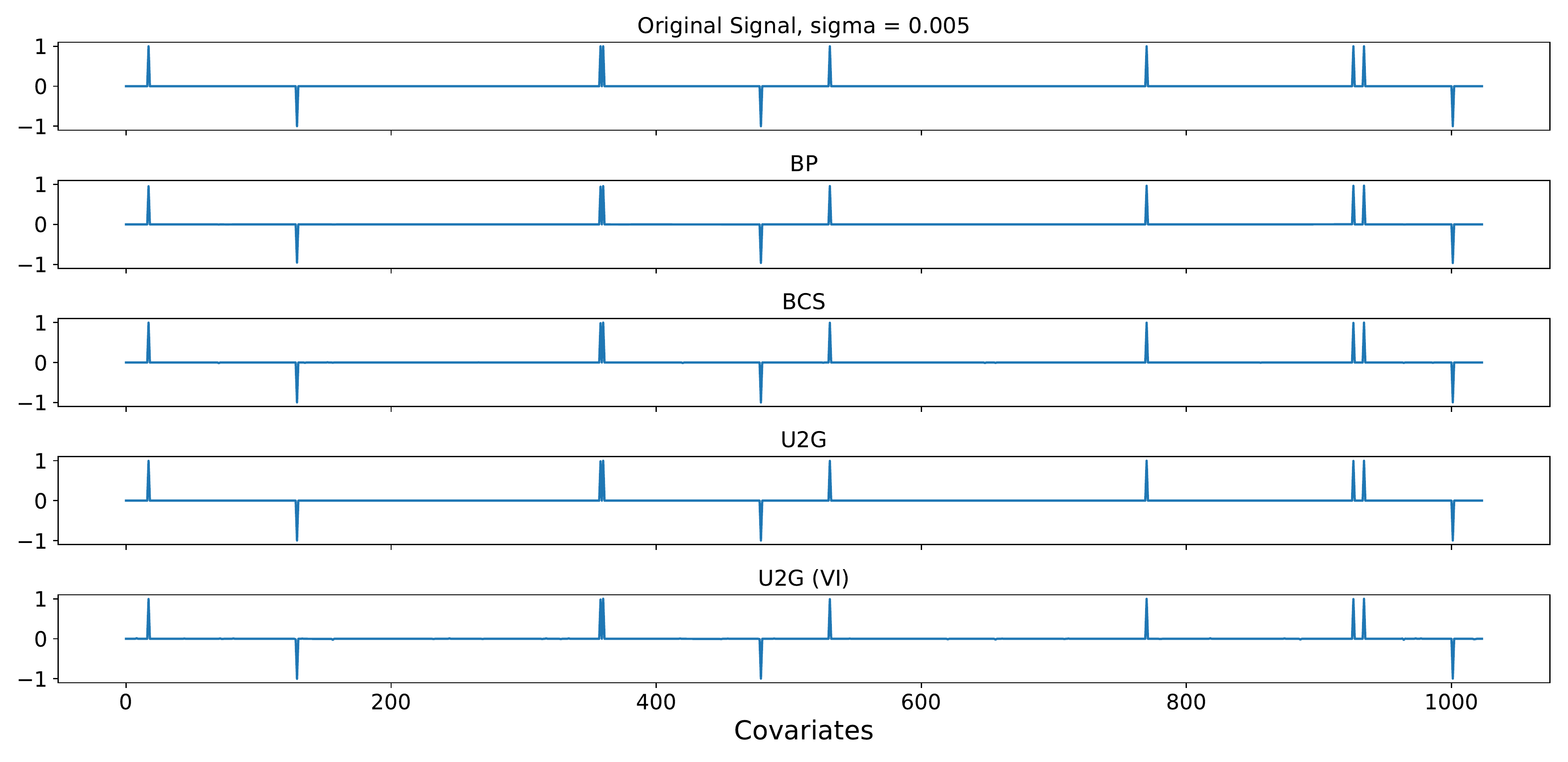}
\caption{\small Results of compressive sensing when SNR$_d$ is high, with $n = 500$, $p = 1024$, $\sigma = 0.005$.}
\label{fig:cs_high}
\end{figure}

\begin{table}[ht]
\centering
\caption{\small Results of  compressive sensing when the $\text{SNR}$ is high, with $n = 500$, $p = 1024$, $\sigma = 0.005$. }
\begin{tabular}{cccccccc}
\toprule
&Precision&Recall & F1&Nonzero& RR& RTE&PVE \\
 \midrule
BP &0.009&1.000&0.019&1024&1e-3&508.0&0.9987\\
BCS &0.322&1.000&0.488&31&1e-4&55.4&0.9998\\
U2G &1.000&1.000&1.000&10&2e-5&9.2&1.0000\\
U2G(VI) & 1.000 & 1.000 & 1.000 & 10 & 3e-5 & 13.7 & 1.0000\\
\bottomrule
\end{tabular}
\label{tab:cs2}
\end{table}

\end{document}